%% file: main.tex
\documentclass[
superscriptaddress,
nofootinbib,longbibliography,
onecolumn,
aps]
{revtex4-2}

\input{packages}
\input{aliases}

\begin{document}

\title{Mutual Information Across Quantum Reference Frame Perspectives}
\author{Leonardo Valente}
    \affiliation{Institute for Theoretical Physics, ETH Z{\"u}rich, 8093 Z{\"u}rich, Switzerland}
    \affiliation{Department of Theoretical Physics, University of Innsbruck, 6020 Innsbruck, Austria}
\author{Carla Ferradini}
   \affiliation{Institute for Theoretical Physics, ETH Z{\"u}rich, 8093 Z{\"u}rich, Switzerland}
\author{Flaminia Giacomini}
	\affiliation{Institute for Theoretical Physics, ETH Z{\"u}rich, 8093 Z{\"u}rich, Switzerland}
    \affiliation{Dipartimento di Fisica, Universit{\`a} di Roma Tor Vergata and Sezione INFN Roma2, Via della Ricerca Scientifica 1, 00133, Roma, Italy}
\author{Viktoria Kabel}
    \affiliation{Institute for Theoretical Physics, ETH Z{\"u}rich, 8093 Z{\"u}rich, Switzerland}

\begin{abstract}
   In quantum information theory, standard notions typically presuppose a fixed partition of the world into subsystems, encoded in a tensor product structure of the Hilbert space. Yet, in some scenarios, the subsystem structure is not uniquely fixed.
    A case in point is the description of subsystems relative to different quantum reference frames (QRFs), where the tensor product structure, although fixed in a specific QRF, changes across different QRFs. Such a \emph{relativity of subsystems} underlies many interesting features of QRFs, such as the relativity of entanglement, entropy, and coherence. It also raises a broader question: how do we compare the information across different subsystem perspectives, especially when the corresponding subsystem algebras overlap?
    Here, we propose a notion of \emph{quantum mutual information} across different QRF perspectives, which captures the correlations between subsystems relative to different frames of reference. To do so, we focus on an algebraic definition of subsystems. Studying in detail the union and overlap between subsystem algebras relative to different QRFs, we find that the usual definition of mutual information in terms of the chain law overcounts the information associated to the overlap. Based on this insight, we propose a modified chain law and show that the resulting notion of mutual information satisfies subadditivity and connects naturally to other measures found in the literature \cite{Gao_2019}. Our investigations not only provide a way to quantify the relation between different QRF perspectives from an information-theoretic point of view but also sheds light on some important features of quantum information in absence of a fixed tensor product structure.
\end{abstract}

\maketitle

\section{Introduction}

The way we partition the world into subsystems plays a crucial role in quantum information theory. Notions such as entanglement, conditional entropy, and mutual information all depend on the choice of subsystem structure. In non-relativistic quantum mechanics, the partition into subsystems is given by the tensor product structure of the joint Hilbert space. In most contexts, the latter is assumed to be fixed and given from the start. Yet, there are scenarios, in which the subsystem structure is more intricate. For example, one might be interested in comparing subsystems at different times, before and after the application of a unitary time evolution $U\left(t\right)$~\cite{Wu_2025,Glorioso2024,PRXQuantum.5.010201,Andreadakis2025}. Unless the latter is given by a tensor product $U\left(t\right) = U_A\left(t\right) \otimes U_B\left(t\right)$ on the initial subsystems $A$ and $B$, the time evolution will mix the subsystems. As a result, two subsystems at different times cannot straightforwardly be embedded into a single overarching tensor product structure. Another example arises in the context of quantum reference frames (QRFs)~\cite{bartlett2007reference, Giacomini_2017_covariance, Vanrietvelde:2018pgb, delaHamette_2020, delaHamette_2021, Castro_Ruiz_2021, Carette_2023}. There, the subsystem structure depends on the choice of reference frame \cite{Giacomini_2017_covariance, Castro_Ruiz_2021, Ahmad_2021, Hoehn_2023}. As a consequence, we cannot rely on a single tensor product structure if we want to speak about subsystems relative to different QRFs. Finally, once one moves towards the relativistic regime, subsystems are no longer characterised by the tensor-product structure in the first place but by local algebras of observables, associated to possibly overlapping regions of spacetime.\\

This raises a crucial question: to what extent can we generalise standard quantum information measures to regimes in which no unique tensor product structure is available? This question has been addressed to some extent by the algebraic approach to quantum theory, in which a subsystem is characterised by an algebra of (experimentally accessible) observables \cite{Ohya_1993, Zanardi_2004}. While it has long been established how to associate a measure of entropy to each of these generalised subsystems, this does not capture the \emph{correlations} between two different subsystems. In presence of a fixed subsystem structure, the latter can be quantified using the \emph{mutual information}, defined either through the relative entropy $D(\rho,\sigma)$, capturing the difference between two states $\rho$ and $\sigma$, or through the so-called chain rule: 
\begin{equation}
    I\left(A:B\right)_\rho = D\left(\rho_{AB},\rho_A\otimes \rho_B\right) =  H\left(A\right)_\rho + H\left(B\right)_\rho - H\left(AB\right)_\rho.
\end{equation}
Either way, the definition heavily relies on a fixed tensor-product structure and cannot be immediately generalised in absence of it, as it is the case when considering subsystems at different times, associated to overlapping spacetime regions, or relative to different quantum frames of reference.\\

Finding a measure of mutual information across different QRF perspectives would allow, for example, to give a quantitative answer to the question of how much information two observers gain about the overall system by combining their information (cf.~\cite{delaHamette_2025_charges, Doat_2025}). 
While much work has gone into studying the frame-dependence of different properties, such as superposition, entanglement, or the entropy associated to a particular subsystem, questions regarding the \emph{interplay} between different perspectives have so far only been addressed qualitatively. 
As we will see, QRFs come with a rich internal structure, that makes the problem of extending the notion of mutual information beyond a single tensor-product structure more tractable. In this work, we leverage this structure to propose a notion of quantum mutual information across ideal QRF perspectives for arbitrary finite-dimensional symmetry groups. We study both the union and the overlap between subsystems relative to different QRFs and show that a naive extension of the usual definition of mutual information to algebraic subsystems overcounts the information associated to the overlap. Based on this insight, we propose a modification to the chain rule for mutual information. We show that this notion satisfies subadditivity and discuss the relation to similar information measures studied in the context of generalised uncertainty relations \cite{Gao_2019, LaRacuente_2019}. Moreover, we use the same reasoning to come up with a notion of \emph{conditional information} and \emph{conditional mutual information} for algebraic subsystems and prove the strong subadditivity of the latter. This understanding of information across different QRF perspectives also leads to a natural interpretation of existing conserved quantities \cite{cepollaro2024,delaHamette_entropy_2026}, as the information associated to the overlap of the subsystem algebras relative to different reference frames, and can also give rise to additional conservation laws. Finally, to gain a clearer physical intuition for the mutual information, we consider several examples, including \emph{non-entangling states}, on which the QRF transformation acts as a tensor product of unitaries. These examples illustrate that the information accessible when combining two perspectives is strictly larger than the sum of its parts -- with the excess characterised by the mutual information.\\

The paper is structured as follows. \cref{sec:preliminaries} provides the necessary background on algebraic subsystems and subsystem relativity under changes of QRF, as well as information-theoretic tools such as conditional expectations, which play an important role in associating entropy measures to algebraic subsystems, and entropic quantities. In \cref{sec:algebraic_structure_QRFs}, we derive in detail the algebraic structure of the various subsystems relative to different QRF perspectives, first for the simple example of a $\mathbb{Z}_2$-symmetry and then for arbitrary finite-dimensional symmetry groups. This forms the basis for the definition of mutual information across QRF perspectives in \cref{sec:mutual_information_QRFs} through a modification of the chain rule. In this section, we also discuss the relation to the measure for algebraic mutual information proposed in \cite{Gao_2019} and extend our reasoning to yield extended notions of conditional information and conditional mutual information. Finally, in \cref{sec:examples_app}, we consider some concrete examples to illustrate the correlations between subsystems relative to different QRFs and show how recently derived conservation laws \cite{cepollaro2024, delaHamette_entropy_2026,Luppi_2026} can be understood more intuitively in light of the present work. We conclude in \cref{sec:conclusion} with a brief summary of the results and an outlook on future work.

\section{Preliminaries}
\label{sec:preliminaries}

In this section, we review several results that allow us to deal with quantum subsystems in absence of a pre-defined tensor product structure. In~\cref{sec:algebraic_subsystems}, we summarise the main argument of Zanardi et al.~\cite{Zanardi_2004} that the tensor-product and, therefore, the subsystem structure is induced by algebras of accessible observables and illustrate how this manifests in the relativity of subsystems under changes of QRF~\cite{Giacomini_2017_covariance,Castro_Ruiz_2021,Ahmad_2021,Hoehn_2023}. This is followed by a brief overview of how to associate entropies and information to such algebraic subsystems via conditional expectations \cite{Ohya_1993} in~\cref{sec:algebraic_entropy}. Moreover, we introduce the usual definition of mutual information and discuss a few subtleties that arise when trying to extend it to general subsystems in~\cref{sec:mutual_information}. Together, these results form the basis for defining a notion of mutual information across different QRF perspectives in the next section.\\

\textbf{Notation.} Given a Hilbert space $\cH$, we denote the set of bounded linear operators as $\cB(\cH)$ and the set of density operators (positive and trace one) as $\mathcal{S}(\cH)$. As motivated in the next section, a subsystem in $\cH$ is a unital $*$-subalgebra of $\cB(\cH)$. We label subsystems with letters $A$, $B$, $C$ and the corresponding algebras with $\alg{A}$, $\alg{B}$, $\alg{C}$. Given a set of operators $\{O_1, O_2, \dots, O_n\}$, we use $\braket{O_1, O_2, \dots, O_n}$ to denote the algebra they generate (closed with respect to linear combinations, products and adjoints). When we use $A$ as a quantum reference frame, we add to spaces, algebras, operators, states, states, etc., in such description the superscript $^{(A)}$ and similarly when we use $C$ as a QRF. For instance, the algebra of the subsystem $B$ in the perspective of $C$ is $\algpov{B}{C}$. Given an algebra $\alg{A}\subseteq \cB(\pov{\cH_{BC}}{A})$ and a map $\mathcal{E}:\cB(\pov{\cH_{BC}}{A})\mapsto\cB(\pov{\cH_{BC}}{A})$, we denote their pull-back (see~\cref{def:pullback_alg,def:pullback_map}) under the QRF transformation from $A$ to $C$ as, respectively, $\alg{A}_\triangleleft\subseteq \cB(\pov{\cH_{AB}}{C})$ and $\mathcal{E}^{\triangleleft}:\cB(\pov{\cH_{AB}}{C})\mapsto\cB(\pov{\cH_{AB}}{C})$.

\subsection{Algebraic Subsystems and Subsystem Relativity Under Changes of QRF}
\label{sec:algebraic_subsystems}

In standard quantum theory, subsystems arise when the Hilbert space has a tensor-product structure. In fact, it is commonly stated as one of the \emph{postulates} of quantum theory that \enquote{[t]he state space of a composite physical system is the tensor product of the component physical systems} \cite{Nielsen_Chuang_2010}. Yet, this definition presupposes an absolute, pre-existing factorisation of the Hilbert space, and thus, subsystem structure. A more operational approach to subsystems is obtained by moving from a view in which the Hilbert space and the quantum states are primary, to one, which focuses on observables and the algebras they form. This ``algebraic approach'' finds applications in quantum information \cite{Ohya_1993, Viola_2001} as well as quantum field theory \cite{Halvorson_2006, Fewster_2020}. The connection to a more operational understanding in terms of experimentally accessible observables was made by Zanardi et al.~\cite{Zanardi_2001, Zanardi_2004}. They show that, given a finite-dimensional Hilbert space $\cH$, a collection $\left\{\cA_i\right\}_{i=1}^n$ of subalgebras of $\cB(\cH)$ corresponding to sets of controllable observables, such as measurements or Hamiltonians with tunable parameters, and satisfying
\begin{enumerate}
    \item[(i)] subsystem independence: $[\cA_i,\cA_j]=0 \; \forall i\neq j$
    \item[(ii)] completeness: $\vee_{i=1}^n \cA_i\simeq \otimes_{i=1}^n \cA_i \simeq \cB(\cH)$
\end{enumerate}
induces a tensor-product structure (TPS) $\cH = \otimes_{i=1}^n\cH_i$ on the Hilbert space.\footnote{Note that, beyond the established observation that, in the presence of symmetries, different choices of QRF induce equally valid TPS, we are not making any assumptions about whether or not the \emph{dynamics} of a systems singles out a \emph{preferred} TPS (cf. \cite{Cotler_2017, Stoica_2021, Carral_2023, Stoica_2024, Soulas_2025, Stoica_2026}).}  This allows us to identify subsystems through the observables that an experimentalist can independently perform on said subsystems. In other words, what constitutes different subsystems is determined through our operational abilities. This algebraic understanding of subsystems has since been applied in the quantum foundations literature: it forms the basis for studying partitions in quantum theory \cite{Vanrietvelde_2025}, can be derived from causality considerations \cite{Ormrod_2025}, plays an important role in understanding local subsystems in gauge theory and gravity \cite{Donnelly_2012, Casini_2014, Soni_2016,VanAcoleyen_2016,Bianchi_2024,Boulle_2025, Spalvieri_2025}, and allows to highlight the relational nature of thermodynamics \cite{Jha_2026}. In what follows, we call algebraic subsystem a unital $*$-subalgebra $\cA$ of $\cB(\cH)$. This is a generalisation of the conditions (i) and (ii) of Ref.\,\cite{Zanardi_2001, Zanardi_2004} which extends the concept of a subsystem beyond the tensor product structure.

The lack of an \emph{a priori} fixed tensor product structure also underlies the \emph{relativity of subsystems} under changes of QRF \cite{Giacomini_2017_covariance, Vanrietvelde:2018pgb, Ahmad_2021, Hoehn_2021, Castro_Ruiz_2021, Hoehn_2023}: the fact that each QRF perspective is endowed with its own Hilbert space decomposition, but the tensor product factorisation changes when changing the QRF. This feature can be traced back to the fact that subsystems relative to different QRF perspectives are described by different relational algebras. In line with the operational understanding discussed above, these can be understood as the algebras of observables accessible from within a given QRF perspective. While different QRF approaches diverge on the precise set of observables that can be accessed from within a given QRF perspective \cite{delaHamette_2025_charges, Doat_2025}, they all agree on the relationality of these algebras. Here, we will focus on the perspectival approach to QRFs~\cite{Giacomini_2017_covariance} and the closely related perspective-neutral approach~\cite{Vanrietvelde:2018pgb}. However, since we follow an algebraic approach, we expect the general methodology to carry over to the other approaches to QRF, even though the specific subalgebras of observables will differ.

We illustrate the situation with a toy model: three systems with $\mathbb{Z}_2$-symmetry \cite{delaHamette_2020,cepollaro2024, Castro-Ruiz_2025}. In this case, the algebra of observables on a single system is spanned by the Pauli-$Z$ and Pauli-$X$ operators. The former encodes the orientation of the system, the latter can be thought of as its $\mathbb{Z}_2$-momentum (or conjugate variable), which is invariant under the action of the symmetry group. If we look at the system from the outside, the symmetry-invariant properties are given by the relative orientations $Z_A\otimes Z_B$, $Z_A\otimes Z_C, Z_B\otimes Z_C$ as well as the momenta $X_A,X_B, X_C$ (see~\Cref{fig:subsystem_relativity_Z2}). 
\begin{figure}[h!]
    \centering
    \includegraphics[width=0.4\linewidth]{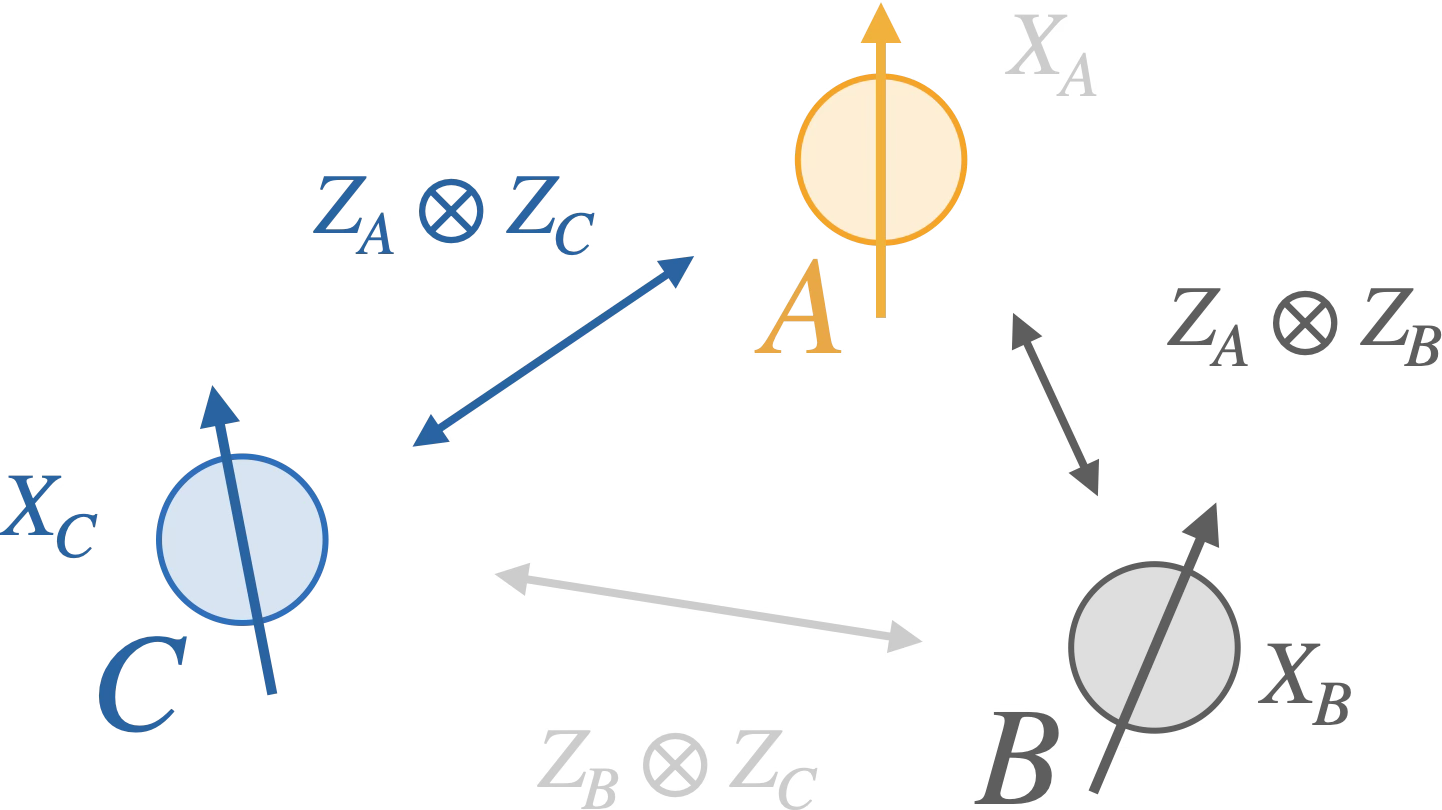}
    \caption{Invariant quantities for a three-qubit system with $\mathbb{Z}_2$-symmetry. The choice of a QRF comes with a preferred partition into subsystems. For example, relative to $A$, the observables on system $B$ are spanned by $Z_A\otimes Z_B, X_B$ and those on $C$ by $Z_A\otimes Z_C, X_C$.}
    \label{fig:subsystem_relativity_Z2}
\end{figure}
In the perspective-neutral and the perspectival approach, one further assumes that the total charge, in this case $X_A\otimes X_B\otimes X_C$, is fixed. Thus, the total algebra of observables of interest can be generated from a subset of only four of the above operators.\footnote{In the extra-particle approach, on the other hand, one includes the total charge in the algebra of observables relative to a chosen QRF.} The choice of QRF can be seen to pick out particular subalgebras of observables. For example, relative to $A$, it is natural to describe system $C$ by its orientation relative to $A$, that is $Z_A\otimes Z_C$, and the canonically conjugate momentum $X_C$. Similarly, system $B$ is characterised by $Z_A\otimes Z_B$ and $X_B$. The formal derivation of these statements is summarised in \cref{app:subsystem_relativity}.

These operators span the subsystem algebras
\begin{equation}
    \algpovemb{B}{A} = \langle X_B, Z_A\otimes Z_B\rangle, \;\; \algpovemb{C}{A} = \langle X_C, Z_A\otimes Z_C\rangle,
\end{equation}
where the subscript ``$\mathrm{emb}$'' signals that these observables are expressed relative to an outside perspective. It is not difficult to see that these algebras mutually commute:
\begin{equation}
\begin{aligned}
    [X_B, X_C] &= 0,\;\; [X_B, Z_A\otimes Z_C]=0,\\
    [Z_A\otimes Z_B, X_C] &= 0,\;\; [Z_A\otimes Z_B, Z_A\otimes Z_C]=0.
\end{aligned}
\end{equation}
They also span the total algebra of observables in the perspective-neutral approach -- the relative orientations can be combined to yield $Z_B\otimes Z_C$ and the momentum of $A$ is not an independent degree of freedom due to the constraint. They thus satisfy the conditions (i) and (ii) required to induce a TPS on the corresponding Hilbert space, called the \emph{physical} Hilbert space \cite{Ahmad_2021}:
\begin{equation}
    \cH_{\textup{phys}} \simeq \pov{\cH_B}{A}\otimes \pov{\cH_C}{A} = \pov{\cH_{BC}}{A}.
\end{equation}
With respect to this TPS, the observables in $\algpov{B}{A}$ act trivially on the space $\pov{\cH_C}{A}$ and vice versa. This captures the partition into subsystems as seen from the perspective of QRF $A$. The same reasoning can be repeated in the perspective of $C$ with
\begin{equation}
    \algpovemb{A}{C} = \langle X_A, Z_A\otimes Z_C\rangle, \;\; \algpovemb{B}{C} = \langle X_B, Z_B\otimes Z_C\rangle
\end{equation}
inducing a TPS on the Hilbert space 
\begin{equation}
    \cH_{\textup{phys}}\simeq \pov{\cH_A}{C}\otimes \pov{\cH_B}{C} = \pov{\cH_{AB}}{C}.
\end{equation}
Importantly, while related by a unitary transformation (the QRF transformation), these two Hilbert spaces carry different TPS and thus come with different partitions into subsystems. This subsystem relativity under changes of QRF has been studied in detail in \cite{Giacomini_2017_covariance, Ahmad_2021, Hoehn_2021, Castro_Ruiz_2021, delaHamette_2021, Hoehn_2023} and extends to arbitrary locally compact symmetry groups $G$~\cite{delaHamette_2021,Carette_2023,Fewster_2024,DeVuyst_2024b}. It also underlies many of the interesting features of QRF transformations, such as the so-called \enquote{paradox of the third particle} \cite{Angelo_2011,Angelo_2012,Krumm_2021,Hoehn_2021, Castro_Ruiz_2021, Hausmann_2026, Brukner_2026, Palumbo_2026} and the frame-dependence of many quantities like entanglement, entropy, and coherence~\cite{Giacomini_2017_covariance,cepollaro2024,Hoehn_2023,deVuyst_2024a}. All of these quantities, however, are derived separately within each QRF perspective. Within each perspective, the TPS is fixed and one can use the common definitions of quantum information measures. Here, we want to study the mutual information \emph{across} the different perspectives. Since the subsystems relative to different QRFs are not defined within the same TPS, we need to go beyond this regime. Given that the subsystems are instead associated with concrete subalgebras of observables, the natural language is a more algebraic approach to quantum information notions, which is commonly used in the context of (algebraic) quantum field theory~\cite{Haag_1996,Witten_2018,Hollands_2018,Araki_1976, Casini_2014} but can also be applied to non-relativistic quantum information \cite{Ohya_1993,Zanardi_2001,Viola_2001,Barnum_2004,bartlett2007reference,Balachandran_2013,Kabernik_2020,Chiribella_2018,Donnelly_2012}.

\subsection{Conditional Expectation}
\label{sec:algebraic_entropy}

Given a subsystem with respect to a fixed TPS, the reduced state of said system is obtained via the partial trace. When dealing with \emph{algebraic subsystems}, defined by more general algebras of observables, the equivalent operation to taking the partial trace is given by the conditional expectation.

\begin{proposition}[Conditional expectation \cite{Ohya_1993}]\label{def:cond_expect1}
  Let there be an algebra of observables $\mathcal{A} \subset\mathcal{M}$. There exists a unique linear map $\cexp:\mathcal{M}  \mapsto \mathcal{A}$  such that
\begin{enumerate}
\vspace{-0.5\baselineskip}    \renewcommand\labelitemi{–} \setlength\itemsep{0.1em}    \setlength\parskip{0pt}
    \item If $O\in \mathcal{M}$ is positive then $E\left(O\right) \in \mathcal{A}$ is also positive.\label{condition_1}
    \item $\cexp\left(A\right)=A$ for every $A \in \mathcal{A}$.\label{condition_2}
    \item $\cexp\left(OA\right)=\cexp\left(O\right)A$ holds for any $O \in \mathcal{M}$ and $A \in \mathcal{A}$.\label{condition_3}
    \item $\Tr[\cexp\left(O\right)]=\Tr[O]$ holds for any $O \in \mathcal{M}$.\label{condition_4}
\end{enumerate}
\end{proposition} 
\begin{proof}
    See Proposition 1.12 of \cite{Ohya_1993}.
\end{proof}
The conditional expectation provides the unique projection onto the subalgebra of interest that preserves the trace and thus maps quantum states to quantum states. To gain some intuition, it is helpful to consider what the conditional expectation looks like for a typical tensor-product subsystem.

\begin{example}[Tensor-product subsystem]
\label{ex:TP_Subsystem}
Consider a bipartite quantum system with Hilbert space $\cH = \cH_A\otimes \cH_B$ and algebra $\mathcal{M}= \cB\left(\cH\right)$. Define the subsystem algebra associated to $A$ as $\alg{A} = \left\{O_A\otimes \mathbb{I}_B|O_A\in \cB\left(\cH_A\right)\right\}$. The conditional expectation onto this subalgebra is given by
\begin{equation}
\begin{aligned}
    \cexparg{A}: \mathcal{M} &\mapsto \alg{A}\\
    O &\mapsto \Tr_B\left(O\right) \otimes \tfrac{\id_B}{d_B}
\end{aligned}
\end{equation}
where $d_B$ is the dimension of $\cH_B$. Indeed, one can easily see that this satisfies the properties \ref{condition_1} to \ref{condition_4} above.
\end{example}

However, the conditional expectation allows one to project onto more general subalgebras of observables, such as those invariant under a the action of a symmetry group.

\begin{example}[Invariant subsystem]
    Let $\mathcal{M}$ be an algebra of observables on Hilbert space $\cH$ and $G$ a finite-dimensional symmetry group with unitary representation $U: G \mapsto \cB(\cH)$. Consider the algebra of invariant observables
    \begin{equation}
        \cA_{\textup{inv}} = \left\{O\in \mathcal{M}|[U,O]=0\right\}.
    \end{equation}
    The conditional expectation onto this algebra is given by the \emph{incoherent G-twirl}
    \begin{equation}
    \begin{aligned}
        \cexp:\mathcal{M} &\mapsto \cA_{\textup{inv}}\\
            O&\mapsto \mathcal{G}_{\textup{inc}}\left(O\right)\coloneqq\frac{1}{|G|}\sum_{g\in G}  U\left(g\right) O U^\dagger \left(g\right)
    \end{aligned}
    \end{equation}
    It is straightforward to see that it satisfies the properties of the conditional expectation: 
    \begin{enumerate}
        \item Conjugation by a unitary preserves positivity and a convex combination of positive operators is positive.
        \item We have:
        \begin{equation}
            \begin{split}
                E\left(E\left(O\right)\right) &= \frac{1}{|G|^2}\sum_{g,h\in G} U\left(h\right)U\left(g\right) O U^\dagger\left(g\right) U^\dagger\left(h\right) \\
                &= \left(\frac{1}{|G|}\sum_{h\in G}1\right)\frac{1}{|G|}\sum_{g'\in G}  U\left(g'\right) O U^\dagger\left(g'\right) = E\left(O\right),
            \end{split}
        \end{equation}
        since $U(h)U(g) = U(hg)=:U(g')$.
        \item For any $O\in \mathcal{M}$, $A\in \mathcal{A}_{\mathrm{inv}}$:
        \begin{equation}
            \begin{split}
                E\left(OA\right) = \frac{1}{|G|}\sum_{g\in G}U\left(g\right) O A U^\dagger \left(g\right) = \frac{1}{|G|}\sum_{g\in G}U\left(g\right) O U^\dagger \left(g\right) A = E\left(O\right) A                
            \end{split}
        \end{equation}
        since by definition $A$ is invariant under the action of the symmetry group.
        \item By linearity and cyclicity of the trace $\Tr[E\left(O\right)] = \frac{1}{|G|}\sum_{g\in G}\Tr[U\left(g\right) O U^\dagger\left(g\right)] = \Tr[O]$.
    \end{enumerate}
\end{example}

\noindent Finally, let us give one more example, which will play an important role throughout the rest of the paper: the pinched subsystem.

\begin{example}[Pinched subsystem (cf.~\cite{Bertoni_2020})]
\label{ex:pinched_subsystem}
     Let $\mathcal{M}=\cB\left(\mathcal{H}\right)$ for a Hilbert space $\cH$ of dimension $d$. Let $X \in \cB\left(\mathcal{H}\right)$ be an observable with spectral decomposition $X = \sum_i \lambda_i P_i$, where $P_i$ is the projector onto the eigenspace with eigenvalue $\lambda_i$. Then, the pinching map
     \begin{equation}
    \begin{aligned}
        \mathcal{P}_{X}: \mathcal{M} &\mapsto \cA\\
        O &\mapsto \sum_i P_i O P_i
    \end{aligned}
    \end{equation}
    is the conditional expectation onto the subalgebra $\cA \subset \mathcal{M}$ that is block-diagonal with respect to the eigenspaces of $X$. To see that it is a conditional expectation, note that pinching is a projective CPTP map, implying conditions \ref{condition_1}, \ref{condition_2}, and \ref{condition_4}, while condition \ref{condition_3} can easily be checked explicitly, $E\left(OA\right) = \sum_i P_i OA P_i =\sum_i P_i O P_i A = E\left(O\right)A$ for $A\in \mathcal{A}$ since $[A,P_i]=0$.
\end{example}

Let us close this subsection by noting a convenient property of the relative entropy under the restriction onto a subalgebra of observables. Recall that the relative entropy is defined as follows.

\begin{definition}[Relative entropy]
\label{def:rel_entropy}
    The relative entropy between two states $\rho,\sigma \in \mathcal{S}\left(\cH\right)$ is 
    \begin{equation}
        D\left(\rho,\sigma\right) = 
        \begin{cases} \Tr\left[\rho \left(\log \rho - \log \sigma\right)\right] &\text{ if } \mathrm{supp}\left(\rho\right) \subseteq \mathrm{supp}\left(\sigma\right)\\ \infty &\text{ otherwise} \end{cases}
    \end{equation}
\end{definition}

Given a subalgebra $\mathcal{A}\subset \mathcal{M}$ and a state $\sigma$ that is invariant under the application of the corresponding conditional expectation $E$, we can decompose the relative entropy into two parts: the relative entropy of $\rho$ compared to its reduced version $E\left(\rho\right)$ and the relative entropy between $E\left(\rho\right)$ and $\sigma$.

\begin{theorem}[Decomposition of relative entropy \cite{Ohya_1993}]
\label{th:dec_rel_entropy}
    Let $\alg{A} \subset \mathcal{M}$ be a subalgebra with conditional expectation $\cexp: \mathcal{M}\mapsto \cA$ and $\sigma$ an invertible state of $\mathcal{M}$ such that $\cexp\left(\sigma\right)= \sigma$. Then,
    \begin{equation}
        D\left(\rho,\sigma\right) =  D\left(\rho, \cexp\left(\rho\right)\right) +D\left(\cexp\left(\rho\right),\sigma\right)
    \end{equation}
    for all states $\rho \in \mathcal{M}$.\label{th:decomposition_rel_entropy}
\end{theorem}
\begin{proof}
    See Theorem 1.13 of \cite{Ohya_1993}.
\end{proof}

\subsection{Mutual Information}
\label{sec:mutual_information}

For tensor-product subsystems, the conditional mutual information is defined as follows.

\begin{definition}[Conditional mutual information \cite{Cerf_1998}]
\label{def:mutual_information}
Let $\rho_{ABC}\in\mathcal{S}\left(\cH_A\otimes \cH_B\otimes \cH_C\right)$ be a quantum state. The conditional mutual information of $A$ and $B$ given $C$ is defined as 
\begin{equation}
    \minfo\left(A:B|C\right)\coloneqq H\left(A|C\right) + H\left(B|C\right) - H\left(AB|C\right),
\end{equation}
where the conditional entropy is defined through the chain rule, e.g., $H\left(A|C\right)=H\left(AC\right)-H\left(C\right)$.
\end{definition}

If the conditioning system is trivial, the mutual information can be written as a relative entropy~(\cref{def:rel_entropy}), 
\begin{equation}
     \minfo\left(A:B\right)\coloneqq H\left(A\right) + H\left(B\right) - H\left(AB\right) = D\left(\rho_{AB}, \rho_{A}\otimes \rho_{B}\right),
\end{equation}
where $\rho_{A}\coloneqq\Tr_{B}[\rho_{AB}]$ and $\rho_{B}\coloneqq \Tr_A[\rho_{AB}]$ are the reduced states of the state $\rho_{AB}$. 
This shows that the mutual information measures the difference between having access to the whole bipartite system, characterised by the state $\rho_{AB}$ with entropy $H\left(AB\right)$, and the information available when accessing each subsystem separately, described by the state $\rho_A\otimes \rho_B$ having entropy $H\left(A\right)+H\left(B\right)$. By subadditivity of entropy, the mutual information is always positive, $\minfo\left(A:B\right)\ge0$. More generally, one can prove that the same holds for the conditional mutual information.

\begin{proposition}[Strong subadditivity~\cite{Nielsen_Chuang_2010}]
For any state $\rho_{ABC} \in \mathcal{S}\left(\cH_A\otimes \cH_B\otimes \cH_C\right)$, the mutual information satisfies $\minfo\left(A:B|C\right)\ge0$.
\end{proposition}

Beyond its role as a correlation measure, the mutual information has several operational interpretations which further justify its use in information theory. For instance, in distributed compression tasks $\minfo\left(A:B\right)$ is the optimal rate of classical communication required to merge the correlations between $A$ and $B$ using shared entanglement~\cite{Horodecki_2005,Horodecki_2006,Schumacher_2006}. The mutual information also provides bounds in correlations that can be exploited for tasks such as distributed computation~\cite{Cleve_1999}, dense coding~\cite{Bennett_1992}, and secret-key extraction~\cite{Csiszar_1978,Christandl_2004}, and it is used to mark phase transitions and long-range correlations in quantum many-body systems~\cite{Wolf_2008,Furukawa_2009,Singh_2011,Alcaraz_2013}. In relation with classical correlations, $\minfo\left(A:B\right)$ also bounds the amount of classical correlations that can be extracted between $A$ and $B$ through local measurements~\cite{Henderson_2001,Ollivier_2001}.\\

The first subtlety that one encounters when extending the mutual information to general subsystems 
is immediately apparent: in absence of a single, fixed TPS, there is no obvious equivalent of the states $\rho_{AB}$ and $\rho_A\otimes \rho_B$. This is true in particular for subsystems relative to different QRFs. \emph{Within} a given QRF perspective, we can define the mutual information as usual since the choice of QRF induces a unique TPS. 
However, the definition of mutual information \emph{across} perspectives, such as $\minfo\left(\pov{B}{A}:\pov{B}{C}\right)$ poses a deeper challenge since $\pov{B}{A}$ and $\pov{B}{C}$ are defined with respect to different tensor-product structures. In this case, there is no analogue of the state $\rho_A\otimes \rho_B$, so we cannot directly apply the definition of mutual information in terms of the relative entropy, $\minfo\left(A:B\right) = D\left(\rho_{AB}, \rho_A\otimes \rho_B\right)$. Instead, we will base our construction on the chain rule $\minfo\left(A:B\right) = H\left(A\right) + H\left(B\right) - H\left(AB\right)$. Since the conditional expectation acts on both observables and quantum states, it allows to generalise the notion of a reduced state and thus define the entropy associated to these algebraic subsystems.

\begin{definition}[Entropy of an algebraic subsystem]
    Given a subalgebra $\cA\subseteq \mathcal{M}$ with conditional expectation $\cexp: \mathcal{M}\mapsto \cA$, and a state $\rho\in \mathcal{S}\left(\cH\right)$, the von Neumann entropy associated to the algebraic subsystem $\cA$ is given by
    \begin{equation}
    H\left(\cA\right)_\rho \coloneqq -\Tr [\cexp\left(\rho\right)\log \cexp\left(\rho\right)].
    \end{equation}
\end{definition}

Note that this definition of the entropy of an algebraic subsystem does not exactly reproduce the usual definition of subsystem entropy, $H\left(A\right)_{\rho_A} = -\Tr \rho_A\log\rho_A$, when restricting to a tensor-product subsystem. The reason is that the conditional expectation maps a state $\rho\in \mathcal{S}\left(\cH_A\otimes \cH_B\right)$ to $\cexparg{A}\left(\rho\right) = \rho_A\otimes \tfrac{\mathbb{I}_B}{d_B}$, which carries additional entropy due to the maximally mixed state on $B$. In other words, we have $H\left(\cA\right)_\rho = H\left(A\right)_{\rho} +\log d_A$. This suggests the use of \emph{information}~\cite{Everett_1956, Everett_1973} instead of entropy.
\begin{definition}[Information of an algebraic subsystem]
    Consider an algebra $\mathcal{M}$ of bounded operators on a Hilbert space $\cH$ of dimension $d$. Given a subalgebra $\cA\subseteq \mathcal{M}$ with conditional expectation $\cexp: \mathcal{M}\mapsto \cA$, and a state $\rho\in \mathcal{S}\left(\cH\right)$, the information associated to the algebraic subsystem is given by
    \begin{equation}
        \info\left(\cA\right)_\rho \coloneqq D\left(\cexp\left(\rho\right),\tfrac{\mathbb{I}}{d}\right)=\log d - H\left(\cA\right)_\rho,
    \end{equation}
    where $D$ is the relative entropy.\label{def:algebraic_information}
\end{definition}

This definition indeed reduces to the usual notion of information when restricting to tensor-product subsystems: $\info\left(\cA\right)_\rho = \log d_A - H\left(A\right)_\rho$. This is because it is defined in terms of a \emph{relative entropy}, which takes into account explicitly a reference state -- the maximally mixed state of the whole system -- and thus provides a reference of maximal entropy. In the following, we will therefore focus on the information of an algebraic subsystem rather than its entropy. When unambiguous, we might omit the subscript $\rho$ and denote the entropy and information of an algebraic system $\cA$ respectively as $H(\cA)$ and $I(\cA)$.\\

We can express the chain rule for mutual information equivalently in terms of information through $\minfo\left(A:B\right) = \info\left(AB\right) - \info\left(A\right) - \info\left(B\right)$. To evaluate these quantities, we need to know the joint system $AB$ even when the subsystems $A$ and $B$ are defined relative to different perspectives. Fortunately, this is relatively straightforward given the algebraic understanding of subsystems, and will be discussed in more detail in~\cref{sec:algebraic_structure_QRFs}. There, we also derive the conditional expectations onto all subsystem algebras and their unions. Given these maps, we can then use the algebraic definitions from the previous subsection to associate information and entropy to the various subsystems. In~\cref{sec:mutual_information_QRFs}, we will then combine them into a measure of mutual information across perspectives using a slightly modified version of the chain rule. The modification is motivated by another subtlety that arises when extending mutual information to general subsystems: unlike TP-subsystems, algebraic subsystems may have non-trivial overlap.

\section{The Algebraic Structure of Subsystems Relative to Different QRFs}
\label{sec:algebraic_structure_QRFs}

The algebraic notion of subsystems allows us to leave the realm of a fixed tensor-product-structure and thereby consider subsystems relative to different QRFs within the same framework. In this section, we work out the algebraic structure of these subsystems, first for the simple example of a $\mathbb{Z}_2$-reference frame, and then for general finite-dimensional symmetry groups. 
More concretely, we will identify the corresponding subalgebras, construct their unions and intersections, and determine the conditional expectation onto each subalgebra. Note that the algebraic structure of subsystems relative to different QRFs also features in several other works \cite{Castro_Ruiz_2021, Hoehn_2023, Araujo-Regado_2025, delaHamette_2025_charges, Doat_2025, Garmier_2025_non_ideal}. Here, we build on these insights and provide a complete characterisation of all the possible subsystems including the conditional expectations connecting the various subalgebras. This will lay the foundation for finding an appropriate notion of mutual information across different perspectives in the next section.

\subsection{Warm-Up: $\mathbb{Z}_2$ Reference Frames}
\label{sec:algebraic_structure_QRFs_Z2}

Let us return to the example of three qubits $A,B,$ and $C$ with symmetry group $G=\mathbb{Z}_2$, briefly discussed in the previous section. There, we established that the choice of QRF induces a TPS on the perspectival Hilbert space. That is, the Hilbert spaces relative to $A$ and $C$ are given, respectively, by
\begin{equation}
    \pov{\cH_{BC}}{A}= \pov{\cH_B}{A}\otimes \pov{\cH_C}{A}, \qquad \pov{\cH_{AB}}{C}= \pov{\cH_A}{C}\otimes \pov{\cH_B}{C}.
\end{equation}
In the following, we will sometimes speak of observers Alice and Charlie with access to QRFs $A$ and $C$ respectively, to refer to these two perspectives. 
The map between these two Hilbert spaces is given by the unitary \emph{QRF transformation} \cite{delaHamette_2020,cepollaro2024}
\begin{equation}
    S_{A\to C} = \ket{0}_A\bra{0}_C\otimes \mathbb{I}_B + \ket{1}_A\bra{1}_C\otimes X_B,
\end{equation}
where $X_B$ is the Pauli-$X$ operator on system $B$. This acts both on states and on observables, through the unitary channel $\mathcal{S}_{A\to C}\left(O\right) \coloneqq S_{A\to C} O S_{A\to C}^\dagger$ for all $O \in \cB\left(\pov{\cH_{BC}}{A}\right)$.
\emph{Within} each Hilbert space, the split into subsystems is straightforward: we can associate each subsystem with a tensor-product-factor and the corresponding algebra of observables as in~\cref{ex:TP_Subsystem}. That is, relative to $A$, we can identify 
\begin{align}
    \algpov{B}{A} &= \left\{O\otimes \pov{\mathbb{I}_C}{A}|O\in \cB\left(\pov{\cH_B}{A}\right)\right\} = \langle \pov{X_B}{A}\otimes \pov{\mathbb{I}_C}{A},\pov{Z_B}{A}\otimes \pov{\mathbb{I}_C}{A}\rangle\subseteq \cB\left(\pov{\cH_{BC}}{A}\right),\\
    \algpov{C}{A} &= \left\{ \pov{\mathbb{I}_B}{A}\otimes O|O\in \cB\left(\pov{\cH_C}{A}\right)\right\} = \langle \pov{\mathbb{I}_B}{A}\otimes \pov{X_C}{A},\pov{\mathbb{I}_B}{A}\otimes \pov{Z_C}{A}\rangle\subseteq \cB\left(\pov{\cH_{BC}}{A}\right). 
\end{align}
Similarly, relative to $C$, we have the algebras
\begin{align}
    \algpov{A}{C} &= \langle \pov{X_A}{C}\otimes \pov{\id_B}{C},\pov{Z_A}{C}\otimes \pov{\id_B}{C}\rangle\subseteq \cB\left(\pov{\cH_{AB}}{C}\right),\label{eq:algebra_AC}\\
    \algpov{B}{C} &= \langle \pov{\id_A}{C}\otimes \pov{X_B}{C},\pov{\id_A}{C}\otimes \pov{Z_B}{C}\rangle\subseteq \cB\left(\pov{\cH_{AB}}{C}\right).\label{eq:algebra_BC}
\end{align}
Now, in order to compare subsystems across different perspectives, we need to embed both algebras into a joint Hilbert space. While we could describe everything from an external point of view, as in the previous section, we will follow a construction that works purely on the perspectival level and thus does not depend on whether such an external description is available. In the following, we will choose to describe everything within the Hilbert space relative to Charlie's perspective, but throughout the construction it will be evident that this is merely a matter of convention. Due to the relativity of subsystems under changes of QRF, the subsystems relative to $A$ will look different when viewed from Charlie's perspective. Let us take as an example the observable $\pov{Z_B}{A}\otimes \pov{\id_C}{A}$. This will be mapped to $\mathcal{S}_{A\to C}\left(\pov{Z_B}{A}\otimes \pov{\id_C}{A}\right) = \pov{Z_A}{C}\otimes \pov{Z_B}{C}$ under the QRF transformation. Thus, what looks like an observable that acts non-trivially only on subsystem $B$ from Alice's point of view, becomes an observable on both $A$ and $B$ from Charlie's perspective. This is not surprising: when Alice measures the orientation of system $B$, she actually measures the orientation of the system relative to her own QRF. In order to determine the full subsystem structure relative to $A$ as seen from Charlie's perspective, we need to apply the QRF transformation to the entire algebra of observables that defines each subsystem. We can think of this as \emph{pulling back} the algebra relative to $A$ \emph{into} the perspective of Charlie.\footnote{This is, in fact, the same strategy as the one used in~\cite{Sannino_2025}, with the goal of embedding subsystems at different times into a joint Hilbert space.} 

\begin{definition}[Pull-back of an algebra across QRFs]
\label{def:pullback_alg}
Let $\cA \subseteq \cB\left(\pov{\cH_{BC}}{A}\right)$ be a subalgebra of observables on $B$ and $C$ relative to $A$. The pull-back of $\cA$ under the QRF transformation $S_{A\to C}$ is defined as
\begin{equation}
    \cA_{\triangleleft} \coloneqq \mathcal{S}_{A\to C}[\cA] = \left\{\mathcal{S}_{A\to C}\left(O\right)|O\in \cA\right\}.
\end{equation}
Similarly, for a subalgebra $\bar{\cA} \subseteq \cB\left(\pov{\cH_{AB}}{C}\right)$ of observables on $A$ and $B$ relative to $C$, we define push forward of $\bar{\cA}$ under the transformation $S_{C\to A} = S_{A\to C}^{-1}$ as $\bar{\cA}_{\triangleright} \coloneqq \mathcal{S}_{C\to A}[\bar{\cA}] \subseteq \cB\left(\pov{\cH_{BC}}{A}\right)$.
The two operations are mutually inverse, i.e., $(\cA_{\triangleleft})_{\triangleright} = \cA$ and $(\bar{\cA}_{\triangleright})_{\triangleleft} = \bar{\cA}$.
\end{definition}

\begin{figure}[h]
    \centering
    \input{fig2}
    \caption{The subsystem algebras $\algpov{B}{A}$ and $\algpov{B}{C}$ can be understood as the algebras of observables accessible to observers Alice (orange) and Charlie (blue), who can only access system $B$ relative to their own QRF $A$ respectively $C$. These subalgebras are easily defined on the respective Hilbert spaces $\pov{\cH}{A}$ and $\pov{\cH}{C}$, which are equipped with different tensor product structures. To compare and combine the two subalgebras, we \enquote{pull-back} $\algpov{B}{A}$ into Charlie's perspective with the QRF transformation $S_{A\to C}$, obtaining $\algpovpb{B}{A}$.}
    \label{fig:pull-back_algebras}
\end{figure}
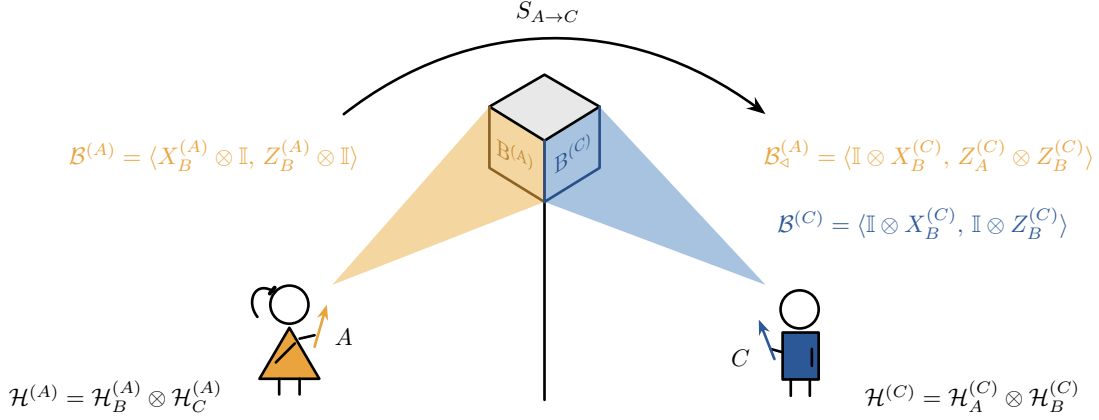
In our simple $\mathbb{Z}_2$-example, the pull-back of Alice's subalgebras into Charlie's perspective leads to
\begin{align}
    \algpovpb{B}{A} &= \langle \pov{\id_A}{C}\otimes \pov{X_B}{C},\pov{Z_A}{C}\otimes \pov{Z_B}{C}\rangle\subseteq \cB\left(\pov{\cH_{AB}}{C}\right), \label{eq:algebra_BApb}\\
    \algpovpb{C}{A} &= \langle \pov{X_A}{C}\otimes \pov{X_B}{C}, \pov{Z_A}{C}\otimes \pov{\id_{B}}{C}\rangle\subseteq \cB\left(\pov{\cH_{AB}}{C}\right).\label{eq:algebra_CApb}
\end{align}
This is illustrated in~\Cref{fig:pull-back_algebras} for $\algpovpb{B}{A}$. Note that both the $X$-orientation of $B$ and the $Z$-orientation of the reference frame remain unchanged -- we can think of these as the $\mathbb{Z}_2$-momentum of the system under consideration and the orientation of the QRF, both of which are conserved under a change of perspective.

Given the algebras $\algpov{B}{A}$ and $\algpov{B}{C}$, we can also make sense of the \emph{joint} algebra of observables on system $B$, as accessible from the union of the two perspectives. This is important for the definition of mutual information across perspectives, since the latter relies on comparing the information contained in each subsystem with the information contained in the joint system. If we think of the algebras $\algpov{B}{A}$ and $\algpov{B}{C}$ as the observables accessible from Alice's and Charlie's perspective respectively, we can think of this joint algebra as those observables that can be accessed when they combine their abilities \cite{Doat_2025}. Denoting by $\alg{A} \vee \alg{B} \coloneqq \braket{\alg{A} \cup \alg{B}}$ the algebra generated by the union of two algebras, we have
\begin{equation}
    \begin{split}
        \algpovpb{B}{A} \vee \algpov{B}{C}&= \; \langle  \pov{\id_A}{C}\otimes \pov{X_B}{C}, \pov{\id_A}{C}\otimes \pov{Z_B}{C}, \pov{Z_A}{C}\otimes \pov{\id_B}{C}\rangle\\
       &\overset{A\leftrightarrow C}{=} \algpov{B}{A}\vee \algpovpf{B}{C}=   \langle  \pov{X_B}{A}\otimes \pov{\id_C}{A} ,  \pov{Z_B}{A}\otimes \pov{\id_C}{A}, \pov{\id_B}{A}\otimes \pov{Z_C}{A}\rangle.
    \end{split}
\end{equation}
Note that the joint algebra is, up to relabeling, independent of whether we choose to embed the subalgebras in the perspective of Charlie or Alice. 
Compared to the total algebra $\algpov{A}{C}\vee \algpov{B}{C}$ of observables relative to $C$, the joint algebra only allows for a subset of measurements on the reference frame $A$, namely its orientation. Formally, this means that any element of $\algpovpb{B}{A} \vee \algpov{B}{C}$ is block-diagonal in $Z_A$. The joint algebra can thus be obtained from the total algebra relative to $C$ through a pinching in $Z_A$:
\begin{equation}
    \algpovpb{B}{A} \vee \algpov{B}{C}=\mathcal{P}_{\pov{Z_A}{C}}\left(\algpov{A}{C}\vee \algpov{B}{C}\right)
\end{equation}
As shown in~\cref{ex:pinched_subsystem}, this map constitutes the unique conditional expectation onto the subalgebra that is diagonal in $\pov{Z_A}{C}$. We will thus write $\cexparg{\pov{AB}{C}\to \pov{B}{C}\pov{B}{A}} \coloneqq \mathcal{P}_{\pov{Z_A}{C}}$.

We may now ask how to go from the joint algebra $\algpovpb{B}{A} \vee \algpov{B}{C}$ back to the subalgebras $\algpov{B}{A}$ and $\algpov{B}{C}$ separately. Again, this step proceeds via appropriate pinching maps. In order to obtain $\algpov{B}{C}$, we merely need to project out any remaining information associated to system $A$. This can be achieved via an additional pinching in the complementary $X$-basis since for any $O\in \cB\left(\pov{\cH_{AB}}{C}\right)$
\begin{equation}
    \mathcal{P}_{\pov{X_A}{C}}\circ \mathcal{P}_{\pov{Z_A}{C}}\left(O\right) = \tfrac{1}{2}\, \pov{\id}{C}_A\otimes \pov{O_B}{C} \text{ where } \pov{O_B}{C} = \Tr_A\left(O\right) .\label{eq:full_pinchingA}
\end{equation}

This follows by first observing that the two pinching maps commute, $\mathcal{P}_{\pov{X_A}{C}}\circ \mathcal{P}_{\pov{Z_A}{C}}=\mathcal{P}_{\pov{Z_A}{C}}\circ \mathcal{P}_{\pov{X_A}{C}}$.
Hence, the resulting operator must be an element of both $\braket{\pov{X_A}{C}} \otimes \cB\left(\pov{\cH_B}{C}\right)$ and of $\braket{\pov{Z_A}{C}} \otimes \cB\left(\pov{\cH_B}{C}\right)$. Since the intersection between $\braket{\pov{Z_A}{C}}$ and $\braket{\pov{X_A}{C}}$ is $\braket{\pov{\id_A}{C}}$, and the pinching maps act trivially on $B$, the resulting operator must be in the form $c \, \pov{\id_A}{C} \otimes \pov{O_B}{C}$. Since the pinching maps are trace-preserving, we must have $c=\tfrac{1}{2}$. Together, the pinching in $X$ and $Z$ recover the usual restriction onto a single tensor-product-subsystem, as expected when going into the algebra associated to $\pov{B}{C}$ within the perspective of Charlie. Hence, we can define $\cexparg{\pov{B}{C}\pov{B}{A}\to \pov{B}{C}}\coloneqq\mathcal{P}_{\pov{X_A}{C}}$. To obtain the correct map onto the subalgebra associated with $\pov{B}{A}$, we use the symmetry of the setup. In the perspective of Alice, we can go from the joint algebra to the subsystem $\pov{B}{C}$ via $\mathcal{P}_{\pov{X_C}{A}}$. Now, let us introduce the pull-back of a map across QRFs.
\begin{definition}[Pull-back of a map across QRFs]
\label{def:pullback_map}
Let $\mathcal{F}:\cB\left(\pov{\cH_{BC}}{A}\right) \mapsto \cB\left(\pov{\cH_{BC}}{A}\right)$ be a linear map. The pull-back of $\mathcal{F}$ under the quantum reference frame transformation $S_{A\to C}$ is defined as
\begin{equation}
\begin{aligned}
    \mappb{\mathcal{F}}: \cB\left(\pov{\cH_{AB}}{C}\right)&\mapsto\cB\left( \pov{\cH_{AB}}{C}\right)\\
    \mappb{\mathcal{F}}\left(\rho\right) &= S_{A\to C}\mathcal{F}\left(S_{A\to C}^\dagger\rho S_{A\to C}\right) S_{A\to C}^\dagger.
\end{aligned}
\end{equation}
Analogously, for a linear map $\mathcal{G}: \cB\left(\pov{\cH_{AB}}{C}\right) \to  \cB\left(\pov{\cH_{AB}}{C}\right)$ we define the push-forward under the QRF transformation $S_{C\to A} = S^{-1}_{A\to C}$, $\mathcal{G}^{\triangleright}: \cB\left(\pov{\cH_{BC}}{A}\right)  \to \cB\left(\pov{\cH_{BC}}{A}\right) $, as  $\mathcal{G}^{\triangleright}\left(\sigma\right) = S_{C\to A}\,\mathcal{G}\left(S_{C\to A}^\dagger\sigma S_{C\to A}\right) S_{C\to A}^\dagger$.
The two operations are mutually inverse, i.e., $\left(\mappb{\mathcal{F}}\right)^{\triangleright} = \mathcal{F}$ and $\left(\mathcal{G}^{\triangleright}\right)^{\triangleleft} = \mathcal{G}$.
\end{definition}
Applied to our pinching-map, we find 
\begin{equation}    \mappb{\mathcal{P}}_{\pov{X_C}{A}}= \frac{1}{2}\left[\pov{\id_{AB}}{C}\left(\cdot\right)\pov{\id_{AB}}{C} + \left(\pov{X_A}{C}\otimes \pov{X_B}{C}\right) \left(\cdot\right) \left(\pov{X_A}{C}\otimes \pov{X_B}{C}\right)\right] =: \mathcal{P}_{\pov{X_A}{C}\otimes \pov{X_B}{C}}.\label{eq:P_XX}
\end{equation}
The derivation is provided in \cref{app:proof_pullback_mapsZ2}. This gives us the conditional expectation $\cexparg{\pov{B}{C}\pov{B}{A}\to \pov{B}{A}}^\triangleleft \coloneqq \mappb{\mathcal{P}}_{\pov{X_C}{A}}=\mathcal{P}_{\pov{X_A}{C}\otimes \pov{X_B}{C}}$. The conditional expectations onto the various subalgebras are depicted in the right branch of~\Cref{fig:encoded-decomposition}.
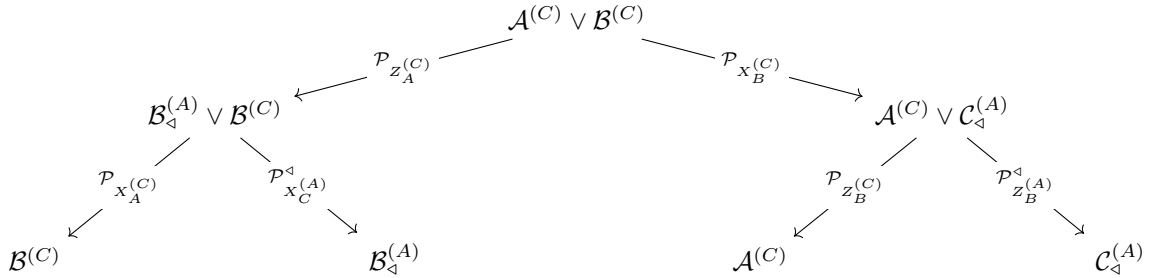
\begin{figure}[h]
\centering
\[\begin{tikzcd}
	&&& {\algpov{A}{C}\vee\algpov{B}{C}} &&& \\
	& {\algpovpb{B}{A}\vee\algpov{B}{C}}  &&&& {\algpov{A}{C}\vee\algpovpb{C}{A}} \\
	\\
	{\algpov{B}{C}} && {\algpovpb{B}{A}} && {\algpov{A}{C}} && {\algpovpb{C}{A}}
	\arrow["{\mathcal{P}_{\pov{X_B}{C}}}"{description}, from=1-4, to=2-6]
	\arrow["{\mathcal{P}_{\pov{Z_A}{C}}}"{description}, from=1-4, to=2-2]
	\arrow["{\mathcal{P}_{\pov{Z_B}{C}}}"{description}, from=2-6, to=4-5]
	\arrow["{\mappb{\mathcal{P}}_{\pov{Z}{A}_B}}"{description}, from=2-6, to=4-7]
	\arrow["{\mathcal{P}_{\pov{X_A}{C}}}"{description}, from=2-2, to=4-1]
	\arrow["{\mappb{\mathcal{P}}_{\pov{X}{A}_C}}"{description}, from=2-2, to=4-3]
\end{tikzcd}\]
\caption{Conditional expectations from the total algebra of observables relative to $C$ onto the various subalgebras for $G=\mathbb{Z}_2$.}
\label{fig:encoded-decomposition}
\end{figure}

Analogous reasoning leads to the subalgebra $\algpov{A}{C}$ associated to $A$ relative to $C$ and the subalgebra $\algpovpb{C}{A}$ of $C$ relative to $A$, pulled back into the perspective of Charlie. Together, they form the joint algebra $\algpov{A}{C}\vee \algpovpb{C}{A}$. The conditional expectation onto this algebra is a pinching in the $\mathbb{Z}_2$-momentum of $B$, that is $\cexparg{\pov{AB}{C}\to \pov{A}{C}\pov{C}{C}}\coloneqq \mathcal{P}_{\pov{X_B}{C}}$ while the subalgebras are obtained via $\cexparg{\pov{A}{C}\pov{C}{A}\to \pov{A}{C}} \coloneqq \mathcal{P}_{\pov{Z_B}{C}}$ and $\cexparg{\pov{A}{C}\pov{C}{A}\to \pov{C}{A}}^{\triangleleft} \coloneqq \mathcal{P}_{\pov{Z_A}{C}\otimes \pov{Z_B}{C}}$. Details of the construction can be found in \cref{app:proof_pullback_mapsZ2}.

Finally, let us note that -- unlike the subsystems defined within a fixed tensor product structure -- the subsystem algebras relative to different QRFs have non-trivial overlap \cite{Castro_Ruiz_2021,Hoehn_2023}. Let us define the intersection between two algebras $\alg{A}, \alg{B}$ as $\alg{A} \wedge\alg{B} \coloneqq \braket{\alg{A} \cap\alg{B}}$. Then,
\begin{equation}
\begin{aligned}
   & \algpovpb{B}{A}\wedge \algpov{B}{C} = \Big\langle \pov{\id_A}{C}\otimes \pov{X_B}{C}\Big\rangle,\\
    &\algpov{A}{C}\wedge \algpovpb{C}{A} = \Big\langle \pov{Z_A}{C}\otimes \pov{\id_{B}}{C}\Big\rangle.
\end{aligned}
\end{equation}
Thus, the overlaps are given by the $\mathbb{Z}_2$-invariant $X$-orientation of system $B$ and the relative orientation of $A$ and $C$ -- both quantities that are preserved under the QRF transformation and thus shared across the perspectives. Starting from the joint algebra, the projector onto the intersection is obtained by composing the conditional expectations onto the two subalgebras of interest:
\begin{equation}
\label{eq:cond_exp_overlap_Z2_2}
    \begin{split}
        \cexparg{\pov{B}{A}\pov{B}{C}\to \algpovpb{B}{A}\wedge \algpov{B}{C}}\left(O\right) &= \left(\mathcal{P}_{\pov{X_A}{C}}\circ \mathcal{P}_{\pov{X_A}{C}\otimes \pov{X_B}{C}}\right)\left(O\right) 
    \\&= \left(\mathcal{P}_{\pov{X_A}{C}\otimes \pov{X_B}{C}}\circ \mathcal{P}_{\pov{X_A}{C}}\right)\left(O\right) = \tfrac{1}{2}\pov{\id_A}{C}\otimes \mathcal{P}_{\pov{X_B}{C}}\left(\Tr_AO\right),
\end{split}
\end{equation}
and
\begin{equation}
    \begin{split}
\label{eq:cond_exp_overlap_Z2_1}
    \cexparg{\pov{A}{C}\pov{C}{A}\to \pov{A}{C}\wedge \pov{C}{A}}\left(O'\right) &= \left(\mathcal{P}_{\pov{Z_B}{C}}\circ \mathcal{P}_{\pov{Z_A}{C}\otimes \pov{Z_B}{C}}\right)\left(O'\right) \\
    &= \left(\mathcal{P}_{\pov{Z_A}{C}\otimes \pov{Z_B}{C}}\circ \mathcal{P}_{\pov{Z_B}{C}}\right)\left(O'\right) = \mathcal{P}_{\pov{Z_A}{C}}\left(\Tr_BO'\right)\otimes \tfrac{1}{2}\pov{\id_{B}}{C}
    \end{split}
\end{equation}
for $O \in \algpovpb{B}{A} \vee \algpov{B}{C}$ and $O'\in \algpovpb{C}{A} \vee \algpov{A}{C}$. The derivation is provided in \cref{app:conditional_exp_construction_proofs}. Note that the commutativity of the conditional expectations is not a given -- it holds for the subsystems defined relative to different QRFs but may fail for more general subsystems. This feature of QRFs will play an important role in the definition of mutual information across perspectives.

\subsection{General Finite-Dimensional QRFs}

Since we are interested in quantities such as entropy and information, which quickly diverge for continuous groups, we focus on the case of QRFs for finite-dimensional symmetry groups. The algebraic construction presented in this section should, however, extend relatively straightforwardly to locally compact groups if one takes care in defining the trace \cite{Sorce_2024}. The case of general finite-dimensional QRFs proceeds along very similar lines as the construction for $\mathbb{Z}_2$ reference frames. We we will thus keep the discussion short and draw upon the previous section for physical intuition, but we will highlight a few subtleties that arise when working with QRFs for general and, in particular, non-abelian symmetry groups. The results are summarised in~\Cref{fig:subsystem_structure_general_groups}.

Consider a finite symmetry group $G$ and three systems $A, B, C$ with perspectival Hilbert spaces isomorphic to $L^2\left(G\right)$ with dimension $|G|$. The QRF transformation from $C$ to $A$ is given by \cite{delaHamette_2020}
\begin{equation}
    S_{C\to A} = \sum_{g\in G}\ket{g^{-1}}_C\bra{g}_A\otimes U_B^\dagger\left(g\right).
\end{equation}
The Hilbert space of $A$ and $B$ relative to $C$ takes the form $\pov{\cH_{AB}}{C}\simeq \pov{\cH_A}{C}\otimes \pov{\cH_B}{C}$ while that of $B$ and $C$ relative to $A$ comes with a different TP-decomposition $\pov{\cH_{BC}}{A}\simeq\pov{\cH_{B}}{A}\otimes \pov{\cH_C}{A}$. Again without loss of generality, we describe everything from the point of view of Charlie. The subalgebras of $A$ and $B$ relative to $C$ can be defined as before
\begin{equation}
\begin{aligned}
    \algpov{A}{C} &= \left\{O\otimes \pov{\id_B}{C}|O\in \cB\left(\pov{\cH_A}{C}\right)\right\}\subseteq \cB\left(\pov{\cH_{AB}}{C}\right),\\ 
    \algpov{B}{C} &= \left\{ \pov{\id_A}{C}\otimes O|O\in \cB\left(\pov{\cH_B}{C}\right)\right\}\subseteq \cB\left(\pov{\cH_{AB}}{C}\right).
\end{aligned}
\end{equation}
To obtain the subalgebras of $B$ and $C$ relative to $A$ in the perspective of Charlie, we use the pull-back of the subsystem algebras defined relative to Alice to obtain
\begin{equation}
\begin{aligned}
     \algpovpb{B}{A} &= \left\{\mathcal{S}_{A\to C}\left(O\otimes \pov{\id_C}{A}\right)|O\in \cB\left(\pov{\cH_B}{A}\right)\right\} 
     = \left\{ \sum_g \ketbra{g}{g}_A\otimes U_B\left(g\right)OU_B^\dagger\left(g\right)|O\in \cB\left(\pov{\cH_B}{C}\right) \right\},\\
    \algpovpb{C}{A} &= \left\{ \mathcal{S}_{A\to C}\left(\pov{\id_B}{A}\otimes O\right)|O\in \cB\left(\pov{\cH_C}{A}\right)\right\} = \left\{ \sum_{g,g'}\ketbra{g}{g'}_A \langle g^{-1}|O|g'^{-1}\rangle\otimes U_B\left(g g'^{-1}\right)|O\in \cB\left(\pov{\cH_A}{C}\right)\right\}.
\end{aligned}
\end{equation}
Based on this, we can construct the joint algebra as the algebra spanned by the observables in both $\algpovpb{B}{A}$ and $\algpov{B}{C}$:
\begin{equation}
    \algpovpb{B}{A} \vee \algpov{B}{C}=\Big \langle \pov{\id_A}{C}\otimes O , \sum_g \pov{\ketbra{g}{g}_A}{C}\otimes U_B\left(g\right)OU_B^\dagger\left(g\right) \Big\rangle_{O\in \cB\left(\pov{\cH_B}{C}\right)} = \mathcal{P}_{\pov{G_A}{C}}\left(\algpov{A}{C}\vee \algpov{B}{C}\right)\label{eq:joint_algebra_BABC}
\end{equation}
where $\mathcal{P}_{\pov{G_A}{C}}\left(\cdot\right) \coloneqq \sum_{g\in G} \pov{\ketbra{g}{g}}{C}_A\otimes \langle g| \cdot \pov{|g \rangle}{C}_A$. The proof that the joint algebra is indeed obtained from the total algebra relative to $C$ via the pinching of $A$ in the group basis is provided in~\cref{app:joint_algebras}. There, we also show that $\mappb{\mathcal{P}}_{\pov{G_C}{A}}=\mathcal{P}_{\pov{G_A}{C}}$, which implies that the definition of the joint algebra is independent of whether it is constructed in Charlie's or Alice's perspective. Notice that the pinching operation makes the joint algebra diagonal in the group basis of system $A$. Operationally, we can think of this as the orientation of $A$ being effectively \enquote{classical} for an agent, who only has access to observables in $\algpovpb{B}{A} \vee \algpov{B}{C}$, since they cannot tell whether $A$ is in a classical mixture or a superposition of different orientations \cite{Vanrietvelde_2025}. As the pinching map defines the unique conditional expectation onto the corresponding block-diagonal subalgebra, we have $\cexparg{\pov{AB}{C}\to \pov{B}{C}\pov{B}{A}}= \mathcal{P}_{\pov{G_A}{C}}$.

Given the joint algebra, we can also construct the conditional expectations into the respective subalgebras. To go into the subalgebra associated to $B$ relative to $C$, we want to eliminate all information about $\pov{A}{C}$. In the case of $\mathbb{Z}_2$-QRFs, we achieved this by combining the pinching in $Z$ with a pinching in $X$. The counterpart of $\mathcal{P}_X$ for general symmetry groups is the \emph{incoherent $G$-twirl} on subsystem $A$ relative to $C$. In particular, we have
\begin{equation}
    \mathcal{G}_{\inc,\pov{A}{C}}\circ \mathcal{P}_{\pov{G_A}{C}}\left(\pov{O_{AB}}{C}\right) = \frac{1}{|G|}\sum_{g,h \in G}U_A\left(g\right)\pov{\ketbra{h}{h}_A}{C}U_A^\dagger\left(g\right)\otimes \braket{h|\pov{O_{AB}}{C}|h}_A = \frac{1}{|G|}\pov{\id_A}{C}\otimes \Tr_A\left(\pov{O_{AB}}{C}\right).\label{eq:partial_trace_genQRF}
\end{equation}
Thus, the conditional expectation is given by $\cexparg{\pov{B}{C}\pov{B}{A}\to \pov{B}{C}}= \mathcal{G}_{\inc,\pov{A}{C}}$. The same reasoning holds when we restrict to $\pov{B}{A}$ in the perspective of Alice. Importing the corresponding map into Charlie's perspective, we obtain $\cexparg{\pov{B}{C}\pov{B}{A}\to \pov{B}{A}}^{\triangleleft} = \mappb{\mathcal{G}}_{\inc,\pov{C}{A}}$ where
\begin{equation}
    \mappb{\mathcal{G}}_{\inc,\pov{C}{A}} = \frac{1}{|G|}\sum_{g\in G} \left(S_{A\to C} U_C\left(g\right)S_{A\to C}^\dagger\right)  \rho \left(S_{A\to C} U_C\left(g\right) S_{A\to C}^\dagger\right)^\dagger.
\end{equation}
In general, we have 
\begin{equation}
    S_{A\to C} \left(\pov{\id_B}{A}\otimes U_C\left(g\right)\right)S^\dagger_{A\to C}= \sum_h \pov{\ketbra{hg^{-1}}{h}_A}{C}\otimes U_B\left(hg^{-1} h^{-1}\right).
\end{equation}
This expression simplifies significantly in the abelian case where $S_{A\to C} \left(\pov{\id_B}{A}\otimes U_C\left(g\right)\right)S^\dagger_{A\to C}=U_A^\dagger\left(g\right)\otimes U_B^\dagger\left(g\right)$ and thus $\mappb{\mathcal{G}^{\mathrm{abelian}}_{\inc,\pov{C}{A}}}\left(\rho\right)= \mathcal{G}_{\inc,\pov{AB}{C}}\left(\rho\right)$. The example of $\mathbb{Z}_2$ falls into this category, see \cref{eq:P_XX}.

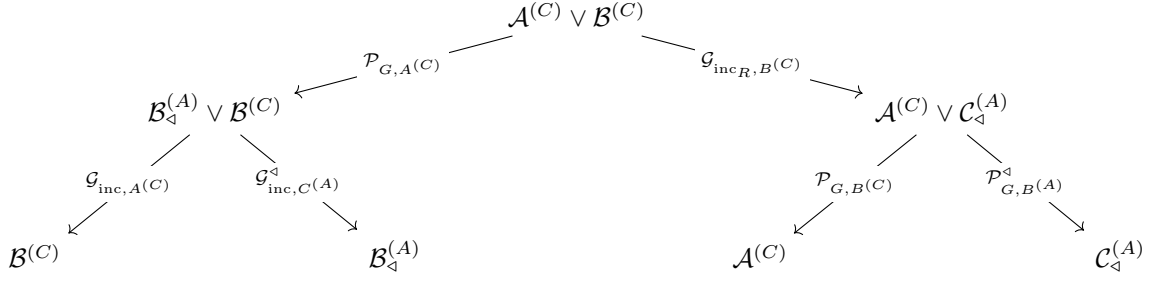
\begin{figure}[h]
\centering
\[\begin{tikzcd}
	&&& {\algpov{A}{C}\vee\algpov{B}{C}} &&& \\
	& {\algpovpb{B}{A}\vee\algpov{B}{C}}  &&&& {\algpov{A}{C}\vee\algpovpb{C}{A}}\\
	\\
	 {\algpov{B}{C}} && {\algpovpb{B}{A}} &&{\algpov{A}{C}} && {\algpovpb{C}{A}} 
	\arrow["{\mathcal{G}_{\inc_R,\pov{B}{C}}}"{description}, from=1-4, to=2-6]
	\arrow["{\mathcal{P}_{G,\pov{A}{C}}}"{description}, from=1-4, to=2-2]
	\arrow["{\mathcal{P}_{G,\pov{B}{C}}}"{description}, from=2-6, to=4-5]
	\arrow["{\mappb{\mathcal{P}}_{G,{\pov{B}{A}}}}"{description}, from=2-6, to=4-7]
	\arrow["{\mathcal{G}_{\inc,\pov{A}{C}}}"{description}, from=2-2, to=4-1]
	\arrow["{\mappb{\mathcal{G}}_{\inc,\pov{C}{A}}}"{description}, from=2-2, to=4-3]
\end{tikzcd}\]
\caption{Conditional expectations from the total algebra of observables relative to $C$ onto the various subalgebras for general finite-dimensional symmetry groups $G$. Here, $\mathcal{P}_G$ denotes pinching in the group basis while $\mathcal{G}_{\inc_R}$ is the right incoherent $\mathcal{G}$-twirl (the distinction between left and right incoherent $\mathcal{G}$-twirl only matters in the non-abelian case). The subscripts further specify on which subsystem the map acts non-trivially. Finally, $\mappb{\mathcal{P}}_{G}$ and $\mappb{\mathcal{G}}_{\inc_R}$ denote the pull-backs of the corresponding maps in Alice's perspective via the QRF transformation.}
\label{fig:subsystem_structure_general_groups}
\end{figure}

Similar reasoning (see~\cref{app:joint_alg_AC}) leads to the subalgebras $\algpov{A}{C}$ describing $A$ relative to $C$ and $\algpovpb{C}{A}$ describing $C$ relative to $A$ -- both in Charlie's perspective. From these, we can construct the joint perspective $\algpov{A}{C}\vee \algpovpb{C}{A}$. The latter is obtained from the total algebra of observables through the \emph{right} incoherent G-twirl on $B$, that is $\cexparg{\pov{AB}{C}\to \pov{A}{C}\pov{C}{C}}=\mathcal{G}_{\inc_R,\pov{B}{C}}$, where
\begin{equation}
    \mathcal{G}_{\inc_R} = \frac{1}{|G|}\sum_{g\in G} R\left(g\right)\cdot R^\dagger\left(g\right)
\end{equation}
with $R\left(h\right)\ket{g} = \ket{gh^{-1}}$ the right-regular representation of $G$. Note that for abelian groups, the left and right incoherent $\mathcal{G}$-twirl are equivalent: their action on an arbitrary basis element $\ketbra{h}{h'}$ is 
\begin{equation}
    \mathcal{G}_{\inc_R}\left(\ketbra{h}{h'}\right) =\frac{1}{|G|} \sum_{g\in G} \ketbra{hg^{-1}}{h'g^{-1}} = \frac{1}{|G|}\sum_{g\in G} \ketbra{hg}{h'g} \overset{\text{abelian}}{=} \frac{1}{|G|}\sum_{g\in G}\ketbra{gh}{gh'} = \mathcal{G}_{\inc}\left(\ketbra{h}{h'}\right)
\end{equation}
for all $h,h'\in G$. 
From the joint algebra, we can go into the different subalgebras via $\cexparg{\pov{A}{C}\pov{C}{A}\to \pov{A}{C}} = \mathcal{P}_{G,\pov{B}{C}}$ and $\cexparg{\pov{A}{C}\pov{C}{A}\to \pov{C}{A}} = \mappb{\mathcal{P}}_{G,\pov{B}{A}}$ where 
\begin{equation}
    \mappb{\mathcal{P}}_{G,\pov{B}{A}}(\pov{\rho_{AB}}{C}) = \sum_{g \in G} \left(S_{A\to C}\left(\pov{\ketbra{g}{g}_B}{A}\otimes \pov{\id_C}{A}\right)S_{A\to C}^\dagger\right) \pov{\rho_{AB}}{C}\left(S_{A\to C} \left(\pov{\ketbra{g}{g}_B}{A}\otimes \pov{\id_C}{A}\right)S_{A\to C}^\dagger \right)^\dagger
\end{equation}
with
\begin{equation}
    S_{A\to C}\left(\pov{\ketbra{g}{g}_B}{A}\otimes \pov{\id_C}{A}\right)S_{A\to C}^\dagger = \sum_{h \in G}\pov{\ketbra{h}{h}_A}{C}\otimes \pov{\ketbra{hg}{hg}_B}{C}.
\end{equation}
The various different subalgebras and corresponding conditional expectations for general finite symmetry groups are depicted in~\Cref{fig:subsystem_structure_general_groups}.\\

Finally, the intersections between the different subalgebras are given by
\begin{equation}
\begin{aligned}
    &\algpovpb{B}{A}\wedge \algpov{B}{C} = \left\{\pov{\id_A}{C}\otimes O \mid O \in \cB\left(\pov{\cH_B}{C}\right) \text{ s.t. } [O,U_B\left(g\right)] =0 \;\; \forall g\in G \right\},\\
    &\algpov{A}{C}\wedge \algpovpb{C}{A} = \left\{O\otimes \pov{\id_{B}}{C}\mid O\in \cB\left(\pov{\cH_A}{C}\right) \text{ s.t. } \braket{g|O|g'} \sim \delta_{gg'} \;\; \forall g,g'\in G\right\}.
\end{aligned}
\end{equation}
Again, these contain the $G$-invariant properties of system $B$ and the relative orientation between $A$ and $C$, which are conserved under the QRF transformation. We can obtain the projectors onto these subalgebras by composing the conditional expectations. In the case of a general symmetry group, it will be easier to work with the conditional expectations from the total algebra $\algpov{A}{C}\vee \algpov{B}{C}$ onto the respective subsystems rather than starting from the joint algebras. Let us thus define $\cexparg{\pov{B}{C}} \coloneqq \tfrac{1}{|G|}\pov{\id_A}{C}\otimes \Tr_A$ as in \cref{eq:partial_trace_genQRF} and $\cexparg{\pov{B}{A}} \coloneqq \Tr_C \otimes \tfrac{1}{|G|}\id _C$.  
Then,
\begin{align}
    \cexparg{\pov{B}{A}\wedge \pov{B}{C}}\left(O\right) \coloneqq \left(\cexparg{\pov{B}{A}}^{\triangleleft}\circ \cexparg{\pov{B}{C}}\right)\left(O\right) = \left(\cexparg{\pov{B}{C}}\circ \cexparg{\pov{B}{A}}^{\triangleleft}\right)\left(O\right) = \frac{1}{|G|}\pov{\id_A}{C}\otimes \mathcal{G}_{\inc,\pov{B}{C}}\left(\Tr_A\left(O\right)\right) 
    \label{eq:CE_overlap_B}
\end{align}
for any $O\in\cB\left(\pov{\cH_{AB}}{C}\right)$. The proof is provided in~\cref{app:overlap_algebras}. As expected, the conditional expectation removes everything but the invariant operators on $B$. This map also recovers the $\mathbb{Z}_2$-case upon identifying the incoherent G-twirl with pinching in $\pov{X_B}{C}$. Similarly, we can define $\cexparg{\pov{A}{C}}\coloneqq \Tr_{\pov{B}{C}}\otimes \frac{1}{|G|}\pov{\id_{B}}{C}$ and $\cexparg{\pov{C}{A}} \coloneqq \frac{1}{|G|}\pov{\id_B}{A}\otimes \Tr_{\pov{B}{A}}$ to obtain
\begin{equation}
    \cexparg{\pov{A}{C}\wedge \pov{C}{A}}\left(O\right) \coloneqq \left(\cexparg{\pov{A}{C}}\circ \cexparg{\pov{C}{A}}^{\triangleleft}\right)\left(O\right) = \left(\cexparg{\pov{C}{A}}^{\triangleleft}\circ \cexparg{\pov{A}{C}}\right)\left(O\right) = \mathcal{P}_{G,\pov{A}{C}}\left(Tr_{\pov{B}{C}}\left(O\right)\right)\otimes \frac{1}{|G|}\pov{\id_{B}}{C}\label{eq:CE_overlap_AC}
\end{equation}
for any $O\in\cB\left(\pov{\cH_{AB}}{C}\right)$. This map projects out everything but the relative orientation of $A$ and $C$. Again, let us stress that the conditional expectations commuting and yielding the projector onto the intersection is a non-trivial feature of subsystems defined relative to different QRFs, which will be crucial in the definition of mutual information below.

\section{Mutual Information Across QRF Perspectives}
\label{sec:mutual_information_QRFs}

Having understood the algebraic structure of all the different subsystems, the joint algebras, and their overlaps, we can now finally turn to our original question: what is the mutual information across the different perspectives? As argued earlier, we will base our construction on the chain rule for mutual information, which, in the tensor-product case, is given by $\minfo\left(A:B\right) = \info\left(AB\right) - \info\left(A\right) - \info\left(B\right)$. In \cref{sec:modified_chain_rule} we will argue that, to account for the non-trivial overlap between the subsystem algebras this rule needs to be slightly modified. This modification vanishes for subsystems defined by a tensor-product structure, as well as for algebras with trivial overlap, while also ensuring the subadditivity of mutual information for the more general subsystems relative to different QRFs. Our arguments lead to a definition of mutual information that agrees with a notion of conditional mutual information put forward in the literature on entropic uncertainty relations \cite{LaRacuente_2019,Gao_2019}. We briefly discuss in \cref{sec:relation_other_work} how our definition relates to these works.
Finally, in \cref{sec:cond_MI}, we apply the same arguments to define conditional information as well as conditional mutual information and prove that the latter satisfies strong subadditivity for the subsystems relative to different QRFs. 

\subsection{Definition via Modified Chain Rule}
\label{sec:modified_chain_rule}

To gain a clearer intuition for the problem at hand, let us begin by focusing on a concrete example: the mutual information $\minfo\left(\pov{B}{A}:\pov{B}{C}\right)$ between $B$ relative to $A$ and $B$ relative to $C$. Returning to~\Cref{fig:pull-back_algebras}, we can imagine two observers Alice and Charlie, who only have access to system $B$ but not to the other's frame of reference. In this case, $\minfo\left(\pov{B}{A}:\pov{B}{C}\right)$ should quantify how much information is shared across their perspectives. A naive application of the chain rule for mutual information suggests that $\minfo\left(\pov{B}{A}:\pov{B}{C}\right)\overset{\text{tentative}}{=}\info\left(\pov{B}{C}\pov{B}{A}\right) - \info\left(\pov{B}{A}\right) - \info\left(\pov{B}{C}\right)$ where, for a given state $\rho \in \cB\left(\pov{\cH_{AB}}{C}\right)$,
\begin{equation}
\begin{aligned}
    \info\left(\pov{B}{C}\pov{B}{A}\right)_\rho &= \info\left(\algpovpb{B}{A}\vee \algpov{B}{C}\right)_\rho 
    =D\left(\cexparg{\pov{B}{A}\vee \pov{B}{C}}\left(\rho\right),\pov{\pi_{AB}}{C}\right)
    ,\\
    \info\left(\pov{B}{A}\right)_\rho &= \info\left(\algpovpb{B}{A}\right)_\rho 
    =D\left(\cexparg{\pov{B}{A}}^{\triangleleft}\left(\rho\right),\pov{\pi_{AB}}{C}\right)
    ,\\
    \info\left(\pov{B}{C}\right)_\rho &= \info\left(\algpov{B}{C}\right)_\rho =
    D\left(\cexparg{\pov{B}{C}}\left(\rho\right),\pov{\pi_{AB}}{C}\right).
\end{aligned}
\end{equation}
Here, we used the algebraic definition of information~(\cref{def:algebraic_information}). 
As shown in the previous section and illustrated in~\Cref{fig:subsystem_structure_general_groups}, the conditional expectations for general finite-dimensional symmetry groups are given by
\begin{align}
    \cexparg{\pov{B}{A}\vee \pov{B}{C}} &= \mathcal{P}_{G,\pov{A}{C}}, \qquad
    \cexparg{\pov{B}{A}}^{\triangleleft} = \mappb{\mathcal{G}}_{\inc,\pov{C}{A}} \circ\mathcal{P}_{G,\pov{A}{C}}, \qquad
    \cexparg{\pov{B}{C}}=\mathcal{G}_{\inc,\pov{A}{C}}\circ\mathcal{P}_{G,\pov{A}{C}}.
\end{align}
We also defined the maximally mixed state on the total Hilbert space relative to Charlie as $\pov{\pi_{AB}}{C}\coloneqq\tfrac{1}{|G|^2}\pov{\id_{AB}}{C}$. The mutual information ought to characterise the difference in the information accessible through the joint algebra $\algpovpb{B}{A} \vee \algpov{B}{C}$, i.e., when Alice and Charlie can combine their operations on system $B$, versus the information accessible when each of them performs operations on system $B$ separately, through $\algpovpb{B}{A}$ and $\algpov{B}{C}$. As already stressed at several points, however, the algebras $\algpovpb{B}{A}$ and $\algpov{B}{C}$ overlap non-trivially (cf.~\cite{Castro_Ruiz_2021, Hoehn_2023}). This is illustrated in~\cref{fig:overlap} for the example of $G=\mathbb{Z}_2$.

\begin{figure}
    \centering        \includegraphics[width=0.33\linewidth]{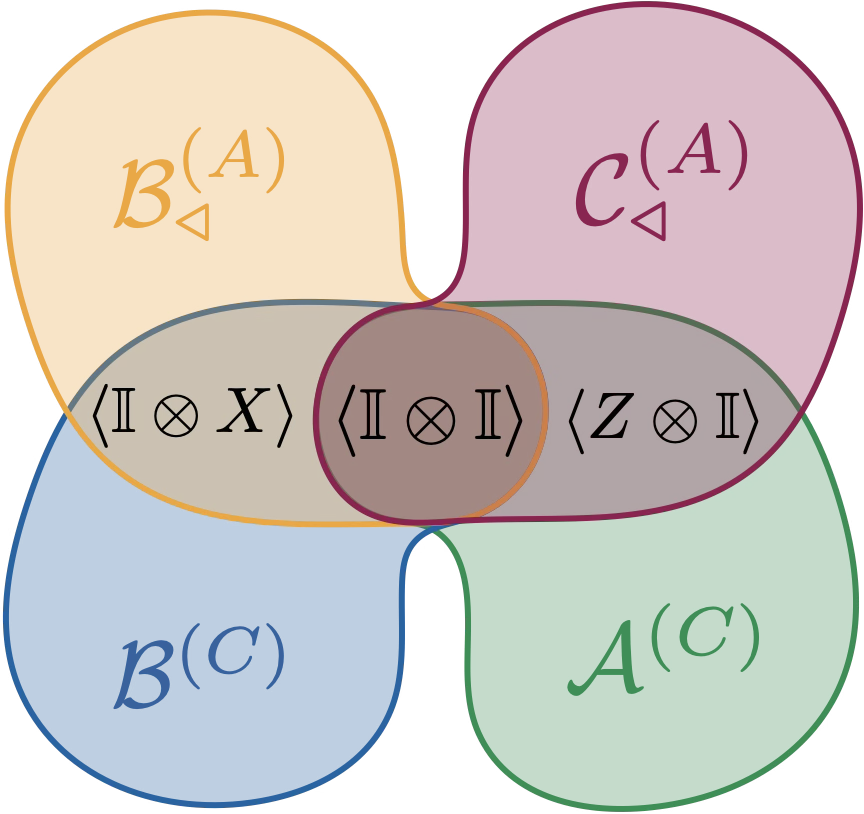}
    \caption{Non-trivial overlap between the algebras relative to different frames of reference for the example $G=\mathbb{Z}_2$. Within one perspective, the subsystem algebras only share multiples of the identity operator -- this is the case for $\algpov{B}{C} \wedge \algpov{A}{C}$ as well as $\algpovpb{B}{A}\wedge \algpovpb{C}{A}$. Across different perspectives, however, the overlap becomes non-trivial. In order to avoid overcounting, we need to subtract the information associated with the overlap in the definition of mutual information.}
    \label{fig:overlap}
\end{figure}
In particular, both Alice and Charlie have access to the $G$-invariant properties of system $B$. When simply adding $\info\left(\algpovpb{B}{A}\right)$ and $\info\left(\algpov{B}{C}\right)$, we run the risk of overcounting the information associated to this overlap. This can be made more precise by explicitly separating the information associated to the overlap from the overall information associated to $B$ relative to $A$ or $C$ through \cref{th:dec_rel_entropy}. Applying this to the information associated to $\algpovpb{B}{A}$ and $\algpov{B}{C}$ with overlap $\algpovpb{B}{A}\wedge \algpov{B}{C}$, we obtain
\begin{align}
    \info\left(\algpovpb{B}{A}\right)_\rho &=   D\left(\cexparg{\pov{B}{A}}^{\triangleleft}\left(\rho\right),\cexparg{\algpovpb{B}{A}\wedge \algpov{B}{C}}\left(\rho\right)\right)+\underset{= \info\left(\algpovpb{B}{A}\wedge \algpov{B}{C}\right)}{\underbrace{D\left(\cexparg{\algpovpb{B}{A}\wedge \algpov{B}{C}}\left(\rho\right), \pov{\pi_{AB}}{C}\right)}}, \label{eq:info_dec1}\\
    \info\left(\algpov{B}{C}\right)_\rho &= D\left(\cexparg{\pov{B}{C}}\left(\rho\right),\cexparg{\algpovpb{B}{A}\wedge \algpov{B}{C}}\left(\rho\right)\right)+\underset{= \info\left(\algpovpb{B}{A}\wedge \algpov{B}{C} \right)}{\underbrace{D\left(\cexparg{\pov{B}{A} \wedge \pov{B}{C}}\left(\rho\right), \pov{\pi_{AB}}{C}\right)}},\label{eq:info_dec2}
\end{align}
where we made use of the fact that the two conditional expectations onto the respective subalgebras commute and yield the projector onto their overlap. Thus, when we naively add $\info(\algpovpb{B}{A})$ and $\info(\algpov{B}{C})$, we are overcounting the information $\info\left(\algpovpb{B}{A}\wedge \algpov{B}{C}\right)$ associated to their overlap. This motivates the following definition for mutual information across perspectives.

\begin{definition}[Mutual information between different QRF perspectives]
\label{def: mutual_info_QRFs}
We define the mutual information between $B$ seen from $A$ and $B$ seen from $C$, as:
\begin{equation}
    \minfo\left(\pov{B}{A}:\pov{B}{C}\right)_\rho \coloneqq \info\left(\algpovpb{B}{A}\vee \algpov{B}{C}\right) - \info\left(\algpovpb{B}{A}\right)-\info\left(\algpov{B}{C}\right) + \info\left(\algpovpb{B}{A}\wedge \algpov{B}{C}\right).
\end{equation}
This quantity is symmetric under exchange of $\pov{B}{A}$ and $\pov{B}{C}$ and independent of whether we construct the subalgebras and conditional expectations in Charlie's perspective or in Alice's perspective.  
\end{definition}

The proof that this quantity is independent of the embedding into Alice's or Charlie's perspective is given in \cref{app:mutual_info}. Moreover, the above definition of mutual information always yields a non-negative quantity, a property that is known as subadditivity. 

\begin{proposition}[Subadditivity]\label{SA_BA_BC}
    The mutual information between $\pov{B}{A}$ and $\pov{B}{C}$ satisfies $\minfo\left(\pov{B}{A}:\pov{B}{C}\right)_\rho \geq 0$ for all states $\rho$.
\end{proposition}

\begin{proof}
Using~\cref{th:decomposition_rel_entropy} to decompose the information associated to the joint algebra,
\begin{equation}
    \info\left(\algpovpb{B}{A}\vee \algpov{B}{C}\right)_\rho = D\left(\cexparg{\pov{B}{C}\pov{B}{A}}(\rho),\cexparg{\pov{B}{A}}^{\triangleleft}(\rho)\right)+\info\left(\algpovpb{B}{A}\right)\label{eq:info_dec3},
\end{equation}
and combining this with \cref{eq:info_dec2}, we obtain
\begin{align}
    \minfo\left(\pov{B}{A}:\pov{B}{C}\right)_\rho = D\left(\cexparg{\pov{B}{C}\pov{B}{A}}\left(\rho\right), \cexparg{\pov{B}{A}}^{\triangleleft}\left(\rho\right)\right) - D\left(\cexparg{\pov{B}{C}}\left(\rho\right), \cexparg{\pov{B}{A} \wedge \pov{B}{C}}\left(\rho\right)\right).
\end{align}
Then, the statement follows directly from the \emph{data processing inequality}~\cite{Nielsen_Chuang_2010}. Note first that $\cexparg{\pov{B}{A} \wedge \pov{B}{C}}= \cexparg{\pov{B}{C}}\circ \cexparg{\pov{B}{A}}^{\triangleleft}$. Next, define $\tilde{\rho} \coloneqq \cexparg{\pov{B}{C}\pov{B}{A}}\left(\rho\right)$ and $\sigma \coloneqq \cexparg{\pov{B}{A}}^{\triangleleft}\left(\rho\right)$. Since $\algpov{B}{C}\subset \algpovpb{B}{A} \vee \algpov{B}{C}$, we also have $\cexparg{\pov{B}{C}}\left(\rho\right) = \cexparg{\pov{B}{C}}\left(\tilde{\rho}\right)$. Thus, we can write 
    \begin{equation}
        \minfo\left(\pov{B}{A}:\pov{B}{C}\right)_\rho = D\left(\tilde{\rho},\sigma\right) - D\left(\cexparg{\pov{B}{C}}\left(\tilde{\rho}\right), \cexparg{\pov{B}{C}}\left(\sigma\right)\right)
    \end{equation}
    Now, the data processing inequality states that $D\left(\rho,\sigma\right) \geq D\left(\Phi\left(\rho\right),\Phi\left(\sigma\right)\right)$ for any CPTP map $\Phi$. Identifying $\Phi = \cexparg{\pov{B}{C}}$ and using the fact that this is indeed a CPTP map, we obtain $D\left(\tilde{\rho},\sigma\right) \geq D\left(\cexparg{\pov{B}{C}}\left(\tilde{\rho}\right), \cexparg{\pov{B}{C}}\left(\sigma\right)\right)$ and thus
\begin{equation}
    \minfo\left(\pov{B}{A}:\pov{B}{C}\right)_\rho \geq 0
\end{equation}
for all states $\rho$.
\end{proof}

The derivation of the modified chain rule is illustrated in Fig.~\ref{fig:information_diagram}. The definition can be extended freely to other combinations of subsystems across different perspectives. For instance, we also have
\begin{equation}
   \minfo\left(\pov{A}{C}:\pov{C}{A}\right)_\rho \coloneqq \info\left(\algpov{A}{C}\vee \algpovpb{C}{A}\right) - \info\left(\algpov{A}{C}\right)-\info\left(\algpovpb{C}{A}\right) + \info\left(\algpov{A}{C}\wedge \algpovpb{C}{A}\right),
\end{equation}
representing the mutual information between each perspective on the other.
In the special case of different subsystems within the same QRF perspective, we recover the usual definition of mutual information. In this case, the systems are in a tensor product, their overlap is trivial and, thus, the associated information is zero while the joint algebra is simply the total algebra of observables in that perspective. Moreover, the conditional expectations onto the subalgebras associated to $\pov{A}{C}$ and $\pov{B}{C}$ are those of the usual tensor-product subsystems (\cref{ex:TP_Subsystem}). By noticing that $\info\left(\algpov{A}{C}\wedge \algpov{B}{C}\right)=D\left(\pov{\pi_{AB}}{C},\pov{\pi_{AB}}{C}\right)=0$, we see that the algebraic notion reduces to
\begin{equation}
\begin{aligned}
    \minfo\left(\pov{A}{C}:\pov{B}{C}\right)_\rho &= \info\left(\algpov{A}{C}\vee \algpov{B}{C}\right) - I \left(\algpov{A}{C}\right)  - \info\left(\algpov{B}{C}\right)\\
    &=\info\left(\pov{AB}{C}\right)-\info\left(\pov{A}{C}\right)-\info\left(\pov{B}{C}\right)\\
    &= H\left(A\right)_{\rho} + H\left(B\right)_{\rho} - H\left(AB\right)_\rho,
\end{aligned}
\end{equation}
namely, the mutual information for generalised subsystems reduces to the usual definition of the mutual information. Naturally, the same arguments apply to the mutual information $\minfo\left(\pov{B}{A}:\pov{C}{A}\right)$ between $B$ and $C$ relative to $A$.

\begin{figure}
    \centering
    \includegraphics[width=0.45\linewidth]{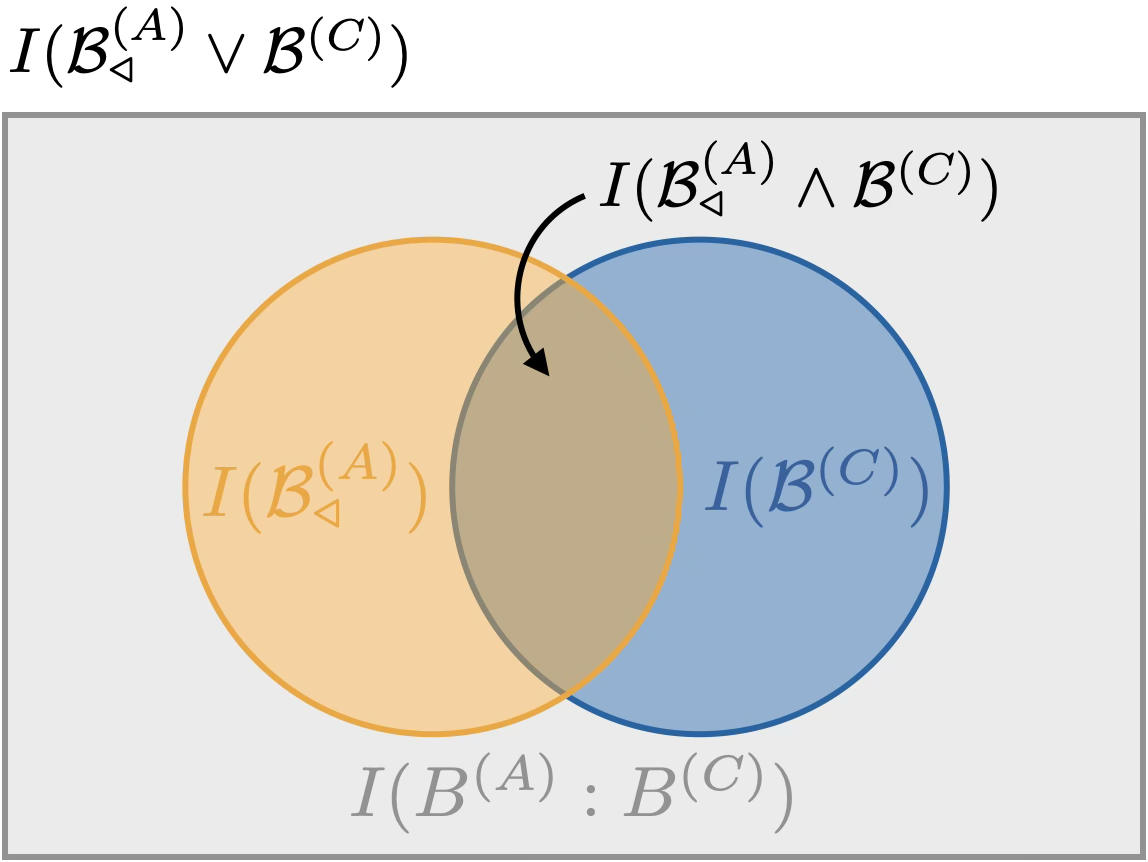}
    \caption{Information diagram illustrating the modified chain rule. The mutual information $I(B^{(A)}:B^{(C)})$ corresponds to the grey area and is obtained by subtracting from the total information $I(\algpovpb{B}{A}\vee \algpov{B}{C})$ the information of each subsystem separately, $I(\algpovpb{B}{A})$ and $I(\algpov{B}{C})$, and correcting for the overlap $I(\algpovpb{B}{A}\wedge \algpov{B}{C})$.}
    \label{fig:information_diagram}
\end{figure}

\subsection{Relation to Other Work}
\label{sec:relation_other_work}
A notion of mutual information in terms of algebraic quantities was also introduced in~\cite{Gao_2019}, in the context of entropic uncertainty relations\footnote{We thank Joe Renes for pointing out this connection through \cite{Bertoni_2020}.}:
\begin{definition}[Mutual information for algebraic subsystems (cf.~\cite{Gao_2019})]\label{DEF:GMI}
Let $\alg{A}$ and $\alg{B}$ be two algebras with conditional expectations $\cexparg{A}:\alg{A} \vee \alg{B} \mapsto \alg{A}$ and $\cexparg{B}:\alg{A} \vee \alg{B} \mapsto \alg{B}$ that satisfy the \emph{commuting square condition} $[\cexparg{A},\cexparg{B}]=0$. Then, the \emph{mutual information} in~\cite{Gao_2019} is given by
\begin{equation}
    I_{\cA}\left(A:B\right) =  \info\left(\alg{A} \vee \alg{B}\right)- \info\left(\alg{A}\right)- \info\left(\alg{B}\right) + \info\left(\alg{A} \wedge \alg{B}\right),\label{eq:mI_general}
\end{equation}
where $\info\left(\cA\right)_\rho = D\left(\cexparg{A}\left(\rho\right), \pi\right)$ denotes the information evaluated on the state associated with the algebra $\cA$. 
\end{definition}
To be precise, this quantity is termed ``conditional mutual information'' in~\cite{Gao_2019}, where the conditioning is intended on the overlap $\alg{A}\wedge\alg{B}$. Because this is better aligned with our notion of mutual information and to avoid confusion with the conditional mutual information defined in the next subsection, we will not refer to it as ``conditional'' here. A crucial element of the above definition is the commuting square condition, illustrated in~\Cref{fig:commuting_square}, whose origin can be traced back to an earlier work by Petz~\cite[Theorem 12]{Petz_1991}. In fact, the commuting square condition is equivalent to non-negativity of the mutual information in~\cref{DEF:GMI} for all states $\rho$: Petz's theorem shows that the commuting square condition implies subadditivity of the entropy and hence non-negativity of the mutual information, while the converse, namely that non-negativity for all states implies the commuting square condition, was established in~\cite{Gao_2019}.

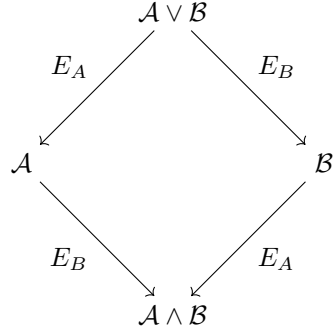
\begin{figure}
    \centering
    \begin{tikzpicture}
  \node (AB) at (0,2) {$\alg{A}\vee \alg{B}$};
  \node (A) at (-2,0) {$\alg{A}$};
  \node (B) at (2,0) {$\alg{B}$};
  \node (AwB) at (0,-2) {$\alg{A}\wedge \alg{B}$};

  \draw[->] (AB) -- node[above left]  {$\cexparg{A}$} (A);
  \draw[->] (AB) -- node[above right] {$\cexparg{B}$} (B);
  \draw[->] (A) -- node[below left]  {$\cexparg{B}$} (AwB);
  \draw[->] (B) -- node[below right] {$\cexparg{A}$} (AwB);
\end{tikzpicture}
    \caption{\emph{Commuting square condition~\cite{Gao_2019, LaRacuente_2019}.} The conditional expectations onto subalgebras $\alg{A}$ and $\alg{B}$ both lead to the same subalgebra, the overlap $\alg{A}\wedge \alg{B}$, independently of the order in which they apply: $[\cexparg{A},\cexparg{B}]=0$. This is satisfied for subsystems relative to different QRFs.}
    \label{fig:commuting_square}
\end{figure}

\Cref{DEF:GMI} is strongly connected to the modified chain rule of~\cref{def: mutual_info_QRFs} which we obtained from independent arguments concerning the overlap between the algebras of two QRF perspectives. For algebras of subsystems defined relative to different (ideal) QRFs, the commuting square condition is satisfied. Indeed, we have from~\cref{eq:CE_overlap_B}:
\begin{equation}
    \cexparg{\pov{B}{A} \wedge \pov{B}{C}} = \cexparg{\pov{B}{A}}^{\triangleleft}\circ \cexparg{\pov{B}{C}} = \cexparg{\pov{B}{C}}\circ \cexparg{\pov{B}{A}}^{\triangleleft}
\end{equation}
and, similarly from~\cref{eq:CE_overlap_AC} 
\begin{equation}
    \cexparg{\pov{A}{C}\wedge \pov{C}{A}} = \cexparg{\pov{A}{C}}\circ \cexparg{\pov{C}{A}}^{\triangleleft}= \cexparg{\pov{C}{A}}\circ \cexparg{\pov{A}{C}}^{\triangleleft}. 
\end{equation}
Thus, in our case the modified chain rule leads to a well-defined, non-negative, and symmetric measure of mutual information across the different perspectives. It is non-trivial for this condition to be satisfied. For completeness, we show a simple case where the conditional expectations of two algebras do not commute and the mutual information is negative.

\begin{example}
Consider a qubit algebra $\cA = \cB\left(\mathbb{C}^2\right)$ and the subalgebras $\cA_Z = \langle Z \rangle$ and $\cA_H = \langle H \rangle$ generated, respectively, by the Pauli-Z and Hadamard gate $H=\frac{1}{\sqrt{2}}\left(X+Z\right)$.
The conditional expectations to these two subalgebras, $\mathcal{P}_Z$ and $\mathcal{P}_H$, act on a state $\rho$ as
\begin{equation}
\begin{aligned}
    \mathcal{P}_{H}\mathcal{P}_{Z} [\rho]&= \frac{1}{4}\left(\rho + H\rho H + Z \rho Z+ H Z \rho Z H\right),\\
    \mathcal{P}_{Z}\mathcal{P}_{H} [\rho]&= \frac{1}{4}\left(\rho + H\rho H + Z \rho Z+ ZH \rho HZ\right),
\end{aligned}
\end{equation}
thus,
\begin{align}
    [\mathcal{P}_H,\mathcal{P}_Z ]\left(\rho\right)= HZ \rho ZH-ZH \rho HZ = \frac{1}{4}\left(XZ \rho - \rho XZ\right)=\frac{i}{4}[\rho, Y].
\end{align}
For all states that are not of the form $\rho = \frac{\id}{2} +\frac{r_y}{2} Y$, and thus do not commute with $Y$, the commuting square condition is violated. Thus, there is no longer a guarantee that the mutual information evaluated according to~\cref{eq:mI_general} is positive. Indeed, the intersection between the algebras is trivial, $\braket{H}\wedge \braket{Z}=\braket{\id}$, and the union gives the whole algebra $\cA$, $\langle H\rangle \vee \langle Z\rangle = \cA$, thus for the state $\ketbra{0}$ we have:
\begin{equation}
    \begin{aligned}
        I\left(\cA_H\vee \cA_Z\right) =& D\left(E_{\cA}\left(\ketbra{0}\right), \frac{\id}{2}\right) =D\left(\ketbra{0}, \frac{\id}{2}\right)  \\
        I\left(\cA_H\wedge \cA_Z\right) =& D\left(E_{\langle \id\rangle}\left(\ketbra{0}\right), \frac{\id}{2}\right) =D\left(\frac{\id}{2}, \frac{\id}{2}\right) = 0\\ 
        I\left(\cA_H\right) =&  D\left(E_{\cA_H}\left(\ketbra{0}\right), \frac{\id}{2}\right) =D\left(\frac{\ketbra{0}+\ketbra{+}}{2}, \frac{\id}{2}\right) \\ 
        I\left(\cA_Z\right) =& D\left(E_{\cA_Z}\left(\ketbra{0}\right), \frac{\id}{2}\right) =D\left(\ketbra{0}, \frac{\id}{2}\right)  
    \end{aligned}
\end{equation}
Thus, 
\begin{align*}
    I_\mathcal{A}\left(\cA_Z :\cA_H\right) = -D\left(\frac{\ketbra{0}+\ketbra{+}}{2}, \frac{\id}{2}\right)\approx -0.4 < 0.
\end{align*}

Notice that for a qubit a pinching in orthogonal bases, such as $\mathcal{P}_X$ and $\mathcal{P}_Z$, satisfies the commuting square condition while non-orthogonal bases, such as $\mathcal{P}_X$ and $\mathcal{P}_H$, lead to its violation (see \cref{fig:bloch}).
\end{example}

\begin{figure}
    \centering
    \begin{subfigure}[t]{0.49\textwidth}
    \resizebox{0.75\textwidth}{0.75\textwidth}{
        \input{fig7a}

    }
         \caption{Commuting square condition satisfied.}
    \label{fig:bloch_cs}
    \end{subfigure}
    \begin{subfigure}[t]{0.49\textwidth}
        \resizebox{0.75\textwidth}{0.75\textwidth}{
        \input{fig7b}

    }
         \caption{Commuting square condition violated.}
    \label{fig:bloch_ncs}
    \end{subfigure}
\caption{\emph{Commuting square condition (CSC) for a qubit.} (a) The projectors $\mathcal{P}_X$ and $\mathcal{P}_Z$ are orthogonal and, as a consequence, satisfy the CSC. (b) If one considers pinching in two non-orthogonal bases, however, the CSC is violated.}
\label{fig:bloch}
\end{figure}
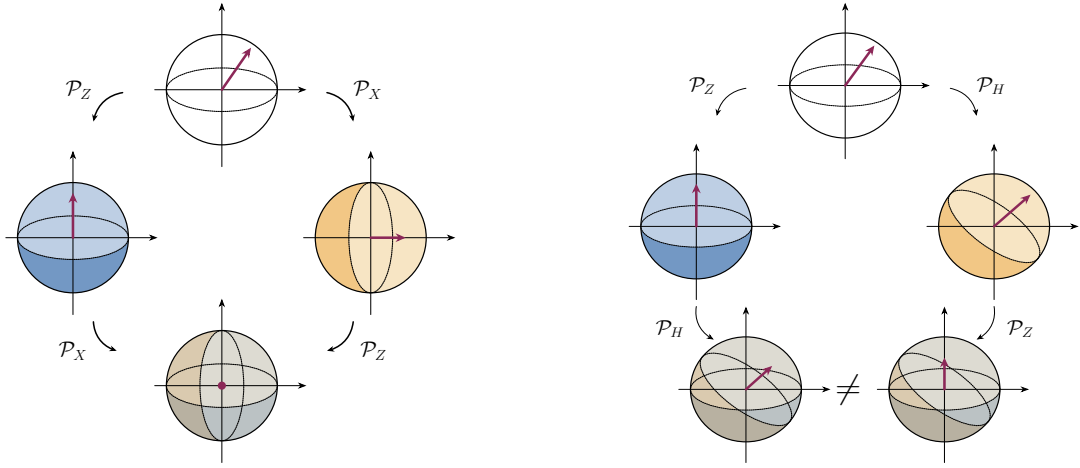

\subsection{Conditional Mutual Information and Strong Subadditivity}
\label{sec:cond_MI}
In addition to quantum mutual information, there is a variety of quantum information measures, such as entropy, conditional entropy, and conditional mutual information. In this section we extend some of these notions to algebraic subsystems. We start from the usual von Neumann quantities~\cite{Cerf_1998, Renes_2022}, but incorporate the insight of~\cref{sec:modified_chain_rule} that, for algebraic subsystems, the information associated to the overlap is often overcounted. 
To extend the notion of conditional mutual information, we first define the measure of conditional information.

\begin{definition}[Conditional information for algebraic subsystems]
    Consider an algebra $\mathcal{M}$ of bounded operators on a Hilbert space $\cH$ of dimension $d$. Given the subalgebras $\alg{A}, \alg{B}\subseteq \mathcal{M}$  and an overall state $\rho\in \mathcal{S}\left(\cH\right)$, the \emph{conditional information} of $\alg{B}$ given $\alg{A}$ is given by
    \begin{equation}
        \info\left(\alg{B}|\alg{A}\right)_\rho \coloneqq \info\left(\alg{A} \vee\alg{B}\right)_\rho - \info\left(\alg{A}\right)_\rho.
    \end{equation}
\end{definition}
This definition of conditional information can be understood as removing all the information related to the algebra of $\alg{A}$ from the global algebra $\alg{B} \vee \alg{A}$. Notice that this quantity is never associated to a subalgebra since the conditioning removes the identity operator, which is then not part of $\{\alg{B}\setminus\alg{A}\}$.  If there is an overlap between $\alg{A}$ and $\alg{B}$, the information associated to the overlapping part of $\alg{B}$ with $\alg{A}$ is removed. 

In the case of a fixed TPS, the conditional information is equal to 
\begin{equation}
    \info\left(B|A\right)_\rho= \info\left(AB\right)-\info\left(A\right)= D\left(\rho_{AB}, \rho_A \otimes \pi_B\right).
\end{equation}
It can be easily seen that, like the notion of information, this measure is related to conditional entropy by the relation $H\left(B|A\right)_\rho=\info\left(B|A\right)_\rho + \log d_B$. In this definition, there is no risk of overcounting the overlap between the algebras, unlike in the definition of mutual information, as we are only subtracting the information associated to one and not both subsystems.

In light of extending conditional mutual information to an algebraic setting, we first notice that, in the usual tensor product case, it is possible to write this measure in terms of information quantities as follows:
\begin{equation}
\label{eq:conditional_I_info}    \info\left(A:B|C\right)=H\left(A|C\right)+ H\left(B|C\right)-H\left(AB|C\right)=\info\left(AB|C\right)-\info\left(A|C\right)-\info\left(B|C\right),
\end{equation}
where we denoted with letters $A$, $B$ and $C$, distinct tensor product subsystems of a joint Hilbert space.  When extending such a definition to algebraic subsystems, we need to carefully account for possible overlaps of the subalgebras. For arbitrary algebras $\alg{A}$, $\alg{B}$ and $\alg{C}$ there is an overlap between $\info\left(\alg{A}|\alg{C}\right)=\info\left(\alg{A}\vee\alg{C}\right)-\info\left(\alg{C}\right)$ and $\info\left(\alg{B}|\alg{C}\right)=\info\left(\alg{B}\vee\alg{C}\right)- \info\left(\alg{C}\right)$ which we should not overcount. The overlapping information between $\info\left(\alg{A}\vee\alg{C}\right)$ and $\info\left(\alg{B}\vee \alg{C}\right)$ is $\info\left((\alg{A}\vee\alg{C})\wedge (\alg{B}\vee\alg{C})\right)$. This overlap contains the algebra $\alg{C}$, and can be strictly larger than it for a general algebraic case. This can be seen in cases where $\alg{A}$ and $\alg{B}$ share elements that do not belong to $\alg{C}$ (see \cref{fig:overlap_condMI}). When considering the overlap between $\info\left(\alg{A}|\alg{C}\right)$ and $\info\left(\alg{B}|\alg{C}\right)$, the information associated with $\alg{C}$ is not overcounted, however, the sum of the two overcounts the overlap between $\alg{A}|\alg{C}$ and $\alg{B}|\alg{C}$, which exists if $\alg{A}$ and $\alg{B}$ overlap non-trivially (see Fig.\,\ref{fig:overlap_condMI} for a graphical representation of this). To account for this, we therefore subtract the information contained in $\alg{C}$ from $\info((\alg{A}\vee\alg{C})\wedge (\alg{B}\vee\alg{C}))$. We define 
\begin{equation}
    \info\left(\alg{A}|\alg{C}\wedge \alg{B}|\alg{C}\right)\coloneqq \info\left((\alg{A}\vee\alg{C})\wedge (\alg{B}\vee\alg{C})\right) -\info\left(\alg{C}\right),
\end{equation}
as the final relevant modification to~\cref{eq:conditional_I_info}.
The considerations above lead us to define conditional mutual information as follows.
\begin{figure}[h]
    \centering
    \input{fig9}
    \caption{Overlap between the algebras $\alg{A} \vee \alg{C}$ and $\alg{B} \vee \alg{C}$. This intersection always contains the algebra $\alg{C}$, but may be be strictly larger, particularly when $\alg{A}$ and $\alg{B}$ have non-trivial overlap.}
    \label{fig:overlap_condMI}
\end{figure}
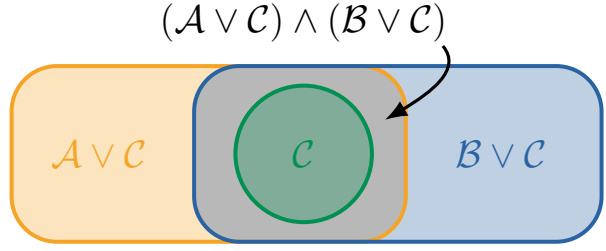

\begin{definition}[Conditional mutual information for algebraic subsystems]\label{CMIfAS}
Consider an algebra $\mathcal{M}$ of bounded operators on a Hilbert space $\cH$ of dimension $d$. Given the subalgebras $\alg{A}, \alg{B},\alg{C}\subseteq \mathcal{M}$  and an overall state $\rho\in \mathcal{S}\left(\cH\right)$, the conditional mutual information between $\alg{A}$ and $\alg{B}$, given the information of $\alg{C}$, is given by
    \begin{equation}
        \minfo\left(\alg{A}:\alg{B}|\alg{C}\right)\coloneqq\info\left(\alg{A} \vee \alg{B}|\alg{C}\right)-\info\left(\alg{A}|\alg{C}\right)-\info\left(\alg{B}|\alg{C}\right)+ \info\left(\alg{A}|\alg{C}\wedge \alg{B}|\alg{C}\right) ,
    \end{equation}
    where $\info\left(\alg{A}|\alg{C}\wedge \alg{B}|\alg{C}\right)\coloneqq \info\left(\left(\alg{A}\vee \alg{C}\right)\wedge \left(\alg{B}\vee \alg{C}\right)\right)-\info\left(\alg{C}\right)$.
\end{definition}

Notice that in the case of a fixed TPS, the algebra $(\alg{A}\vee\alg{C})\wedge (\alg{B}\vee\alg{C})$ coincides with $\alg{C}$. Thus, \cref{CMIfAS} reduces to the standard conditional mutual information. Furthermore, using the chain law of conditional information, all the terms capturing the information of $\alg{C}$ cancel out, giving
\begin{equation}\label{SSA_P1}
   \minfo\left(\alg{A}:\alg{B}|\alg{C}\right)\coloneqq\info\left(\alg{A} \vee \alg{B}\vee\alg{C}\right)-\info\left(\alg{A}\vee\alg{C}\right)-\info\left(\alg{B}\vee\alg{C}\right)+ \info\left(\left(\alg{A}\vee\alg{C}\right)\wedge \left(\alg{B}\vee\alg{C}\right)\right)= \minfo\left(\alg{A} \vee\alg{C} :\alg{B}\vee\alg{C}\right).
\end{equation}
We notice that through this approach of defining information measures, conditional mutual information is a type of mutual information. This immediately implies that $I(\mathcal{A}:\mathcal{B}|\mathcal{C})$ satisfies strong subadditivity as long as the subsystems $\mathcal{A}\vee \mathcal{C}$ and $\mathcal{B} \vee \mathcal{C}$ satisfy the commuting square condition. 
\begin{proposition}[Strong subadditivity] Given a Hilbert space $\cH$ of dimension $d$, consider three subalgebras $\alg{A}, \alg{B},\alg{C}\subseteq \cB(\cH)$. Then 
$\minfo\left(\alg{A}:\alg{B}|\alg{C}\right)_{\rho} \ge 0$ for all states $\rho\in \mathcal{S}(\cH)$ if and only if the algebras $\alg{A} \vee\alg{C}$ and $ \alg{B} \vee\alg{C} $ satisfy the commuting square relation (see \cref{DEF:GMI}). 
\end{proposition}
\begin{proof}
    We have established in \cref{SSA_P1} that $\minfo\left(\alg{A}:\alg{B}|\alg{C}\right)$ equals $\minfo\left(\alg{A} \vee\alg{C} :\alg{B}\vee\alg{C}\right)$. It follows from~\cite{Gao_2019} that $\minfo\left(\alg{A} \vee\alg{C} :\alg{B}\vee\alg{C}\right)\ge0$ if and only if the algebras $\alg{A} \vee\alg{C}$ and $ \alg{B} \vee\alg{C} \subseteq \alg{A} \vee\alg{B}\vee\alg{C}$ satisfy the commuting square relation. Hence this condition applies as well to the  conditional mutual information for algebraic subsystems.
\end{proof}

\section{Examples and Application}
\label{sec:examples_app}

In this section, we consider some concrete examples and applications of the theoretical insights gained in the previous sections. In~\cref{sec:non_entangling}, we provide physical intuition for the mutual information across different QRFs by considering a concrete example: a product state that does not entangle under a change of QRF. Moreover, we show how the algebraic subsystems and their associated information provide a natural interpretation of previously established conservation laws under QRF transformations \cite{cepollaro2024,delaHamette_entropy_2026, Luppi_2026} in~\cref{sec:conservation_laws}. 

\subsection{Example: Non-Entangling States}
\label{sec:non_entangling}

To gain some intuition for the proposed notion of mutual information across perspectives, it will be helpful to determine its value for concrete examples. A particularly interesting class of states is given by those that do not entangle under a change of QRF. For example, in the case of $G=\mathbb{Z}_2$, states of the form
\begin{equation}
    \pov{\ket{\psi}_{AB}}{C}= \frac{1}{\sqrt{2}}\left(\alpha \ket{0}_A^{\left(C\right)}+\beta \ket{1}_A^{\left(C\right)}\right)\otimes \pov{\ket{\pm}_B}{C}
\end{equation}
retain their product-form under a change of QRF, which, at most, causes a phase-kickback depending on the value of $X_B$:
\begin{equation}
    \pov{\ket{\psi}_{BC}}{A}= S_{C\to A}\left(\pov{\ket{\psi}_{AB}}{C} \right)  =   \pov{\ket{\pm}_B}{A}\otimes \frac{1}{\sqrt{2}}\left(\alpha \ket{0}_C^{\left(A\right)}\pm\beta \ket{1}_C^{\left(A\right)}\right).
\end{equation}
More generally, we can characterise states that do not entangle under a change of QRF by the following definition. 
\begin{definition}[Non-entangling states]\label{defNES}
Consider a tripartite system $A,B,C$ with a pure product state of $A$ and $B$ in the perspective of $C$, $\ket{\Psi}_{AB}^{\left(C\right)} = \ket{\psi}_A^{\left(C\right)} \otimes \ket{\phi}_B^{\left(C\right)}$. This state is called \emph{non-entangling} if for all $s, s' \in G $ such that  $\braket{s|\psi}_A^{\left(C\right)} \ne 0 $ and $\braket{s'|\psi}_A^{\left(C\right)} \ne 0 $, it holds
\begin{equation}\label{4.1} e^{i\theta\left(s\right)} U_B^{\dagger}\left(s\right) \ket{\phi} = e^{i\theta(s')} U_B^{\dagger}\left(s'\right) \ket{\phi},
\end{equation}
where $\theta\left(s\right) \in [0, 2\pi) \subset \mathbb{R}$ is an arbitrary real-valued function and $U_B\left(s\right)$ is the regular left representation of the group $G$ acting on the vector representation of itself.
\end{definition}
These states are closely related to the so-called ``TPS-invariant operators'' of \cite{Hoehn_2023} but provide a more concrete criterion that can be checked easily. Moreover, both coherence and entanglement entropy of these states  are preserved under changes of QRF (see \cref{Non_Ent_Sta_App}).

For states of this type, it would be natural to think that the subsystem $B$ remains the same under a change of QRF -- after all, its state factorises from the state of $A$ in both frames of reference. Yet, this state-based thinking can be misleading. In particular, it would lead one to conclude that the joint system of $\pov{B}{A}$ and $\pov{B}{C}$ is simply given by $B = \pov{B}{A}=\pov{B}{C}$.  When understanding subsystems in terms of observational capabilities and the associated subalgebras of observables, we established that the algebra $\algpovpb{B}{A} \vee \algpov{B}{C}$ (relative to $C$) contains not only the observables on system $B$ but also the orientation $\pov{Z_A}{C}$ of the other reference frame. This has an important effect on the mutual information associated to the state.
\begin{example}
    Consider $G=\mathbb{Z}_2$ and the following pure state relative to $C$:
    \begin{equation}
        \pov{\ket{\psi}_{AB}}{C}=\pov{\ket{0}_A}{C} \otimes\pov{\ket{+}_B}{C}
    \end{equation}
    This state is non-entangling and, in fact, entirely invariant under the change of QRF:
    \begin{equation}
        \pov{\ket{\psi}_{BC}}{A}= \pov{\ket{+}_B}{A}\otimes \pov{\ket{0}_C}{A} .
    \end{equation}
    The mutual information of $\pov{B}{A}$ and $\pov{B}{C}$ associated to  $\rho^{(A)}_{BC} = \ketbra{\psi}{\psi}^{(A)}_{BC}$ is 
    \begin{align*}
         \minfo\left(\pov{B}{C}:\pov{B}{A}\right) = D\left( \mathcal{P}_{\pov{Z_A}{C}}[\pov{\rho_{AB}}{C}], \pov{\pi_{AB}}{C}\right) 
         -D\left(\pov{\rho_B}{C}, \pov{\pi_B}{C}\right) 
         -D\left(\pov{\rho_B}{A}, \pov{\pi_B}{A}\right) +D\left(\mathcal{P}_{\pov{X_B}{C}}[\pov{\rho_B}{C}], \pov{\pi_B}{C}\right).
    \end{align*}
    Computing explicitly the pinched and reduced states within this definition, we obtain
    \begin{equation}
    \begin{split}
             \minfo\left(\pov{B}{C}:\pov{B}{A}\right)
             =\underset{=2}{\underbrace{D\left(\pov{\ketbra{0}{0}_{A}}{C}\otimes\pov{\ketbra{+}{+}_{B}}{C}, \pov{\pi_{AB}}{C}\right)}} 
        & -\underset{=1}{\underbrace{D\left(\pov{\ketbra{+}{+}_B}{C}, \pov{\pi_B}{C}\right)}}+ \\
         &
         -\underset{=1}{\underbrace{D\left(\pov{\ketbra{+}{+}_B}{A}, \pov{\pi_B}{A}\right)}} +\underset{=1}{\underbrace{D\left(\pov{\ketbra{+}{+}_B}{C}, \pov{\pi_B}{C}\right)}}
         =1.
    \end{split}
    \end{equation}
    The mutual information being non-zero highlights that the global algebra $\algpovpb{B}{A} \vee \algpov{B}{C}$ entails access to additional observables, such as $\pov{Z_A}{C} \otimes \pov{\id_{B}}{C}$, which provides information about $\pov{\ket{\psi}_{AB}}{C}$ that cannot be obtained by operations in $\algpovpb{B}{A}$ and $\algpov{B}{C}$ separately.
\end{example}

In the concrete setup, illustrated in~\Cref{fig:pull-back_algebras}, we can imagine that Charlie measures the orientation of $B$ relative to his frame of reference via $\pov{\id_A}{C}\otimes \pov{Z_B}{C}$ while Alice measures the orientation of $B$ relative to her own frame of reference through $\pov{Z_A}{C}\otimes \pov{Z_B}{C}$ (as seen from Charlie's point of view). On their own, with access to $\algpov{B}{C}$ respectively $\algpovpb{B}{A}$, none of them can infer the orientation of the other frame of reference. Yet, when they combine their resources, they can access $\pov{Z_A}{C}\otimes \pov{\id_{B}}{C}=\left(\pov{\id_A}{C}\otimes \pov{Z_B}{C}\right)\left(\pov{Z_A}{C}\otimes \pov{Z_B}{C}\right)$ and thus the relative orientation between their frames.
If the relative orientation of the reference frames is maximally mixed, there is no information to be gained about this quantity. Hence, in this situation, we expect that the mutual information should be zero. We now show that this is the case.
\begin{example}
    Consider $G=\mathbb{Z}_2$ and the following state relative to $C$:
    \begin{equation}
        \pov{\rho_{AB}}{C}=\frac{1}{2}\pov{\id_A}{C} \otimes\pov{\ketbra{+}{+}_B}{C}.
    \end{equation}
    Under a change of QRF, this becomes
    \begin{equation}
        \pov{\rho_{BC}}{A}= \pov{\ketbra{+}{+}_B}{A}\otimes \frac{1}{2}\pov{\id_C}{A} .
    \end{equation}
    The mutual information of $\pov{B}{A}$ and $\pov{B}{C}$ associated to this state is
    \begin{equation}
    \begin{split}
         \small \minfo\left(\pov{B}{C}:\pov{B}{A}\right)=\underset{=1}{\underbrace{D\left(\pov{\pi_A}{C}\otimes \pov{\ketbra{+}{+}_B}{C}, \pov{\pi_{AB}}{C}\right)}} 
         &-\underset{=1}{\underbrace{D\left(\pov{\ketbra{+}{+}_B}{C}, \pov{\pi_B}{C}\right)}} 
         \\
         &-\underset{=1}{\underbrace{D\left(\pov{\ketbra{+}{+}_B}{A}, \pov{\pi_B}{A}\right)}} +\underset{=1}{\underbrace{D\left(\pov{\ketbra{+}{+}_B}{C}, \pov{\pi_B}{C}\right)}}
         =0.
    \end{split}
    \end{equation}
\end{example}

\subsection{Conservation Laws Under Changes of QRF}
\label{sec:conservation_laws}

Understanding the algebraic and information-theoretic structure of subsystems relative to different QRFs also allows us to give a clear interpretation to recently discovered conserved quantities under changes of QRF. In particular, it was shown in \cite{cepollaro2024} that the sum of coherence and entanglement for pure states is conserved under changes of ideal QRFs:
\begin{equation}
    \pov{\mathcal{C}_e}{C}+\pov{\mathcal{E}_e}{C}=\pov{\mathcal{C}_e}{A}+\pov{\mathcal{E}_e}{A},
\end{equation}
where $\pov{\mathcal{E}_e}{C} \coloneqq H\left(\pov{A}{C}\right)_{\rho}= -\Tr\left[\pov{\rho_A}{C} \log\pov{\rho_A}{C}\right]$ is the entanglement entropy and $\pov{\mathcal{C}_e}{C}  \coloneqq H\left(\pov{A}{C}\right)_{\rho} - H\left(\pov{A}{C}\right)_{\rho_{d}}$ is the relative entropy of coherence of $A$ relative to $C$. Here, $\pov{\rho_{d,A}}{C} = \mathcal{P}_{G,A}\left(\pov{\rho_A}{C}\right)$ is the state of $A$ upon removing all off-diagonal terms in the group basis. The above conservation law then follows from the fact that $H\left(\pov{A}{C}\right)_{\rho_d} = -\Tr\left[\pov{\rho_{d,A}}{C}\log \pov{\rho_{d,A}}{C}\right]$ is conserved under changes of QRF.  Generalisations of this result to an arbitrary number of subsystems and the non-ideal case were subsequently studied in \cite{delaHamette_entropy_2026}.\\

Now, the state $\pov{\rho_{d,A}}{C} = \mathcal{P}_{G,\pov{A}{C}}\left(\pov{\rho_A}{C}\right)$ is already familiar to us from the study of the algebraic subsystems relative to different QRFs. In particular, as shown in \cref{eq:CE_overlap_AC}, projecting onto the overlap between subsystems $\pov{A}{C}$ and $\pov{C}{A}$ from Charlie's point of view gives
\begin{equation}
    \cexparg{\pov{A}{C}\wedge \pov{C}{A}}\left(\rho\right) = \mathcal{P}_{G,\pov{A}{C}}\left(\rho\right)\otimes \tfrac{1}{|G|}\pov{\id_{B}}{C}.
\end{equation}
Given that this state lies in the overlap of the two perspectives, it is naturally preserved under changes of QRFs. As such, the information associated to it,
\begin{equation}
    \info\left(\algpov{A}{C}\wedge \algpovpb{C}{A}\right) = D\left(\mathcal{P}_{G,\pov{A}{C}}\left[\rho\right]\otimes \tfrac{1}{|G|}\pov{\id_{B}}{C}, \pov{\pi_{AB}}{C}\right) = D\left(\mathcal{P}_{G,\pov{A}{C}}\left[\rho\right], \pov{\pi_{A}}{C}\right) = \log d_A - H\left(A\right)_{\pov{\rho_{d,A}}{C}}.
\end{equation}
is invariant, too. Thus, the sum of entanglement and coherence really is the \emph{entropy associated to the overlap of $\pov{A}{C}$ and $\pov{C}{A}$} and, as such, naturally preserved under changes of QRF.\footnote{Note that we are using here only the definition of entropy for standard tensor-product subsystems. As discussed in~\cref{sec:algebraic_subsystems}, the entropy associated to the algebraic subsystems additionally tracks the maximal entropy of system $B$.}\\

This reasoning also suggests a second conserved quantity associated to the other non-trivial overlap between subsystems relative to different QRFs: $\algpovpb{B}{A}\wedge \algpov{B}{C}$. As shown in \cref{eq:CE_overlap_B}, the conditional expectation onto this subalgebra is given by $\cexparg{\pov{B}{A} \wedge \pov{B}{C}}\left(\pov{\rho_{AB}}{C}\right)= \tfrac{1}{|G|}\pov{\id_A}{C}\otimes \mathcal{G}_{\inc}\left(\pov{\rho_B}{C}\right)$. The information associated to it is given by
\begin{equation}\label{Entropy_BA_BC}
    \info\left(\algpovpb{B}{A}\wedge \algpov{B}{C}\right) =  D\left(\mathcal{G}_{\inc}\left[\pov{\rho_B}{A}\right],\pov{\pi_B}{A}\right) = D\left(\mathcal{G}_{\inc}\left[\pov{\rho_B}{C}\right],\pov{\pi_B}{C}\right)= \log d_B - H\left(B\right)_{\mathcal{G}_{\inc}[\rho]}.
\end{equation}
Thus, we obtain another conserved quantity -- the information or entropy associated to the incoherent $G$-twirl of $B$. This quantity is related to the conservation law proven in \cite[Th.~3.8]{Luppi_2026}. There, it was shown that the sum of coherence and entanglement of the system $B$, 
\begin{equation}\label{Thm_3.8_Luppi}
    \pov{\mathcal{C}_{e,B}}{C}+\pov{\mathcal{E}_B}{C}=\pov{\mathcal{C}_{e,B}}{A}+\pov{\mathcal{E}_B}{A},
\end{equation}
where $C_{e,B} = H(B)_{\rho_d} - H(B)_{\rho}$ and $\rho_{B,d}$ = $\sum_n P_n \rho_{B}P_n$, is preserved if the reduced QRF transformation on system $B$ is \emph{strongly block-preserving} with respect to the set of non-degenerate PVMs $\{P_n\}_n$. This condition essentially requires that the QRF transformation leads to at most a relabeling of the projectors. As before, the conservation law in Eq.~\eqref{Thm_3.8_Luppi} captures the conservation of the entropy associated to the pinched state $\rho_{B,d}$.
To see how this relates to the entropy associated to the overlap, note that in the abelian case, all representation operators $U(g)$ pairwise commute and therefore admit a common orthonormal eigenbasis. Accordingly, the incoherent $\mathcal G$-twirl can be expressed as a pinching in the charge sectors, $\mathcal G_{\mathrm{inc}}(\rho)=\sum_q P_q\rho P_q$ \cite{bartlett2007reference}, where the projectors $\{P_q\}$ are non-degenerate. Moreover, the projectors onto the different charge sectors are invariant under a change of QRF, so that the block-preserving condition is satisfied. We can thus see the conservation of the information associated to the overlap $\algpovpb{B}{A}\wedge \algpov{B}{C}$ as a special case of \cref{Thm_3.8_Luppi} for abelian groups.

In the non-abelian case, the incoherent G-twirl no longer takes this simple form \cite{bartlett2007reference}. Nevertheless, we can also define a conservation law for this case, using the concept of \textit{$G$-asymmetry} \cite{Frameness}, which is defined as $A_G[\rho]\coloneqq H(\mathcal{G}_{\inc}[\rho])-H(\rho)$. This quantity is usually described as a measure of \textit{frameness} and represents the entropy increase resulting from erasing information about a state's relation to a QRF. By considering that, for a bipartite pure state $\pov{\rho_{AB}}{C}$, the equality $H(\pov{B}{C})=H(\pov{A}{C})=\pov{\mathcal{E}_e}{C}$ always holds, and that, by \cref{Entropy_BA_BC}, $H\left(\pov{B}{C}\right)_{\mathcal{G}_{\inc}[\rho]}=H\left(\pov{B}{A}\right)_{\mathcal{G}_{\inc}[\rho]}$ is invariant under the choice of QRF, we obtain
\begin{equation}
    \pov{A_G}{C}+ \pov{\mathcal{E}_e}{C}= \pov{A_G}{A}+ \pov{\mathcal{E}_e}{A},
\end{equation}
where $\pov{A_G}{C}\coloneqq A_G[\pov{\rho_B}{C}]$. This law shows a trade-off between entanglement and the relative information (frameness) held by the system when changing QRF.

\section{Conclusion}
\label{sec:conclusion}

Due to the relativity of subsystems, we cannot directly apply the standard notion of mutual information to quantify the correlations between subsystems relative to different quantum reference frames. In this paper, we showed how following an algebraic approach and taking into account the overlap between different subsystem algebras nevertheless leads to a well-defined notion of mutual information across QRF perspectives for ideal QRFs and arbitrary finite-dimensional symmetry groups. This notion recovers the usual definition of mutual information for subsystems relative to the same frame of reference, is non-negative for arbitrary states, and relates naturally to other proposals in the literature.

In order to derive this measure, we studied in detail the various subsystem algebras relative to different frames of reference, including their union and their overlap. We found that these algebras satisfy a particularly convenient property: the projectors onto different subsystem algebras commute. This ``commuting square condition'' \cite{Gao_2019} allows to separate the information associated to the overlap of two subsystem algebras from the remainder unambiguously and is what ensures a well-motivated and non-negative measure of mutual information across QRFs. This  suggests a particular type of compatibility between the perspectives associated to different QRFs. We extended our algebra-based strategy to define additional quantum information measures, namely conditional information and conditional mutual information, and showed that the latter satisfies strong-subadditivity for subsystems relative to different QRFs. Finally, we considered the implications of our definition for a special class of states, which do not entangle under a change of QRF, and explained how recently proposed conservation laws \cite{cepollaro2024, delaHamette_entropy_2026, Luppi_2026} reflect the fact that the information associated to the overlap between subsystems relative to different QRFs remain unaffected by a change of reference frame.

The present work suggests many routes for further investigation. Staying within the field of QRFs, it would be interesting to extend the analysis to non-ideal frames \cite{delaHamette_2021, Garmier_2025_non_ideal} and perhaps even systems relative to QRFs modelled by different frameworks. The latter can be captured quite naturally within the algebraic approach, as different approaches to modelling QRFs can be understood to reflect different assumptions about which algebras of observables are accessible within a given perspective \cite{Castro-Ruiz_2025, delaHamette_2025_charges, Doat_2025}. In this context, it would also be interesting to understand better the operational meaning of measuring in a given QRF perspective \cite{Vanrietvelde_2026}. 

From a quantum information perspective, the next natural step is to provide our measure with operational meaning. Indeed, the introduction of entropic quantities is commonly justified as figure of merit of an information processing task or protocol. While the definitions in this work are motivated by the intuition behind their structure, an operational interpretation a crucial next step. A further open question is whether this quantity satisfies usual compositional properties. For instance, the chain rule $I(A:BC) = I(A:B)+I(A:C|B)$, states that correlations between $A$ and the joint system $BC$ can be accounted for sequentially: first those between $A$ and $B$ and then those between $A$ and $C$ given $B$. Whether such decomposition holds, possibly as an inequality, across QRF perspective is not obvious, since the joint system $BC$ may include overlaps between $B$ and $C$. More generally, it would be valuable to relate our measure to known entropic quantities and examine whether information diagrams can jointly summarize their relations.  
In~\cite{Rothlin_2026,carrozza_2025}, a concrete correspondence between the perspective neutral formalism to QRFs and quantum error correction (QEC) was formulated. This connection suggests that the algebraic mutual information can be carried over to QEC, for instance to quantify correlations between logical subsystems associated with different codes on the same space. Thus, our measures might prove useful for quantitative comparison between quantum error correction codes.

Finally, this work has stayed firmly within the realm of finite-dimensional and non-relativistic quantum information theory. Going beyond the first restriction to infinite-dimensional symmetry groups is a realistic next step, although one might have to accept certain divergences in the information measures in this context -- after all, a perfectly localised position reference frame contains infinite information. Moving towards the relativistic regime, even more care has to be taken to avoid divergences. Here, it would be interesting to see whether one can utilise the results on type reduction through the inclusion of a QRF \cite{Witten_2021, Chandrasekaran_2022, Fewster_2024, deVuyst_2024a, DeVuyst_2024b}, which have already been shown to cure divergences in the entropy associated to local algebras of observables in spacetime. Combining these results with the strategy outlined in the present work might allow one to associate more general information measures, such as mutual information, to local and perhaps even overlapping subsystems in spacetime.

\section{Acknowledgments}
All authors would like to thank the QIT group at ETH, in particular Joseph Renes, Lorenzo Sannino, and Renato Renner for insightful discussions. In addition, we thank Renato Renner for providing feedback on the thesis. We also acknowledge the assistance of Claude Sonnet 5.5 and Opus 5 -- 5.5 in developing and streamlining some of the proofs for arbitrary finite-dimensional symmetry groups. This research was funded in part by the
Austrian Science Fund (FWF) 10.55776/PIN1152626.  CF and VK acknowledge support through the SNSF grant No.~200021E\_232425, from ETH Zurich Quantum Center and NCCR SwissMAP. FG acknowledges support from the Swiss National Science Foundation via the Ambizione Grant PZ00P2-208885, from the ETH Zurich Quantum Center, and from the Italian Ministry of University and Research (MUR) under the Funding for Research Projects (FIS 2), pursuant to Ministerial Decree No. 23314 of 11/12/2024 (Project Q-GraSp, FIS-2023-02629). 
This work was made possible through the support of the WOST, WithOut SpaceTime project (\href{https://withoutspacetime.org}{withoutspacetime.org}), led by the Center for Spacetime and the Quantum (CSTQ), and supported by Grant ID\#~63683 from the John Templeton Foundation (JTF). The opinions expressed in this work are those of the author(s) and do not necessarily reflect the views of the John Templeton Foundation or those of the Center for Spacetime and the Quantum. This work contributes to the European Union COST Action CA23130 \textit{Bridging high and low energies in search of quantum gravity} and to the European Union COST Action CA23115 \textit{Relativistic Quantum Information}.

\section{Author Contributions}
LV developed the main results of this work as part of his master’s thesis, supervised by VK. The thesis  grew from an earlier project by LV, supervised by FG and VK, some of whose results are presented in Sec.~V; this earlier work provided a starting point for the present project. FG, VK, and LV discussed the general direction of the work, whose development was led by VK together with LV, with FG advising, and CF taking part in some discussions during LV’s thesis. Some aspects of the algebraic approach build on ideas CF had developed in a related work, including a semester project she supervised. CF and FG contributed complementary expertise, on quantum information measures and on QRFs respectively. All authors contributed to the development and interpretation of the results and reviewed and revised the manuscript.

\bibliography{bibliography}
\bibliographystyle{apsrev4-2.bst}

\newpage
\singlespacing

\appendix
\addappheadtotoc
\addtocontents{toc}{\protect\setcounter{tocdepth}{1}}

\crefalias{section}{appendix}
\crefalias{subsection}{appendix}

\section{Subsystem Relativity}
\label{app:subsystem_relativity}

Here, we give the formal justification for the intuitive explanation of subsystem relativity provided in~\cref{sec:algebraic_subsystems} within the framework of the perspective-neutral approach, following \cite{Ahmad_2021,Hoehn_2021, Hoehn_2023}, and highlighting differences to the extra-particle approach throughout. In both cases, the starting point is the so-called \emph{kinematical Hilbert space} associated to an external point of view. In the perspective-neutral approach, the external frame is a mathematical tool that is useful to recover the internal perspectives; in the extra-particle approach, however, its role is more nuanced, and it can be either a physical frame in which one can make operations, or an abstract mathematical structure.  When considered from an outside perspective, we can distinguish all orientations perfectly and describe the qubits in terms of the Hilbert space $\cH_{\textup{kin}}= \cH_A\otimes \cH_B\otimes \cH_C$
and $\cH_X = \mathbb{C}^2$ for $X\in \{A,B,C\}$. At this level, we can associate to each of the subsystems an algebra of observables that an external observer could perform on said subsystem,
\begin{equation}
\begin{aligned}
    \algkin{A} &= \langle X_A\otimes\mathbb{I}_{BC}, Z_A\otimes \mathbb{I}_{BC}\rangle,\\
    \algkin{B} &= \langle \mathbb{I}_A\otimes X_B\otimes\mathbb{I}_{C}, \mathbb{I}_A\otimes Z_B\otimes \mathbb{I}_{C}\rangle,\\
    \algkin{C} &= \langle \mathbb{I}_{AB}\otimes X_C, \mathbb{I}_{AB}\otimes Z_C\rangle.
\end{aligned}
\end{equation}
In the presence of symmetries, we no longer have access to all of these observables. The precise algebra of accessible observables depends on the particular approach to QRFs. In the quantum-information approaches \cite{bartlett2007reference, Castro_Ruiz_2021}, one considers all invariant observables. The algebra of these observables is obtained from the total algebra of observables through the so-called \emph{incoherent G-twirl}:
\begin{equation}
    \mathcal{G}_{\textup{inc}}\coloneqq\frac{1}{|G|} \sum_{g\in G}U_{ABC}\left(g\right)\cdot U_{ABC}^\dagger\left(g\right),
\end{equation}
for a general finite symmetry group $G$. In the perspective-neutral approach, one requires also the states to be invariant under the action of the symmetry group. When considering gauge symmetries, this follows from the requirement that the constraint generating the gauge symmetry vanishes. This restriction is implemented through the \emph{coherent G-twirl}
\begin{equation}
    \mathcal{G}_{\textup{coh}}\coloneqq \frac{1}{|G|^2}\sum_{g,g'\in G} U_{ABC}\left(g\right) \cdot U_{ABC}^\dagger\left(g'\right).
\end{equation}
In the simple example of $\mathbb{Z}_2$, the difference between these two approaches lies in how they deal with the total charge $ X_A\otimes X_B\otimes X_C$ \cite{Castro-Ruiz_2025}. While this is part of the invariant observables of the extra-particle approach, it is constrained when acting on states within the \emph{physical} Hilbert space obtained via the coherent G-twirl. The corresponding physical algebra of observables, used in the perspective-neutral approach, is
\begin{equation}
    \mathcal{G}_{\textup{coh}}\left(\algkin{A}\right)=\Pi_{\textup{phys}}\left(\cB\left(\cH_A\otimes \cH_B\otimes \cH_C\right)\right),
\end{equation}
where $\Pi_{\textup{phys}}$ projects onto the phyisical Hilbert space $\cH_{\textup{phys}}=\Pi_{\textup{phys}}\left(\cH_{\textup{kin}}\right)$.
In our example, the algebra of physical observables contains the relative orientations $Z_A\otimes Z_B$, $Z_B\otimes Z_C$, $Z_A\otimes Z_C$ of the qubits as well as their ``momenta'' $X_A, X_B, X_C$, constrained by the condition $X_A\otimes X_B\otimes X_C\ket{\psi}_{\textup{phys}} = \ket{\psi}_{\textup{phys}}$ on states $\ket{\psi}_{\textup{phys}}\in\Pi_{\textup{phys}}\left(\cH_{\textup{kin}}\right)$. Within this algebra, we can identify subalgebras of observables associated to different QRFs by applying the \emph{trivialisation map}\cite{Hoehn_2021, Castro-Ruiz_2025}
\begin{equation}
    \mathcal{T}_{A,BC} \coloneqq \ketbra{0}{0}_A\otimes \id_{BC} + \ketbra{1}{1}\otimes X_B\otimes X_C.
\end{equation}
This ensures that $\mathcal{T}_{A,BC}\left(X_A\otimes X_B\otimes X_C\right)= X_A\otimes \id_{BC}$, \enquote{trivialising} the constraint to only act on the chosen QRF. Moreover, the relative orientations become $\mathcal{T}_{A,BC}\left(Z_A\otimes Z_B\right) = Z_B$ and $\mathcal{T}_{A,BC}\left(Z_A\otimes Z_C\right) = Z_C$ while the ``momenta'' $X_B$ and $X_C$ are left invariant. Discarding the degree of freedom associated to the constraint, which is fixed to a single classical value, we obtain an algebra generated by $X_B\otimes \id_C,Z_B\otimes \id_C,\id_B\otimes X_C,\id_B\otimes Z_C$. Pulling back the algebras $\algpov{B}{A}\coloneqq \langle X_B\otimes \id_C, Z_B\otimes \id_C\rangle$ and $\algpov{C}{A}\coloneqq \langle \id_B\otimes X_C,\id_B\otimes Z_C\rangle $ into the kinematical Hilbert space via the trivialisation map (which, in the case of $\mathbb{Z}_2$, is its own inverse), one obtains
\begin{equation}
    \algpovemb{B}{A} = \langle X_B\otimes \id_C, Z_A\otimes Z_B\rangle, \;\; \algpovemb{C}{A} = \langle \id_B\otimes X_C, Z_A\otimes Z_C\rangle,
\end{equation}
Analogous reasoning using the trivialisation map on $\mathcal{T}_{C,AB} \coloneqq \id_{AB}\otimes \ketbra{0}{0}_C+X_A\otimes X_B\otimes \ketbra{1}{1}_C$ yields  
\begin{equation}
    \algpovemb{A}{C} = \langle X_A\otimes \id_B, Z_A\otimes Z_C\rangle, \;\; \algpovemb{B}{C} = \langle \id_A\otimes X_B, Z_B\otimes Z_C\rangle.
\end{equation}
Since the perspectival approach \cite{Giacomini_2017_covariance, delaHamette_2020} is formally equivalent to the perspective-neutral approach for ideal QRFs, these are also the algebras associated with the perspectives of this approach. Let us note, however, that the derivation differs. Instead of deriving the subsystem algebras relative to different QRFs from the perspective-neutral structure, one starts from within a given perspective and determines how the subsystem observables transform under a change of QRF \cite{Giacomini_2017_covariance}.

\section{Algebraic Subsystems for $\mathbb{Z}_2$-QRFs}
\label{app:algebras_Z2}

In this section we prove several claims relative to~\cref{sec:algebraic_structure_QRFs_Z2}.

\subsection{Pull-back of pinching maps}
\label{app:proof_pullback_mapsZ2}

In~\cref{eq:P_XX}, we claim:
\begin{equation}
    \mappb{\mathcal{P}}_{\pov{X_C}{A}}= \frac{1}{2}\left[\pov{\id_{AB}}{C}\left(\cdot\right)\pov{\id_{AB}}{C} + \left(\pov{X_A}{C}\otimes \pov{X_B}{C}\right) \left(\cdot\right) \left(\pov{X_A}{C}\otimes \pov{X_B}{C}\right)\right] =: \mathcal{P}_{\pov{X_A}{C}\otimes \pov{X_B}{C}}.
\end{equation}
This follows immediately from the following lemma.
\begin{lemma}
\label{lem:pullback_pinching}
    Consider an observable in Alice's perspective $O\in\cB\left(\pov{\cH_{BC}}{A}\right)$ with spectral decomposition $O= \sum_i \lambda_i P_i$, and the pinching map $\mathcal{P}_O$ (see~\cref{ex:pinched_subsystem}). Then, the pull-back into Charlie's perspective satisfies
    \begin{equation}
        \mappb{\mathcal{P}}_{O} = \mathcal{P}_{\left(S_{A\to C}OS_{A\to C}^\dagger\right)},\label{eq:pinch_trafo}
    \end{equation}
    i.e., the pinching in the eigenbasis of the operator $S_{A\to C}OS_{A\to C}^\dagger$.
\end{lemma}
\begin{proof}
Since $S_{A\to C}$ is unitary, the operator $S_{A\to C}OS_{A\to C}^\dagger$ has the same spectrum as $O$. Indeed,
    \begin{equation}
        \begin{split}
            S_{A\to C}OS_{A\to C}^\dagger &= \sum_i \lambda_i S_{A\to C}P_iS_{A\to C}^\dagger = \sum_i \lambda_i \tilde{P}_i,\\
        \end{split}
    \end{equation}
    where $\tilde{P}_i$ are orthogonal projectors summing to the identity. Indeed,
    \begin{equation}
        \begin{split}
            \tilde{P}_i\tilde{P}_j = S_{A\to C}P_iS_{A\to C}^\dagger S_{A\to C}P_jS_{A\to C}^\dagger =S_{A\to C}P_iP_jS_{A\to C}^\dagger =\delta_{i,j}S_{A\to C}P_iS_{A\to C}^\dagger=\delta_{i,j} \tilde{P}_i
        \end{split}
    \end{equation}
    and
    \begin{equation}
            \sum_i \tilde{P}_i= S_{A\to C}\sum_iP_iS_{A\to C}^\dagger = S_{A\to C}S_{A\to C}^\dagger = \id_A.
    \end{equation}
    Thus, $\sum_i \lambda_i \tilde{P}_i$ is the spectral decomposition of $S_{A\to C}OS_{A\to C}^\dagger$, and the corresponding pinching map is $\mathcal{P}_{(S_{A\to C}OS_{A\to C}^\dagger)}\left(\pov{\rho}{C}\right) = \sum_i \tilde{P}_i\pov{\rho}{C} \tilde{P}_i$. Evaluating the pull-back, we get:
    \begin{equation}
        \mappb{\mathcal{P}}_{O}\left(\pov{\rho}{C}\right) = \sum_i\underset{=:\tilde{P}_i}{\underbrace{\left(S_{A\to C}P_i S_{A\to C}^\dagger\right)}} \pov{\rho}{C}\underset{=:\tilde{P}_i}{\underbrace{\left(S_{A\to C}P_i S_{A\to C}^\dagger\right)}} = \sum_i\tilde{P}_i\rho^{\left(C\right)}\tilde{P}_i =\mathcal{P}_{(S_{A\to C}OS_{A\to C}^\dagger)}\left(\pov{\rho}{C}\right) .\label{eq:pinch_trafo_1}
    \end{equation}
\end{proof}

\Cref{eq:P_XX}, follows since 
\begin{equation}
    S_{A\to C}\left(\pov{\id_B}{A}\otimes \pov{X_C}{A}\right)S_{A\to C}^\dagger = \pov{\ketbra{0}{1}_A}{C}\otimes \pov{X_B}{C} + \pov{\ketbra{1}{0}_A}{C}\otimes \pov{X_B}{C}= \pov{X_A}{C}\otimes \pov{X_B}{C},
\end{equation}
thus
\begin{equation}
    \mappb{\mathcal{P}}_{\pov{X_C}{A}}= \frac{1}{2}\left[\pov{\id_{AB}}{C}\left(\cdot\right)\pov{\id_{AB}}{C} + \left(\pov{X_A}{C}\otimes \pov{X_B}{C}\right) \left(\cdot\right) \left(\pov{X_A}{C}\otimes \pov{X_B}{C}\right)\right]=\mathcal{P}_{\pov{X_A}{C}\otimes\pov{X_B}{C} }.
\end{equation}
Using that the orientation of the reference frame and the $\mathbb{Z}_2$-momentum are invariant under the QRF transformation, i.e.,~$S_{A\to C} \left(\id_B^{\left(C\right)}\otimes Z_C^{\left(A\right)}\right) S_{A\to C}^\dagger=Z_A^{\left(C\right)}\otimes \id_B^{\left(C\right)}$ and $S_{A\to C} \left(X_B^{\left(C\right)}\otimes \id_C^{\left(A\right)}\right) S_{A\to C}^\dagger=\id_A^{\left(C\right)}\otimes X_B^{\left(C\right)}$, we further have
\begin{equation}
    \mappb{\mathcal{P}}_{Z_C^{\left(A\right)}} = \mathcal{P}_{Z_A^{\left(C\right)}}, \qquad
    \mappb{\mathcal{P}}_{X_B^{\left(A\right)}} = \mathcal{P}_{X_B^{\left(C\right)}}.
\end{equation}

\subsection{Algebraic construction of the left branch of~\Cref{fig:encoded-decomposition}}
\label{app:conditional_exp_construction_proofs}

Consider the joint algebra of $\pov{A}{C}$ and $\pov{C}{A}$
\begin{equation}
    \algpovpb{C}{A} \vee \algpov{A}{C} =\langle  \pov{\id_A}{C}\otimes \pov{X_B}{C}, \pov{X_A}{C}\otimes \pov{\id_{B}}{C}, \pov{Z_A}{C}\otimes \pov{\id_{B}}{C}\rangle = \cB\left(\pov{\cH_{A}}{C}\right)\otimes \braket{\pov{X_B}{C}},
\end{equation}
where the last equality follows by noticing that the full algebra of operators acting on $\pov{A}{C}$ is included, but only the Pauli-X on $\pov{B}{C}$ is in $ \algpovpb{C}{A} \vee \algpov{A}{C}$. Thus, the conditional expectation from the full algebra $\cB\left(\pov{\cH_{AB}}{C}\right)$ to the algebra $\algpovpb{C}{A} \vee \algpov{A}{C}$ is the pinching map $\mathcal{P}_{\pov{X_B}{C}}$, i.e.,
\begin{equation}
    \cexparg{\pov{AB}{C}\to \pov{A}{C}\pov{C}{C}}\coloneqq \mathcal{P}_{\pov{X_B}{C}}.
\end{equation}
Likewise, we can find the conditional expectations onto the subalgebras
\begin{align}
    \algpov{A}{C} = \langle \pov{X_A}{C}\otimes \pov{\id_B}{C},\pov{Z_A}{C}\otimes \pov{\id_B}{C}\rangle, \qquad \algpovpb{C}{A} = \langle \pov{X_A}{C}\otimes \pov{X_B}{C}, \pov{Z_A}{C}\otimes \pov{\id_{B}}{C}\rangle.
\end{align}
starting from the joint algebra $\algpovpb{C}{A} \vee \algpov{A}{C}$. 
Pinching in $\pov{Z_B}{C}$ trivialises the term $\pov{\id_A}{C}\otimes \pov{X_B}{C}$, giving the algebra $\algpov{A}{C}$. Thus,
\begin{equation}
    \cexparg{{\pov{A}{C}\pov{C}{A}}\to \pov{A}{C}} =\mathcal{P}_{Z_B^{\left(C\right)}}.
\end{equation}
The conditional expectation to $\pov{C}{A}$'s algebra can be obtained through the symmetry between $A$ and $C$:
\begin{equation}
    \cexparg{{\pov{A}{C}\pov{C}{A}}\to \pov{C}{A}} = \mappb{\mathcal{P}}_{\pov{Z_B}{A}} = \mathcal{P}_{\pov{Z_A}{C} \otimes \pov{Z_B}{C}},
\end{equation}
where in the last line, we used~\cref{lem:pullback_pinching} and
\begin{equation}
    S_{A\to C}\left(Z_B^{\left(A\right)}\otimes \id_C^{\left(A\right)}\right)S_{A\to C}^\dagger = \ketbra{0}{0}_A^{\left(C\right)}\otimes Z_B^{\left(C\right)}+ \ketbra{1}{1}_A^{\left(C\right)}\otimes X_B^{\left(C\right)}Z_B^{\left(C\right)}X_B^{\left(C\right)}= Z_A^{\left(C\right)}\otimes Z_B^{\left(C\right)}.
\end{equation}

\subsection{Conditional expectation onto overlap}

Here, we show~\cref{eq:cond_exp_overlap_Z2_1}, i.e.,~that the conditional expectation onto the overlap $\algpovpb{B}{A}\wedge \algpov{B}{C}$ can be obtained by composing the conditional expectations onto the subalgebras involved. Let us first observe that for an Hermitian observable $O$ such that $O^2 =\id$ it holds $\mathcal{P}_O\left(\rho\right) = \frac{1}{2}\left(\rho + O \rho O\right).$\footnote{ Indeed, $O$ has eigenvalues $\pm 1$ and eigenspaces with projectors $P_{\pm}$, using that $P_++P_- = \id$ and $O=P_+ - P_-$ we have: $\rho+O\rho O = (P_+ +P_-) \rho (P_+ +P_-) + (P_+ -P_-)\rho (P_+ -P_-) = 2 (P_+\rho P_++P_-\rho P_-)$, thus $\frac{1}{2}(\rho+O\rho O) = \mathcal{P}_O(\rho)$.} Thus, we have
\begin{align*}
    \mathcal{P}_{\pov{X_A}{C}}\circ \mathcal{P}_{\pov{X_A}{C}\otimes \pov{X_B}{C}}\left(\rho\right)&= \small \frac{1}{4}\left(\rho + X_A^{\left(C\right)}\rho\, X_A^{\left(C\right)}+  X_B^{\left(C\right)}\rho\,  X_B^{\left(C\right)}+X_A^{\left(C\right)}\otimes X_B^{\left(C\right)}\rho \, X_A^{\left(C\right)}\otimes X_B^{\left(C\right)}\right)\\
    &= \mathcal{P}_{X_B^{\left(C\right)}}\left(\frac{1}{2}\left[\rho+X_A^{\left(C\right)}\rho \,X_A^{\left(C\right)}\right]\right)=  \mathcal{P}_{X_B^{\left(C\right)}}\circ\mathcal{P}_{X_A^{\left(C\right)}}\left(\rho\right).
\end{align*}
But if we start from $\rho\in\algpovpb{B}{A} \vee \algpov{B}{C}$, we have $\rho =\mathcal{P}_{Z_A^{\left(C\right)}}\left(\rho\right)$ and, by~\cref{eq:full_pinchingA}, $\mathcal{P}_{X_A^{\left(C\right)}}\left(\rho\right) = \mathcal{P}_{X_A^{\left(C\right)}}\circ \mathcal{P}_{Z_A^{\left(C\right)}}\left(\rho\right) = \tfrac{1}{2}\id_A^{\left(C\right)}\otimes \Tr_{A^{\left(C\right)}}\left(\rho\right)$. Hence,
\begin{align}
\mathcal{P}_{\pov{X_A}{C}}\circ \mathcal{P}_{\pov{X_A}{C}\otimes \pov{X_B}{C}}\left(\rho\right) = \tfrac{1}{2}\id_A^{\left(C\right)}\otimes \mathcal{P}_{X_B^{\left(C\right)}}\left(\Tr_{A^{\left(C\right)}}\rho\right) \text{ for } \rho \in \algpovpb{B}{A} \vee \algpov{B}{C}.
\end{align}
The same reasoning leads to
\begin{align*}
    \mathcal{P}_{\pov{X_A}{C}\otimes \pov{X_B}{C}}\circ \mathcal{P}_{\pov{X_A}{C}}\left(\rho\right)=\tfrac{1}{2}\id_A^{\left(C\right)}\otimes \mathcal{P}_{X_B^{\left(C\right)}}\left(\rho\right) \text{ for } \rho \in \algpovpb{B}{A} \vee \algpov{B}{C}.
\end{align*}
This also proves that the conditional expectations commute. The proof for~\cref{eq:cond_exp_overlap_Z2_2} proceeds completely analogously.

\section{Algebraic Subsystems for Finite-Dimensional Symmetry Groups}
\label{app:algebraicSubsystemFDG}

\subsection{Joint Algebra: $\pov{B}{A}$ and $\pov{B}{C}$}
\label{app:joint_algebras}
We first prove the statement that the joint algebra of $\pov{B}{A}$ and $\pov{B}{C}$ is obtained from the total algebra relative to $C$ by a pinching of $A$ in the group basis~(\cref{eq:joint_algebra_BABC}), that is
\begin{equation}
\algpovpb{B}{A} \vee \algpov{B}{C}=\Big \langle \pov{\id_A}{C}\otimes O_B , \sum_g \ketbra{g}{g}_A\otimes U_B\left(g\right)O_BU_B^\dagger\left(g\right) \Big\rangle_{O_B\in \cB\left(\pov{\cH_B}{C}\right)} = \mathcal{P}_{\pov{G_A}{C}}\left(\algpov{A}{C}\vee \algpov{B}{C}\right)
   \label{eqApp:pinchingG}
\end{equation}
where $\mathcal{P}_{\pov{G_A}{C}}\left(\cdot\right) \coloneqq \sum_{g\in G} \pov{\ketbra{g}{g}}{C}_A\otimes \langle g| \cdot \pov{|g \rangle}{C}_A$. In what follows, we sometimes omit the superscript $(C)$ denoting that we are working in the perspective of $C$, unless ambiguous. 

Let us first define 
\begin{equation}
    \mathcal{D} = \left\{ \sum_g \ketbra{g}{g}_A\otimes O_g \right\}_{O_g\in\cB(\cH_B)},
\end{equation}
we will prove that both algebras of interest, $\algpovpb{B}{A} \vee \algpov{B}{C}$ and $\mathcal{P}_{\pov{G_A}{C}}\left(\algpov{A}{C}\vee \algpov{B}{C}\right)$ are equal to $\mathcal{D}$.
One can easily see that $\mathcal{D}$ itself is a unital $*$-subalgebra of $\algpov{A}{C}\vee \algpov{B}{C}$. Indeed, it contains the identity (take $O_g=\id_B$) and it is closed under products, linear combinations and taking the adjoint. In addition, it is equal to the algebra
\begin{equation}
\label{eq:tmp_algebra}
    \Big\langle \ketbra{g}{g}_{A}\otimes \id_B, \id_A \otimes O_B\Big\rangle_{g\in G, O_B\in \cB(\cH_{B})}.
\end{equation}
On one hand, the generators in~\cref{eq:tmp_algebra} are in $\mathcal{D}$ (take $O_g = \delta_{g,g'}\id_B$ and $O_g = O$). Since we have proven that $\mathcal{D}$ is an algebra, and it contains $\ketbra{g}{g}_A\otimes \id_B$ and $\id_A\otimes O_B$ for all $g\in G$ and $O\in\cB(\cH_B)$, then it also contains the algebra generated by them, i.e.,~\cref{eq:tmp_algebra}. Conversely, one can see that the elements of $\mathcal{D}$ can be written as products of generators~\cref{eq:tmp_algebra}: $\sum_g \ketbra{g}{g}_A\otimes O_g = \sum_g (\ketbra{g}{g}_A\otimes\id_B) (\id_A\otimes O_g)$, thus $\mathcal{D}$ is contained in~\cref{eq:tmp_algebra}, implying that
\begin{equation}
    \mathcal{D} = \left\{ \sum_g \ketbra{g}{g}_A\otimes O_g \right\}_{O_g\in\cB(\cH_B)}= \Big\langle \ketbra{g}{g}_{A}\otimes \id_B, \id_A \otimes O_B\Big\rangle_{g\in G, O_B\in \cB(\cH_{B})}.
\end{equation}

We now show that $\mathcal{D}=\mathcal{P}_{\pov{G_A}{C}}\left(\algpov{A}{C}\vee \algpov{B}{C}\right)$. Let $O_{AB}\in\algpov{A}{C}\vee \algpov{B}{C}$, we have $\mathcal{P}_{\pov{G_A}{C}}\left(O_{AB}\right) = \sum_g \ketbra{g}{g}_A\otimes \bra{g}O_{AB}\ket{g}_A = \sum_g \ketbra{g}{g}_A\otimes O(g)_{AB}\in \mathcal{D}$. Thus, $\mathcal{P}_{\pov{G_A}{C}}\left(\algpov{A}{C}\vee \algpov{B}{C}\right)\subseteq \mathcal{D}$. Conversely, $\mathcal{P}_{\pov{G_A}{C}}(\mathcal{D})=\mathcal{D}$, thus since $\mathcal{D}$ is a subalgebra of $\algpov{A}{C}\vee \algpov{B}{C}$. We have $\mathcal{D}=\mathcal{P}_{\pov{G_A}{C}}(\mathcal{D}) \subseteq \mathcal{P}_{\pov{G_A}{C}}(\algpov{A}{C}\vee \algpov{B}{C})$.
Thus, proving equality between the two algebras.

Let us now prove that $\mathcal{D}=\algpovpb{B}{A} \vee \algpov{B}{C}$. Clearly, $\algpovpb{B}{A} \vee \algpov{B}{C} \in \mathcal{D}$, since both generators are in $\mathcal{D}$ (take $O_g = U_B(g)OU_B^\dagger(g)$). For the other direction, we need to show that $\ketbra{g}{g}_A\otimes \id_B \in \algpovpb{B}{A} \vee \algpov{B}{C}$ for all $g\in G$, since the other generator of $\mathcal{D}$ is already in $\algpovpb{B}{A} \vee \algpov{B}{C}$. First consider the seed state $\ket{e}_B$ with $e$ being the identity in $G$, satisfying $U_B(g)\ket{e}_B=\ket{g}_B$, we have:
\begin{equation}
    \sum_g \ketbra{g}{g}_A \otimes U_B(g)\ketbra{e}{e}_BU^\dagger_B(g) =  \sum_g \ketbra{g}{g}_A \otimes \ketbra{g}{g}_B \in \algpovpb{B}{A} \vee \algpov{B}{C},
\end{equation}
then for every $g,k\in G$, recalling that algebras are closed under products, we have:
\begin{equation}
    (\id_A\otimes \ketbra{k}{g}_B) (\sum_{g'} \ketbra{g'}{g'}_A \otimes \ketbra{g'}{g'}_B)(\id_A\otimes \ketbra{g}{k}_B) = \ketbra{g}{g}_A \otimes \ketbra{k}{k}_B \in \algpovpb{B}{A} \vee \algpov{B}{C}. 
\end{equation}
Thus, summing over $k$ and using that algebras are closed under linear combinations, we have $\sum_k \ketbra{g}{g}_A \otimes \ketbra{k}{k}_B = \ketbra{g}{g}_A \otimes \id_B$ for all $g\in G$. Thus, $\mathcal{D} \subseteq \algpovpb{B}{A} \vee \algpov{B}{C}$, implying $\mathcal{D} = \algpovpb{B}{A} \vee \algpov{B}{C}$. 

By transitivity, $\algpovpb{B}{A} \vee \algpov{B}{C} = \mathcal{D} = \mathcal{P}_{\pov{G_A}{C}}\left(\algpov{A}{C}\vee \algpov{B}{C}\right)$.

Finally, we want to convince ourselves that the joint algebra does not depend on whether it is constructed from Alice's or Charlie's point of view, that is $\mappb{\mathcal{P}}_{\pov{G_A}{C}}=\mathcal{P}_{\pov{G_C}{A}}$.

To see this, note first that
\begin{align}
    \mappb{\mathcal{P}}_{\pov{G_C}{A}}\left(\pov{\rho_{AB}}{C}\right) &= S_{A\to C} \left[\sum_g \left(\pov{\id_B}{A}\otimes \ketbra{g}{g}_C\right)S_{A\to C}^\dagger \pov{\rho_{AB}}{C}S_{A\to C}\left(\pov{\id_B}{A}\otimes \ketbra{g}{g}_C\right)\right]S_{A\to C}^\dagger.
\end{align}
For each $g\in G$, we have
\begin{align}
   S_{A\to C} &\left(\pov{\id_B}{A}\otimes \ketbra{g}{g}_C\right)S_{A\to C}^\dagger\\
   &=\left(\sum_h \ket{h^{-1}}_A\bra{h}_C\otimes  U_B^\dagger(h)\right)\left(\pov{\id_B}{A}\otimes \ketbra{g}{g}_C\right) \left(\sum_{h'} \ket{h'}_C\bra{h^{-1}}_A\otimes  U_B(h')\right)\\
   &= \ketbra{g^{-1}}{g^{-1}}_A\otimes \id_B.
\end{align}
Renaming $g\to g^{-1}$, we obtain
\begin{align}
    \mappb{\mathcal{P}}_{\pov{G_C}{A}}\left(\pov{\rho_{AB}}{C}\right)= \sum_g\left(\ketbra{g}{g}_A\otimes \id_B\right)\pov{\rho_{AB}}{C}\left(\ketbra{g}{g}_A\otimes \id_B\right)= \mathcal{P}_{\pov{G_A}{C}}\left(\pov{\rho_{AB}}{C}\right).
\end{align}

\subsection{Joint algebra: $\pov{C}{A}$ and $\pov{A}{C}$}
\label{app:joint_alg_AC}
We now want to show that the joint algebra of $\pov{A}{C}$ and $\pov{C}{A}$ is obtained from the total algebra relative to $C$ by the right incoherent G-twirl of $B$, that is
\begin{equation}
    \algpov{A}{C} \vee \algpovpb{C}{A} \coloneqq \Big\langle O \otimes \pov{\id_{B}}{C}, \sum_{g,g'}\ketbra{g}{g'}\langle g^{-1}|O|g'^{-1}\rangle \otimes U_B\left(gg'^{-1}\right) \Big\rangle_{O\in \cB\left(\pov{\cH_A}{C}\right)} = \mathcal{G}_{\inc_R,B}\left(\algpov{A}{C}\vee \algpov{B}{C}\right),\label{eqApp:incGtwirl}
\end{equation}
where $\mathcal{G}_{\inc_R,B} = \frac{1}{|G|}\sum_g (\id_A\otimes R\left(g\right))\cdot (\id_A\otimes R^\dagger\left(g\right))$ with $R\left(h\right)\ket{g} = \ket{gh^{-1}}$ the right-regular representation of $G$. In what follows, we omit the superscript $(C)$ denoting that we are working in the perspective of $C$. 

Let us define
\begin{equation}
    \mathcal{E} =  \left\{\sum_k O_k \otimes U_B(k): O_k\in\cB(\pov{\cH_A}{C})\right\}.
\end{equation}
This is a unital $*$-subalgebra of $\algpov{A}{C}\wedge \algpov{B}{C}$. Indeed, it is easy to see that it contains the identity ($U_B(e)=\id_B$) and it is closed under products and adjoint ($U^\dagger_B(k)=U_B(k^{-1})$). Considering the algebra
\begin{equation}
\label{eq:tmp_algebra2}
    \Big\langle O_A\otimes \id_B, \id_A \otimes U_B(k)\Big\rangle_{k\in G, O_A\in \cB(\cH_{A})},
\end{equation}
we notice that the generators of~\cref{eq:tmp_algebra2} are in $\mathcal{E}$ (take, $U_B(e)=\id_B$ and $O_k=\id_A$), thus, since $\mathcal{E}$ is an algebra,~\cref{eq:tmp_algebra2} is contained in $\mathcal{E}$. Conversely $\sum_k O_k\otimes U_B(k)=\sum_k (O_k\otimes\id_B)(\id_A\otimes U_B(k))$ is a linear combination of products of generators in~\cref{eq:tmp_algebra2}. Thus, $\mathcal{E}$ is contained in~\cref{eq:tmp_algebra2}, implying the equality between the two.

Consider any $O_{AB}\in\algpov{A}{C}\vee \algpov{B}{C}$, we have $O_{AB}=\sum_{h,k} O^{(h,k)}_A\otimes \ketbra{h}{k}_B$
\begin{equation}
\begin{split}
     \mathcal{G}_{\inc_R,B}(O_{AB}) &= \frac{1}{|G|}\sum_g \sum_{h,k} (\id_A\otimes R\left(g\right))(O^{(h,k)}_A\otimes \ketbra{h}{k}_B) (\id_A\otimes R^\dagger\left(g\right))\\
     &=\frac{1}{|G|}\sum_{g,h,k} O^{(h,k)}_A \otimes R\left(g\right)\ketbra{h}{k}R^\dagger\left(g\right)\\
     &=\frac{1}{|G|}\sum_{g,h,k} O^{(h,k)}_A \otimes \ketbra{hg^{-1}}{kg^{-1}}\\
     &=\frac{1}{|G|}\sum_{h,k} O^{(h,k)}_A \otimes U_B(h)\sum_{g}\ketbra{g^{-1}}{g^{-1}}U^\dagger_B(k) \\
     &= \frac{1}{|G|}\sum_{h,k} O^{(h,k)}_A \otimes U_B(h)U^\dagger_B(k) \\
     &= \frac{1}{|G|}\sum_{h,k} O^{(h,k)}_A \otimes U_B(hk^{-1})= \frac{1}{|G|}\sum_{h'} O_{h'} \otimes U_B(h') \in \mathcal{E}
\end{split}
\end{equation}
where we defined $\sum_k O^{(h'k,k)}_A = O_{h'}\in\cB(\pov{\cH_A}{C})$.
Conversely, 
\begin{equation}
\begin{split}
    \mathcal{G}_{\inc_R,B}\left(\sum_k O_k \otimes U_B(k)\right) &=  \frac{1}{|G|}\sum_g  (\id_A\otimes R\left(g\right))\left(\sum_k O_k \otimes U_B(k)\right) (\id_A\otimes R^\dagger\left(g\right))\\
    &=\frac{1}{|G|}\sum_k O_k \otimes \sum_g  R\left(g\right)U_B(k) R^\dagger\left(g\right)\\
    &=\frac{1}{|G|}\sum_k O_k \otimes \sum_{g,h}  R\left(g\right)\ketbra{kh}{h}R^\dagger\left(g\right)\\
    &=\frac{1}{|G|}\sum_k O_k \otimes \sum_{g,h}  \ketbra{khg^{-1}}{hg^{-1}}\\
    &=\frac{1}{|G|}\sum_k O_k \otimes \sum_{g,h}  U_B(kh)\ketbra{g^{-1}}{g^{-1}}U_B^\dagger(h) \\
    &= \frac{1}{|G|}\sum_k O_k \otimes \sum_{h}  U_B(kh)U_B^\dagger(h) = \sum_k O_k \otimes U_B(k).
\end{split}
\end{equation}
Since $\mathcal{E}$ is a subalgebra of $\algpov{A}{C}\vee \algpov{B}{C}$ and $ \mathcal{G}_{\inc_R,B}\left(\mathcal{E}\right) = \mathcal{E}$, we have $\mathcal{E}\subseteq \mathcal{G}_{\inc_R,B}\left(\algpov{A}{C}\vee \algpov{B}{C}\right)$. Implying $\mathcal{E}= \mathcal{G}_{\inc_R,B}\left(\algpov{A}{C}\vee \algpov{B}{C}\right)$.

To see that $\algpov{A}{C} \vee \algpovpb{C}{A}\subseteq\mathcal{E}$, we can simply regroup $gg^{-1}=k$ in the second generator of $\algpov{A}{C}\vee\algpovpb{C}{A} $ and find
\begin{equation}
    \sum_{g,g'}\ketbra{g}{g'}\langle g^{-1}|O|g'^{-1}\rangle \otimes U_B\left(gg'^{-1}\right) = \sum_k \left(\sum_{g'}\ketbra{kg'}{g'}\langle g^{'-1}k^{-1}|O|g'^{-1}\rangle \right) \otimes U_B\left(k\right) = \sum_k O_k \otimes U_B\left(k\right)
\end{equation}
which is in the algebra $\mathcal{E}$, thus also contains all elements generated by it.
Conversely, to show that $\mathcal{E}\subseteq\algpov{A}{C}\vee\algpovpb{C}{A} $, we simply need to show that the generator $\id_A\otimes U_B(k)$ is in $\algpov{A}{C}\vee\algpovpb{C}{A} $ for all $k\in G$, since the other generator is included by definition. Let us first take $O = \ketbra{a}{b}$:
\begin{equation}
    \sum_{g,g'}\ketbra{g}{g'}\langle g^{-1}\ketbra{a}{b}g'^{-1}\rangle \otimes U_B\left(gg'^{-1}\right) = \ketbra{a^{-1}}{b^{-1}} \otimes U_B(a^{-1}b)\in \algpov{A}{C}\vee \algpovpb{C}{A},
\end{equation}
renaming $a=k^{-1}$ and choosing $b=e$ (the identity on $G$), and multiplying by two other elements of $\algpov{A}{C} \vee \algpovpb{C}{A} $, we have:
\begin{equation}
    (\ketbra{m}{k}_A\otimes \id_B) \left(\ketbra{k}{e^{-1}} \otimes U_B(ke)\right)(\ketbra{e}{m}_A\otimes \id_B) = \ketbra{m}{m}\otimes U_B(k) \in \algpov{A}{C} \vee \algpovpb{C}{A} 
\end{equation}
thus summing over $m\in G$ we obtain $\id_A\otimes U_B(k)\in  \algpov{A}{C}\vee\algpovpb{C}{A}$. We can conclude $\algpov{A}{C}\vee \algpovpb{C}{A}  = \mathcal{E}$

By transitivity, $ \algpov{A}{C}\vee \algpovpb{C}{A}=\mathcal{E} = \mathcal{G}_{\inc_R,B}\left(\algpov{A}{C}\vee \algpov{B}{C}\right)$.

To see that this is independent of whether we construct the algebra from Alice's or Charlie's point of view, note that
\begin{equation}
    \mappb{\mathcal{G}}_{\inc_R,B^{\left(A\right)}}\left(\rho\right) = \frac{1}{|G|}\left(S_{A\to C}R_{B}\left(g\right) S_{A\to C}^\dagger \right)\rho \left(S_{A\to C}R_{B}\left(g\right) S_{A\to C}^\dagger \right)^\dagger
\end{equation}
with
\begin{align}
    S_{A\to C}\left(R_{B}\otimes \id_C\right) S_{A\to C}^\dagger &= \sum_{h,h'\in G}|h^{-1}\rangle\langle h'^{-1}|_A\underset{\delta_{hh'}}{\underbrace{\langle h|h'\rangle}}\otimes U_B^\dagger\left(h\right) R_B\left(g\right) U_B\left(h'\right)\\ &= \sum_{h\in G}|h^{-1}\rangle \langle h^{-1}|_A\otimes R_B\left(g\right) =\id_A\otimes R_B\left(g\right).
\end{align}
Thus, leaving the identity implicit again,
\begin{equation}
\mappb{\mathcal{G}}_{\inc_R,\pov{B}{A}}(\rho)=\frac{1}{|G|}\sum_{g\in G}R_{B}(g) \rho R_{B}(g) = \mathcal{G}_{\inc_R,\pov{B}{C}}(\rho).
\end{equation}

\subsection{Conditional expectation onto the overlap of algebras}
\label{app:overlap_algebras}
Let us start proving that 
\begin{align}
    \cexparg{\pov{B}{A} \wedge \pov{B}{C}}(O) \coloneqq \cexparg{\pov{B}{A}}^{\triangleleft}\circ \cexparg{\pov{B}{C}}(O) = \cexparg{\pov{B}{C}}\circ \cexparg{\pov{B}{A}}^{\triangleleft} (O)= \frac{1}{|G|}\pov{\id_A}{C}\otimes \mathcal{G}_{\inc,\pov{B}{C}}(\tr_A(O)).
\end{align}
To do so, we us first establish two features of QRF transformations. We have
    \begin{align}
        \Tr_C\left(\mathcal{S}_{C\to A}\left(\tfrac{1}{|G|}\pov{\id_A}{C}\otimes \pov{\rho_B}{C}\right)\right) &= \Tr_C\left(\frac{1}{|G|}\sum_{g}\ketbra{g^{-1}}{g^{-1}}_C \otimes U_B^\dagger\left(g\right)\pov{\rho_B}{C}U_B\left(g\right) \right)  \\
        &=  \frac{1}{|G|}\sum_{g}U_B\left(g\right)\pov{\rho_B}{C}U_B^\dagger\left(g\right)=
        \mathcal{G}_{\inc,\pov{B}{C}}\left(\pov{\rho_B}{C}\right)\label{eq:G_inc_trace}
    \end{align}
and the same for $A\leftrightarrow C$. Then we also have
\begin{equation}
\begin{split}
    \mathcal{G}_{\inc,B}\left(\pov{\rho_B}{A}\right) &= \mathcal{G}_{\inc, B}\left(\Tr_C\left(S_{C\to A }\pov{\rho_{AB}}{C}S_{C\to A}^\dagger\right)\right)\\
    &= \mathcal{G}_{\inc,B}\left(\sum_g \left(\bra{g}_A\otimes U_B^\dagger\left(g\right)\right)\pov{\rho_{AB}}{C}\left(\ket{g}_A\otimes U_B\left(g\right)\right)\right)\\
    &= \sum_g \bra{g}_A\left(\mathcal{G}_{\inc,B}\left(U_B^\dagger\left(g\right) \pov{\rho_{AB}}{C}U_B\left(g\right)\right)\right)\ket{g}_A\\
    &= \frac{1}{|G|}\sum_{g,g'} \bra{g}_A\left(U_B(g')U_B^\dagger\left(g\right) \pov{\rho_{AB}}{C}U_B\left(g\right)U^\dagger_B(g')\right)\ket{g}_A\\
    &= \frac{1}{|G|}\sum_{g,h} \bra{g}_A\left(U_B(h)\pov{\rho_{AB}}{C}U^\dagger_B(h)\right)\ket{g}_A\\
    &= \mathcal{G}_{\inc,B}\left(\sum_g \langle g|\pov{\rho_{AB}}{C}|g\rangle \right) = \mathcal{G}_{\inc,B}\left(\pov{\rho_B}{C}\right).\label{eq:G_inc_invariant}
\end{split}
\end{equation}
Using these two properties as well as $\cexparg{\pov{B}{C}} \coloneqq \tfrac{1}{|G|}\pov{\id_A}{C}\otimes \Tr_A$ and $\cexparg{\pov{B}{A}} \coloneqq \Tr_C \otimes \tfrac{1}{|G|}\id_C^{\left(A\right)}$, we have
\begin{align}
    \cexparg{\pov{B}{A}}^{\triangleleft}\left(\cexparg{\pov{B}{C}}\left(\rho\right)\right) = &S_{A\to C}\left(\frac{1}{|G|}\pov{\id_C}{A}\otimes \Tr_C\left(S_{C\to A}\left(\frac{1}{|G|}\pov{\id_A}{C}\otimes \pov{\rho_B}{C}\right)S_{C\to A}^\dagger\right)\right)S_{A\to C}^\dagger\\
    \overset{\labelcref{eq:G_inc_trace}}{=} &S_{A\to C}\left(\frac{1}{|G|}\pov{\id_C}{A}\otimes \mathcal{G}_{\inc,B}\left(\pov{\rho_B}{C}\right)\right) S_{A\to C}^\dagger\\ = &\frac{1}{|G|}\sum_g \ketbra{g}{g}_A\otimes U_B\left(g\right) \mathcal{G}_{\inc,B}\left(\pov{\rho_B}{C}\right)U_B^\dagger\left(g\right) = \frac{1}{|G|}\pov{\id_A}{C}\otimes\mathcal{G}_{\inc,B}\left(\pov{\rho_B}{C}\right).
\end{align}
Similarly, we have
\begin{align}
    \cexparg{\pov{B}{C}}\left(\cexparg{\pov{B}{A}}^{\triangleleft}\left(\rho\right)\right) = &\frac{1}{|G|}\pov{\id_A}{C}\otimes \Tr_A \left(S_{A\to C} \left(\frac{1}{|G|}\pov{\id_C}{A}\otimes \Tr_C\left(\pov{\rho_{BC}}{A}\right)\right)S_{A\to C}^\dagger\right)\\
    \overset{\labelcref{eq:G_inc_trace}}{=} &\frac{1}{|G|}\pov{\id_A}{C}\otimes \mathcal{G}_{\inc,B}\left(\pov{\rho_B}{A}\right)  \overset{\labelcref{eq:G_inc_invariant}}{=} \frac{1}{|G|}\pov{\id_A}{C}\otimes \mathcal{G}_{\inc,B}\left(\pov{\rho_B}{C}\right).
\end{align}
Next, we show that 
\begin{equation}
    \left(\cexparg{\pov{A}{C}\wedge \pov{C}{A}}\right)\left(\rho\right) \coloneqq \left(\cexparg{\pov{A}{C}}\circ \cexparg{\pov{C}{A}}^{\triangleleft}\right)\left(\rho\right) = \left(\cexparg{\pov{C}{A}}^{\triangleleft}\circ \cexparg{\pov{A}{C}}\right)\left(\rho\right) = \mathcal{P}_{G,\pov{A}{C}}\left(\rho\right)\otimes \frac{1}{|G|}\pov{\id_{B}}{C}.
\end{equation}
To do so, let us first prove the following
\begin{align}
    \Tr_{\pov{B}{C}}\left(S_{A\to C} \left(\frac{1}{|G|}\pov{\id_B}{A}\otimes \pov{\rho_C}{A}\right)S_{A\to C}^\dagger\right) =& \frac{1}{|G|} \sum_{g,g'}\langle g|\pov{\rho_C}{A}|g'\rangle_C \ketbra{g^{-1}}{g'^{-1}}_A   \otimes \tr_B\left(U_B^\dagger\left(g\right) U_B\left(g'\right)\right)\\
    = &\sum_{g}\pov{\ketbra{g}{g}_A}{C} \langle g^{-1}|\pov{\rho_C}{A}|g^{-1}\rangle_C\label{eq:pinching_trace},
\end{align}
where we used that $\tr_B\left(U_B^\dagger\left(g\right) U_B\left(g'\right)\right) =\sum_{h,h'} \tr_B (\ketbra{h}{gh}g'h'\rangle\bra{h'}) = \sum_h \delta_{g,g'} = |G|\delta_{g,g'}$. Notice that this also holds under exchanging $A$ and $C$. 
In addition, we have
\begin{align}
    \braket{g|\pov{\rho_A}{C}|g}_A = &\Tr_B\left(\braket{g|\pov{\rho_{AB}}{C}|g}_A\right) = \Tr_B\left(\braket{g|S_{A\to C}\pov{\rho_{BC}}{A}S_{A\to C}^\dagger|g}_A\right)\\
    =& \sum_{h,h'}\Tr_B\left[\left(\braket{ g|h^{-1}}_A \bra{h}_C \otimes U_B^\dagger(h) \right) \pov{\rho_{BC}}{A}\left(\ket{h'}_C \braket{h'^{-1}|g}_A\otimes U_B(h')\right) \right]\\
    = &\bra{ g^{-1}}\Tr_B(U_B^\dagger(g) \pov{\rho_{BC}}{A}U_B(g))\ket{g^{-1}}_C = \braket{g^{-1}|\pov{\rho_C}{A}|g^{-1}} _C.\label{eq:pinching_invariance}
\end{align}
Using this as well as the definition of the conditional expectations, $\cexparg{\pov{A}{C}}\coloneqq \Tr_B\otimes \frac{1}{|G|}\pov{\id_{B}}{C}$ and $\cexparg{\pov{C}{A}} = \frac{1}{|G|}\pov{\id_B}{A}\otimes \Tr_B$, we obtain
\begin{align}
    \left(\cexparg{\pov{A}{C}}\circ \cexparg{\pov{C}{A}}^{\triangleleft}\right)\left(\rho\right) = &\Tr_{\pov{B}{C}}\left(S_{A\to C} \left(\frac{1}{|G|}\pov{\id_B}{A}\otimes \pov{\rho_C}{A}\right)S_{A\to C}^\dagger\right)\otimes \frac{1}{|G|}\pov{\id_{B}}{C}\\
    \overset{\labelcref{eq:pinching_trace}}{=}&\sum_{g}\ketbra{g}{g}_A\otimes \langle g^{-1}|\pov{\rho_C}{A}|g^{-1}\rangle\otimes \frac{1}{|G|}\pov{\id_{B}}{C}\\
    \overset{\labelcref{eq:pinching_invariance}}{=}&\sum_g\ketbra{g}{g}_A\otimes \langle g|\pov{\rho_A}{C}|g\rangle \otimes \frac{1}{|G|}\pov{\id_{B}}{C}
\end{align}
as well as 
\begin{align}
    \left(\cexparg{\pov{C}{A}}^{\triangleleft}\circ \cexparg{\pov{A}{C}}\right)\left(\rho\right)=&S_{A\to C}\left(\frac{1}{|G|}\pov{\id_B}{A}\otimes \Tr_{\pov{B}{A}}\left(S_{C\to A}\left(\pov{\rho_A}{C}\otimes \frac{1}{|G|}\pov{\id_{B}}{C}\right)S_{C\to A^\dagger}\right)\right)S_{A\to C}^\dagger\\
    \overset{\labelcref{eq:pinching_trace}}{=} & S_{A\to C}\left(\frac{1}{|G|}\pov{\id_B}{A}\otimes \sum_g \ketbra{g}{g}_C\langle g^{-1}|\pov{\rho_A}{C}|g^{-1}\rangle_C\right) S_{A\to C}^\dagger  \\
    =&\frac{1}{|G|}\sum_{g,h,h'}|h^{-1}\rangle_A\underset{=\delta_{hg}}{\underbrace{\langle h|g\rangle_C}}\underset{=\delta_{h'g}}{\underbrace{\langle g|h'\rangle_C}}\langle h'^{-1}|_A \; \langle g^{-1}|\pov{\rho_A}{C}|g^{-1}\rangle_C \otimes U_B^\dagger\left(h\right) U_B\left(h'\right)\\
    =& \sum_g\ketbra{g}{g}_A\langle g|\pov{\rho_A}{C}|g\rangle_C\otimes \frac{1}{|G|}\pov{\id_{B}}{C}.
\end{align}

\section{Perspective-Independence of Mutual Information}
\label{app:mutual_info}
In~\cref{def: mutual_info_QRFs}, we claimed that the mutual information across different QRF perspectives is independent of whether we construct the subalgebras and conditional expectations in Charlie's perspective or in Alice's perspective. This is formally proven here.
\begin{proposition}
    The mutual information $\minfo\left(\pov{B}{A}:\pov{B}{C}\right)$ is independent of whether we construct the subalgebras and conditional expectations relative to Alice or Charlie.
\end{proposition}
\begin{proof}
    The relative entropy is invariant under unitaries, i.e.,
    \begin{equation}
        D\left(U\rho U^\dagger, U\sigma U^\dagger\right) = D\left(\rho,\sigma\right)
    \end{equation}
    for any unitary $U$. Since the QRF transformation $S_{C\to A}$ is unitary, for any subalgebra $\mathcal{M}\subseteq\cB\left(\pov{\cH_{AB}}{C}\right)$ with conditional expectation $E$, we have
    \begin{equation}
    \begin{split}
        \info\left(\mathcal{M}\right)_{\pov{\rho}{C}} &= D\left(\cexp(\pov{\rho}{C}),\pov{\pi}{C}\right) = D\left(S_{C\to A}\cexp\left(\pov{\rho}{C}\right)S_{C\to A}^\dagger,S_{C\to A}\pov{\pi}{C} S_{C\to A}^\dagger\right)\\
    &= D\left(S_{C\to A}\cexp\left(S_{C\to A}^\dagger \pov{\rho}{A}S_{C\to A}\right)S_{C\to A}^\dagger, \pov{\pi}{A}\right)= D\left(\mappb{\cexp}(\pov{\rho}{A}), \pov{\pi}{A}\right).
    \end{split}
    \end{equation}
     Notice that $\mappb{\cexp}$ is precisely the conditional expectation onto the subalgebra relative to Alice, expressed through the push-forward $\mathcal{M}_\triangleright =\left\{S_{C\to A}OS_{C\to A}^\dagger|O\in \mathcal{M}\right\}$. For every element $O\in \cB\left(\pov{\cH_{AB}}{C}\right)$, there is an element $O_\triangleright = S_{C\to A}OS_{C\to A}^\dagger\in \cB\left(\pov{\cH_{BC}}{A}\right)$ which gets mapped to
    \begin{equation}
    \begin{split}
        \mappb{\cexp}\left(O_\triangleright\right) &= S_{C\to A}\cexp\left(S_{C\to A}^\dagger O_\triangleright S_{C\to A}\right)S_{C\to A}^\dagger \\
        &=  S_{C\to A}\cexp\left(S_{C\to A}^\dagger S_{C\to A}O S_{C\to A}^\dagger S_{C\to A}\right)S_{C\to A}^\dagger = S_{C\to A} \cexp\left(O\right)S_{C\to A}^\dagger \in \mathcal{M}_\triangleright
    \end{split}
    \end{equation}
    and together they span $\mathcal{M}_\triangleright$.
    Thus, we have 
    \begin{equation}
        \info(\mathcal{\mathcal{M}})_{\pov{\rho}{C}} = D\left(\mappb{\cexp}(\pov{\rho}{A}),\pov{\pi}{A}\right) = \info(\mathcal{M}_\triangleright)_{\pov{\rho}{A}}.
    \end{equation}
    Finally, we can use the fact that, since the QRF transformation is unitary, we have $\left(\mathcal{M} \vee \mathcal{N}\right)_\triangleright = \mathcal{M}_\triangleright \vee \mathcal{N}_\triangleright$, and $\left(\mathcal{M}_\triangleright\right)_\triangleleft = \mathcal{M}$ for all algebras $\mathcal{M}, \mathcal{N} \subseteq \cB\left(\pov{\cH_{AB}}{C}\right)$, to write
    \begin{equation}
    \begin{split}
         \minfo\left(\pov{B}{A}:\pov{B}{C}\right)_\rho &=\info\left(\algpovpb{B}{A}\vee \algpov{B}{C}\right)_{\pov{\rho}{C}} - \info\left(\algpovpb{B}{A}\right)_{\pov{\rho}{C}}-\info\left(\algpov{B}{C}\right)_{\pov{\rho}{C}} + \info\left(\algpov{B}{A}\wedge\algpovpb{B}{C} \right)_{\pov{\rho}{C}}\\
         &=\info\left((\algpovpb{B}{A}\vee \algpov{B}{C})_\triangleright\right)_{\pov{\rho}{A}} - \info\left((\algpovpb{B}{A})_{\triangleright}\right)_{\pov{\rho}{A}}-\info\left((\algpov{B}{C})_\triangleright\right)_{\pov{\rho}{A}} + \info\left((\algpovpb{B}{A}\wedge \algpov{B}{C})_\triangleright \right)_{\pov{\rho}{A}}\\
        &= \info\left(\algpov{B}{A}\vee \algpovpf{B}{C}\right)_{\pov{\rho}{A}} - \info\left(\algpov{B}{A}\right)_{\pov{\rho}{A}} -\info\left(\algpovpf{B}{C}\right)_{\pov{\rho}{A}} + \info\left(\algpov{B}{A}\wedge \algpovpf{B}{C}\right)_{\pov{\rho}{A}}.
    \end{split}
    \end{equation}
    This is precisely the expression we would have obtained, had we done the entire construction by embedding all the subalgebras in Alice's perspective.
\end{proof}

\section{Non-Entangling States}\label{Non_Ent_Sta_App}

In this appendix, we show that for the non-entangling states in \cref{defNES}, the entropy and coherence are preserved. They thus satisfy the condition $\Delta\mathcal{C}\Delta \mathcal{E} = 0$ from \cite{cepollaro2024}.

\begin{theorem}\label{thm:nonentangling}
Let $\ket{\Psi}_{AB}^{(C)} = \ket{\psi}_A^{(C)} \otimes \ket{\phi}_B^{(C)}$ be a product state. Let $\pov{\mathcal{E}}{C}= H(\pov{A}{C})_\rho$ and $\pov{\mathcal{E}}{A}= H(\pov{C}{A})_\rho$ be the entanglement entropies respective to each QRF. Let also $\pov{\mathcal{C}}{C}= H(\pov{A}{C})_{\mathcal{P}_G(\rho)}- H(\pov{A}{C})_\rho$ and $\pov{\mathcal{C}}{A}= H(\pov{C}{A})_{\mathcal{P}_G(\rho)}- H(\pov{C}{A})_\rho$ be the coherence entropies in respect to each QRF, where $\mathcal{P}_G(\cdot)$ is the pinching in the group basis. Then $\Delta \mathcal{C}=\pov{\mathcal{C}}{C}-\pov{\mathcal{C}}{A}$ and $\Delta \mathcal{E}=\pov{\mathcal{E}}{C}-\pov{\mathcal{E}}{A}$ satisfy the condition $\Delta \mathcal{C}\Delta \mathcal{E}=0$ if and only if $\ket{\Psi}_{AB}^{(C)}$ is a non-entangling state.
\end{theorem}

\begin{proof} Let us first prove that if $\pov{\ket{\Psi}_{AB}}{C} =\pov{\ket{\psi}_A}{C} \otimes\pov{\ket{\phi}_B}{C}$ is a non-entangling state, then $\Delta \mathcal{C} \Delta\mathcal{E} = 0$.
We first notice that
\begin{equation}
    \pov{\mathcal{E}}{C}=H\left(\pov{A}{C}\right)_{\ketbra{\psi}{\psi}_{A}^{(C)}}=0
\end{equation}
since $\pov{\ket{\Psi}_{AB}}{C}$ is a pure product state. We apply the QRF transformation $S_{C\to A}$, obtaining
\begin{align*} 
 \pov{\ket{\Psi}_{BC}}{A}=S_{C\to A}\left(\pov{\ket{\psi}_A}{C} \otimes\pov{\ket{\phi}_B}{C}\right) = \sum_{g\in G}  U^{\dagger}_B (g)\ket{\phi}_B  \otimes \braket{g|\psi} \, \ket{g^{-1}}_C. 
 \end{align*}
Since only the elements of $G$ such that $\pov{\braket{g|\psi}_A}{C}\ne0$ will contribute to the sum, we restrict it to the set $S = \left\{ g \in G  \Big|\braket{g|\psi}_A^{(C)} \ne 0\right\}$, obtaining
\begin{align} \label{S_1}
 \pov{\ket{\Psi}_{BC}}{A}=S_{C\to A}( \pov{\ket{\psi}_A}{C} \otimes\pov{\ket{\phi}_B}{C} ) = \sum_{s\in S} U^{\dagger}_B (s)\pov{\ket{\phi}_B}{C} \otimes \pov{\braket{s|\psi}_A}{C} \; \ket{s^{-1}}_C.
 \end{align}
 From the definition of a non-entangling state, $U_B^{\dagger}(s)\ket{\phi}_B = e^{i(\theta(s')-\theta(s))}U_B^{\dagger}(s')\ket{\phi}_B $ for all $ s,s' \in S$. Therefore:
 \begin{align*} 
\sum_{s\in S} U^{\dagger}_B (s)\pov{\ket{\phi}_B }{C}\otimes\pov{\braket{s|\psi}_A}{C} \,\ket{s^{-1}}_C    = U^{\dagger}_B (s')\ket{\phi}_B \otimes \left( \sum_{s\in S} e^{i(\theta(s')-\theta(s))}  \braket{s|\psi}_A \, \ket{s^{-1}}_C  \right)= \pov{\ket{\Psi}_{BC}}{A}.
 \end{align*}
Thus, $\pov{\ket{\Psi}_{BC}}{A}$ is also a pure product state, meaning that $\pov{\mathcal{E}}{A}=0$ and implying $\Delta \mathcal{C} \Delta\mathcal{E} = 0$.\\

Let us now prove the converse, namely that $\Delta \mathcal{C} \Delta\mathcal{E} = 0$ implies that $\pov{\ket{\Psi}_{AB}}{C}$ is non-entangling. From~\cref{Thm_3.8_Luppi}, we have that $\Delta \mathcal{C}= -\Delta\mathcal{E}$, thus $\Delta \mathcal{C} \Delta\mathcal{E} = 0$ implies that $\Delta \mathcal{C} = \Delta\mathcal{E} = 0$. Thus, $\pov{\mathcal{E}}{A}= \pov{\mathcal{E}}{C}=0$, implying that the marginal $\pov{\rho_B}{A}$ of the pure state $\pov{\ket{\Psi}_{BC}}{A}$ in~\cref{S_1} is pure. This is the case if and only if $\pov{\ket{\Psi}_{BC}}{A} = \pov{\ket{\phi'}_B}{A}\otimes \pov{\ket{\psi'}_C}{A}$. From~\cref{S_1}, we have
\begin{equation*}
    \pov{\rho_B}{A} =\Tr_C \left[\ket{\Psi}\bra{\Psi}_{CB}^{(A)} \right]= \sum_{s\in S} |\braket{s|\psi}|^2 \,   U^{\dagger}_B (s)\ketbra{\phi}{\phi}_BU_B (s) = \sum_{s\in S} c_s \ketbra{\phi_s}{\phi_s}_B
\end{equation*}
where we defined $c_s = |\braket{s|\psi}|^2 \in (0,1]$, satisfying $\sum_s c_s =1$, and $\ket{\phi_s} = U^{\dagger}_B (s)\ket{\phi}$ is a normalised state since $U_B(s)$ is unitary. By the argument above, we have
\begin{equation}
    1= \bra{\phi'}\pov{\rho_B}{A}\ket{\phi'} = \sum_{s\in S} c_s |\bra{\phi'}\phi_s\rangle|^2 \leq \sum_s c_s = 1.
\end{equation}
Since $|\bra{\phi'}\phi_s\rangle|^2\in [0,1]$, the inequality is strict if and only if there exists $s$ such that $|\bra{\phi'}\phi_s\rangle|^2 < 1$. Thus, for all $s$ it must hold $|\bra{\phi'}\phi_s\rangle|^2=1$. Since $\ket{\phi'}$ and $\ket{\phi_s}$ are normalised vectors, their inner product is 1 if and only if they are the same vector, up to a phase. Therefore, it holds
\begin{equation}
   \ket{\phi} =  e^{i\theta(s)}\ket{\phi_s} =  e^{i\theta(s)} U^{\dagger}_B (s)\ket{\phi},
\end{equation}
for some phase $\theta(s) \in [0,2\pi)$ dependent on $s\in S$. Hence, $e^{i\theta(s)} U^{\dagger}_B (s)\ket{\phi}=e^{i\theta(s')} U^{\dagger}_B (s')\ket{\phi}$.
\end{proof}

As a remark, this theorem also holds if we use instead the linear entanglement entropy $\mathcal{E}_l(\rho_{AB}) \coloneqq 1 - \Tr[\rho_A^2]$ and the $l^2$-norm of coherence $\mathcal{C}_{l^2}[\rho] \coloneqq \sum_{i \ne j} \left| \bra{i}\rho\ket{j}\right|^2$ (cf.~\cite{cepollaro2024}). This is because, also in this case, the entanglement entropy vanishes if and only if $\rho_A^{(C)}$ is pure.

\end{document}

%% file: packages.tex
\usepackage[english]{babel}
\usepackage[letterpaper,top=2cm,bottom=2cm,left=2.5cm,right=2.5cm,marginparwidth=1.75cm]{geometry}

\usepackage{setspace}
\usepackage{soul}
\usepackage{xcolor}
\sethlcolor{yellow}

\usepackage{amsmath,amssymb,physics,braket,mathtools}
\usepackage{graphicx}
\usepackage{enumitem}
\usepackage{tikz}
\usepackage{tikz-cd}
\usepackage{cancel}
\usetikzlibrary{arrows.meta, positioning}

\usepackage{amsthm}
\usepackage{epsdice}
\usepackage{comment}
\usepackage{csquotes}
\usepackage{thmtools}
\usepackage{subcaption}
\usepackage[]{appendix}
\usepackage[normalem]{ulem}

\newcommand{\showcomments}{true}

\newcommand{\lv}[1]%
{\ifthenelse{\equal{\showcomments}{true}}{{\color{blue}{[LV: #1]}}}{\xspace}}%
\newcommand{\cf}[1]%
{\ifthenelse{\equal{\showcomments}{true}}{{\color{magenta}{[CF: #1]}}}{\xspace}}%
\newcommand{\fg}[1]%
{\ifthenelse{\equal{\showcomments}{true}}{{\textcolor[rgb]{0.1,0.5,0.5}{[FG: #1]}}}{\xspace}}%
\newcommand{\fedit}[1]%
{\ifthenelse{\equal{\showcomments}{true}}{{\textcolor[rgb]{0.1,0.5,0.5}{#1}}}{\xspace}}%
\newcommand{\vk}[1]%
{\ifthenelse{\equal{\showcomments}{true}}{{\color{orange}{[VK: #1]}}}{\xspace}}%
\newcommand{\vktext}[1]%
{\ifthenelse{\equal{\showcomments}{true}}{{\color{orange}{#1}}}{\xspace}}%

\newcommand{\lvtext}[1]%
{\ifthenelse{\equal{\showcomments}{true}}{{\color{blue}{#1}}}{\xspace}}%

\newcommand{\cftext}[1]%
{\ifthenelse{\equal{\showcomments}{true}}{{\color{magenta}{#1}}}{\xspace}}%

\usepackage[colorlinks=true, allcolors=blue]{hyperref}
\usepackage{nameref}
\usepackage[capitalise]{cleveref} 

\usepackage{graphicx}
\usepackage{caption}

\declaretheorem[name=Proposition,
	refname={Proposition,Propositions},
	Refname={Proposition,Propositions},
    numberwithin=section]{proposition}

\declaretheorem[name=Theorem,
	refname={Theorem,Theorems},
	Refname={Theorem,Theorems},
	numberlike=proposition]{theorem}
\declaretheorem[name=Lemma,
	refname={Lemma,Lemmas},
	Refname={Lemma,Lemmas},
	numberlike=proposition]{lemma}

\declaretheorem[name=Definition,
	refname={Definition,Definitions},
	Refname={Definition,Definitions},
	numberlike=proposition, style=definition]{definition}
\declaretheorem[name=Example,
	refname={Example,Examples},
	Refname={Example,Examples},
	numberlike=proposition, style=definition]{example}

\usepackage{mathrsfs}   
\usetikzlibrary{arrows.meta}

%% file: aliases.tex
\newcommand{\id}{\mathbb{I}} 
\newcommand{\inc}{\mathrm{inc}}

\newcommand{\cA}{\mathcal{A}}
\newcommand{\cB}{\mathcal{L}} 
\newcommand{\cH}{\mathcal{H}} 

\newcommand{\cexp}{E}  
\newcommand{\cexparg}[1]{\cexp_{#1}} 

\newcommand{\info}{I} 
\newcommand{\minfo}{I} 

\newcommand{\pov}[2]{#1^{(#2)}} 

\newcommand{\alg}[1]{\mathcal{#1}}
\newcommand{\algpov}[2]{\mathcal{#1}^{(#2)}}

\newcommand{\algpovemb}[2]{\mathcal{#1}_\mathrm{emb}^{(#2)}} 

\newcommand{\algpovpb}[2]{\mathcal{#1}^{(#2)}_{\triangleleft}} 

\newcommand{\algpovpf}[2]{\mathcal{#1}^{(#2)}_{\triangleright}}

\newcommand{\algkin}[1]{\mathcal{#1}_{\textup{kin}}}

\newcommand{\mappb}[1]{{#1}^{\triangleleft}}

\definecolor{orA}{RGB}{233,164,61}
\definecolor{orAlight}{RGB}{244,214,160}
\definecolor{blC}{RGB}{44,88,152}
\definecolor{blClight}{RGB}{168,195,224}
\definecolor{topgray}{RGB}{232,232,232}

\definecolor{marA}{RGB}{122,26,72}
\definecolor{grC}{RGB}{46,120,64}

\usetikzlibrary{arrows.meta,bending}

\definecolor{marA}{RGB}{140,40,88}
\definecolor{blLight}{RGB}{190,207,228}
\definecolor{blDark}{RGB}{115,152,196}
\definecolor{orLight}{RGB}{247,229,196}
\definecolor{orDark}{RGB}{242,199,133}

\definecolor{marA}{RGB}{140,40,88}
\definecolor{blFull}{RGB}{44,98,168}
\definecolor{blLight}{RGB}{190,207,228}
\definecolor{blDark}{RGB}{115,152,196}
\definecolor{orFull}{RGB}{238,170,70}
\definecolor{orLight}{RGB}{247,229,196}
\definecolor{orDark}{RGB}{242,199,133}

\definecolor{myorange}{HTML}{F5A427}
\definecolor{mymagenta}{HTML}{941751}
\definecolor{mygreen}{HTML}{009051}
\definecolor{myblue}{HTML}{2A63A2}
\definecolor{mygray}{HTML}{929292}

\colorlet{myoverlap}{mygray!65}                 
\colorlet{mycirclefill}{mygreen!30!myoverlap}

%% file: fig2.tex
\begin{tikzpicture}[line cap=round, line join=round,scale=0.9, transform shape]

\draw[thick,line width =1pt] (8.00,3.5) -- (8.00,0.5);
\coordinate (top) at (8.00,5.28);
\coordinate (tl)  at (7.19,4.81);
\coordinate (tr)  at (8.81,4.81);
\coordinate (ctr) at (8.00,4.31);
\coordinate (bl)  at (7.19,3.91);
\coordinate (bot) at (8.00,3.41);
\coordinate (br)  at (8.81,3.91);

\coordinate (eyeA) at (4.88,2.19);
\coordinate (eyeC) at (11.25,2.19);

\fill[orAlight] (eyeA) -- (tl) -- (bot) -- cycle;
\fill[blClight] (eyeC) -- (tr) -- (bot) -- cycle;

\fill[topgray]  (top) -- (tl) -- (ctr) -- (tr) -- cycle;
\fill[orAlight] (tl)  -- (ctr) -- (bot) -- (bl) -- cycle;
\fill[blClight] (tr)  -- (ctr) -- (bot) -- (br) -- cycle;
\draw[thick,line width= 1pt] (top) -- (tl) -- (ctr) -- (tr) -- cycle;
\draw[thick,orA!60!black,line width= 1pt] (tl) -- (bl) -- (bot) -- (ctr);
\draw[thick,blC!90!black,line width= 1pt]          (tr) -- (br) -- (bot) -- (ctr);

\begin{scope}[cm={0.808389304633, -0.4990057436005, 0, 0.95,(7.19,3.91)}]
    \node[orA!80!black] at (0.475, 0.44)
        {$\mathrm{\pov{B}{A}}$};
\end{scope}

\begin{scope}[cm={0.808389304633, 0.4990057436005, 0, 0.95,(8,3.41)}]
    \node[blC!95!black] at (0.525, 0.485)
        {$\mathrm{\pov{B}{C}}$};
\end{scope}



\draw[thick,-{Stealth[length=2.5mm]},line width =1pt] (5.05,4.70) to[bend left=38] (11.25,4.70);
\node at (8.05,6.2) {$S_{A\to C}$};


\draw[thick, line width =1pt] (4.25,1.9) circle (0.28);                                      
\draw[thick, line width =1pt] (4.04,2.08) .. controls (3.75,2.28) and (3.60,1.95) .. (3.78,1.72); 
\draw[black, line width=1.5pt] (3.962,2.063) -- (4.034,2.150); 
\draw[thick,fill=orA, line width =1pt] (4.25,1.56) -- (3.81,0.81) -- (4.69,0.81) -- cycle;    
\draw[thick, line width =1pt] (4.05,1.05) -- (4.33,1.33);                                     
\draw[thick, line width =1pt] (4.10,0.81) -- (4.10,0.58)  (4.40,0.81) -- (4.40,0.58);         
\draw[thick, line width =1pt] (4.40,1.40) -- (4.62,1.45);                                     
\draw[thick,orA,-{Stealth[length=2mm]}, line width =1pt] (4.62,1.28) -- (4.80,1.90);          
\node at (5.05,1.42) {$A$};

\begin{scope}[shift= {(0,-0.2)}]
    \draw[thick,line width=1pt] (11.75,2.03) circle (0.28);                                     
\draw[thick,fill=blC,line width=1pt] (11.50,1.69) rectangle (12.03,1.00);                   
\draw[thick,line width=1pt] (11.93,1.45) -- (11.93,1.20);                                   
\draw[thick,line width=1pt] (11.62,1.00) -- (11.62,0.76)  (11.90,1.00) -- (11.90,0.76);     
\draw[thick,line width=1pt] (11.50,1.40) -- (11.32,1.45);                                   
\draw[thick,blC,-{Stealth[length=2mm]},line width=1pt] (11.38,1.30) -- (11.15,1.88);        
\node at (10.88,1.32) {$C$};
\end{scope}

\node[orA,anchor=east] at (5.40,4.10)
  {$\algpov{B}{A}=\langle \pov{X_B}{A}\otimes\id,\,\pov{Z_B}{A}\otimes\id\rangle$};
\node[orA,anchor=west] at (11.10,4.10)
  {$\algpovpb{B}{A}=\langle \id\otimes \pov{X_B}{C},\,\pov{Z_A}{C}\otimes \pov{Z_B}{C}\rangle$};
\node[blC,anchor=west] at (11.30,3.10)
  {$\algpov{B}{C}=\langle \id\otimes \pov{X_B}{C},\,\id\otimes \pov{Z_B}{C}\rangle$};

\node[anchor=west] at (0.00,0.55)
  {$\cH^{(A)}=\cH^{(A)}_B\otimes\cH^{(A)}_C$};
\node[anchor=east] at (16.00,0.55)
  {$\cH^{(C)}=\cH^{(C)}_A\otimes\cH^{(C)}_B$};

\end{tikzpicture}

%% file: fig7a.tex
\begin{tikzpicture}[
  line cap=round, line join=round,
  axis/.style={line width=1pt, -{Stealth[length=3mm,width=2.2mm]}},
  sphere/.style={line width=1.2pt},
  eq/.style={line width=1.1pt, densely dotted},
  vec/.style={marA, line width=2.5pt, -{Stealth[length=3.5mm,width=3mm]}},
  parrow/.style={line width=1.5pt, -{Stealth[length=2.5mm,width=2mm]}},
  font=\LARGE
]
\def\R{1.9}    
\def\ry{0.75}  
\def\rx{0.75}  

\begin{scope}[shift={(0,0)}]
  \draw[axis] (0,-2.65) -- (0,3.0);
  \draw[axis] (-2.3,0) -- (2.9,0);
  \draw[sphere] (0,0) circle (\R);
  \draw[eq] (0,0) ellipse ({\R} and \ry);
  \draw[vec] (0,0) -- (1.0,1.45);
\end{scope}

\begin{scope}[shift={(-5.1,-5.1)}]
  \fill[blLight] (0,0) circle (\R);
  \fill[blDark] (-\R,0) arc[start angle=180,end angle=360,radius=\R]
                arc[start angle=0,end angle=-180,x radius=\R,y radius=\ry] -- cycle;
  \draw[axis] (0,-2.65) -- (0,3.0);
  \draw[axis] (-2.3,0) -- (2.9,0);
  \draw[sphere] (0,0) circle (\R);
  \draw[eq] (0,0) ellipse ({\R} and \ry);
  \draw[vec] (0,0) -- (0,1.55);
\end{scope}

\begin{scope}[shift={(5.1,-5.1)}]
  \fill[orLight] (0,0) circle (\R);
  \fill[orDark] (0,\R) arc[start angle=90,end angle=270,radius=\R]
              arc[start angle=270,end angle=90,x radius=\rx,y radius=\R] -- cycle;
  \draw[axis] (0,-2.65) -- (0,3.0);
  \draw[axis] (-2.3,0) -- (2.9,0);
  \draw[sphere] (0,0) circle (\R);
  \draw[eq] (0,0) ellipse ({\rx} and \R);
  \draw[vec] (0,0) -- (1.2,0);
\end{scope}

\begin{scope}[shift={(0,-10.2)}]
  \fill[blLight] (0,0) circle (\R);
  \fill[blDark] (-\R,0) arc[start angle=180,end angle=360,radius=\R]
                arc[start angle=0,end angle=-180,x radius=\R,y radius=\ry] -- cycle;
  \begin{scope}[transparency group, opacity=0.55]
    \fill[orLight] (0,0) circle (\R);
    \fill[orDark] (0,\R) arc[start angle=90,end angle=270,radius=\R]
              arc[start angle=270,end angle=90,x radius=\rx,y radius=\R] -- cycle;
  \end{scope}
  \draw[axis] (0,-2.65) -- (0,3.0);
  \draw[axis] (-2.3,0) -- (2.9,0);
  \draw[sphere] (0,0) circle (\R);
  \draw[eq, black] (0,0) ellipse ({\R} and \ry);
  \draw[eq] (0,0) ellipse ({\rx} and \R);
  \fill[marA] (0,0) circle (0.14);
\end{scope}

\draw[parrow] (-3.4,-0.1) to[bend right=35] (-4.4,-1.1);
\node[font=\huge] at (-4.9,-0.0) {$\mathcal{P}_Z$};

\draw[parrow] (3.6,-0.1) to[bend left=35] (4.5,-1.1);
\node[font=\huge] at (5.0,-0.0) {$\mathcal{P}_X$};

\draw[parrow] (-4.4,-8.0) to[bend right=35] (-3.6,-9.0);
\node[font=\huge] at (-5.1,-9.0) {$\mathcal{P}_X$};

\draw[parrow] (4.5,-8.0) to[bend left=35] (3.6,-9.0);
\node[font=\huge] at (5.2,-9.0) {$\mathcal{P}_Z$};
\end{tikzpicture}

%% file: fig7b.tex
\begin{tikzpicture}[
  line cap=round, line join=round,
  axis/.style={line width=1pt, -{Stealth[length=3mm,width=2.2mm]}},
  sphere/.style={line width=1.2pt},
  eq/.style={line width=1.1pt, densely dotted},
  vec/.style={marA, line width=2.5pt, -{Stealth[length=3.5mm,width=3mm]}},
  parrow/.style={line width=1pt, -{Stealth[length=2.5mm,width=2mm]}},
  font=\LARGE,
]
\def\R{1.9}     
\def\ry{0.75}   
\def\tilt{-38}  

\newcommand{\equatorsplit}[2]{%
  \fill[#1] (0,0) circle (\R);
  \fill[#2] (-\R,0) arc[start angle=180,end angle=360,radius=\R]
            arc[start angle=0,end angle=-180,x radius=\R,y radius=\ry] -- cycle;}
\newcommand{\tiltedsplit}[2]{%
  \fill[#1] (0,0) circle (\R);
  \begin{scope}[rotate=\tilt]
    \fill[#2] (-\R,0) arc[start angle=180,end angle=360,radius=\R]
              arc[start angle=0,end angle=-180,x radius=\R,y radius=\ry] -- cycle;
  \end{scope}}
\newcommand{\axes}{\draw[axis] (0,-2.65) -- (0,3.0); \draw[axis] (-2.3,0) -- (2.9,0);}

\begin{scope}
  \axes
  \draw[sphere] (0,0) circle (\R);
  \draw[eq] (0,0) ellipse ({\R} and \ry);
  \draw[vec] (0,0) -- (1.0,1.45);
\end{scope}

\begin{scope}[shift={(-5.1,-5.1)}]
  \equatorsplit{blLight}{blDark}
  \axes
  \draw[sphere] (0,0) circle (\R);
  \draw[eq] (0,0) ellipse ({\R} and \ry);
  \draw[vec] (0,0) -- (0,1.55);
\end{scope}

\begin{scope}[shift={(5.1,-5.1)}]
  \tiltedsplit{orLight}{orDark}
  \axes
  \draw[sphere] (0,0) circle (\R);
  \draw[eq, rotate=\tilt] (0,0) ellipse ({\R} and \ry);
  \draw[vec] (0,0) -- (1.25,1.15);
\end{scope}

\begin{scope}[shift={(-3.4,-11)}]
  \equatorsplit{blLight}{blDark}
  \begin{scope}[transparency group, opacity=0.55]
    \tiltedsplit{orLight}{orDark}
  \end{scope}
  \axes
  \draw[sphere] (0,0) circle (\R);
  \draw[eq] (0,0) ellipse ({\R} and \ry);
  \draw[eq, rotate=\tilt] (0,0) ellipse ({\R} and \ry);
  \draw[vec] (0,0) -- (0.9,0.85);
\end{scope}

\begin{scope}[shift={(3.4,-11)}]
  \equatorsplit{blLight}{blDark}
  \begin{scope}[transparency group, opacity=0.55]
    \tiltedsplit{orLight}{orDark}
  \end{scope}
  \axes
  \draw[sphere] (0,0) circle (\R);
  \draw[eq] (0,0) ellipse ({\R} and \ry);
  \draw[eq, black, rotate=\tilt] (0,0) ellipse ({\R} and \ry);
  \draw[vec] (0,0) -- (0,1.15);
\end{scope}

\draw[parrow] (-3.4,-0.1) to[bend right=35] (-4.4,-1.1);
\node[font=\huge] at (-4.9,0) {$\mathcal{P}_Z$};
\draw[parrow] (3.6,-0.1) to[bend left=35] (4.5,-1.1);
\node[font=\huge] at (5.0,0) {$\mathcal{P}_H$};
\draw[parrow] (-5.1,-8.0) to[bend right=35] (-4.5,-9.2);
\node[font=\huge] at (-6,-8.8) {$\mathcal{P}_H$};
\draw[parrow] (5.1,-8.0) to[bend left=35] (4.5,-9.2);
\node[font=\huge] at (6,-8.8) {$\mathcal{P}_Z$};
\node at (0.1,-11) {\scalebox{2}{$\neq$}};
\end{tikzpicture}

%% file: fig9.tex
\begin{tikzpicture}[scale=0.7, line width=1.5pt, line join=round]
 
  \def\rad{0.6}   
 
  \fill[myorange!30, rounded corners=\rad cm] (0.05,0) rectangle (7.5,3.32);
  \fill[myblue!30,   rounded corners=\rad cm] (3.5,0)  rectangle (11.2,3.32);
 
  \fill[myoverlap, rounded corners=\rad cm] (3.5,0) rectangle (7.5,3.32);
 
  \draw[myorange, rounded corners=\rad cm] (0.05,0) rectangle (7.5,3.32);
  \draw[myblue,   rounded corners=\rad cm] (3.5,0)  rectangle (11.2,3.32);
 
  \fill[mycirclefill] (5.57,1.66) circle (1.29);
  \draw[mygreen]      (5.57,1.66) circle (1.29);
 
  \draw[-{Latex[length=3mm]}, black, line width=1.2pt]
        (8.2,3.7) to[out=-50,in=30,looseness=1] (7.1,2.30);

\node[text=myorange] at (1.7,1.66) {\Large $\mathcal{A}\vee\mathcal{C}$};

\node[text=myblue] at (9.3,1.66)
    {\Large $\mathcal{B}\vee\mathcal{C}$};

\node[text=mygreen] at (5.57,1.66)
    {\Large $\mathcal{C}$};

\node[black]    at (5.57,4.15)  {\Large $(\mathcal{A}\vee\mathcal{C})\wedge(\mathcal{B}\vee\mathcal{C})$};

\end{tikzpicture}